\documentclass{amsart}
\usepackage{amssymb}
\usepackage{amsmath}
\usepackage{eurosym}
\usepackage{amsfonts}
\usepackage{graphicx,comment}
\usepackage{graphicx}
\usepackage{booktabs,multirow,graphicx}
\usepackage[nomarkers,nolists, tablesonly]{endfloat}
\usepackage{xcolor}

\newtheorem{theorem}{Theorem}
\theoremstyle{plain}

\newtheorem{corollary}{Corollary}

\newtheorem{definition}{Definition}
\newtheorem{example}{Example}

\newtheorem{lemma}{Lemma}

\newtheorem{proposition}{Proposition}

\numberwithin{equation}{section}

\begin{document}
\title[Testing nonparametric shape restrictions ]{\textbf{TESTING
NONPARAMETRIC SHAPE RESTRICTIONS }}
\author{Tatiana Komarova}
\address{Economics Department\\
London School of Economics\\
Houghton Street\\
London WC2A 2AE\\
U.K.}
\email{t.komarova@lse.ac.uk}
\urladdr{http://www.authortwo.twouniv.edu}
\author{Javier Hidalgo}
\address{Economics Department\\
London School of Economics\\
Houghton Street\\
London WC2A 2AE\\
U.K.}
\email{f.j.hidalgo@lse.ac.uk}
\date{\textbf{7 June 2020}}
\subjclass[2000]{Primary 05C38, 15A15; Secondary 05A15, 15A18}
\keywords{Monotonicity, convexity, U-shape, S-shape, symmetry, quasi-convexity, log-convexity, $r$-convexity, mean convexity. B-splines. Khmaladze's transformation.
Distibution-free-estimation.}

\begin{abstract}
We describe and examine a test for a general class of shape constraints,
such as constraints on the signs of derivatives, U-(S-)shape, symmetry, quasi-convexity, log-convexity, $r$-convexity, 
among others, in a nonparametric framework using partial sums empirical
processes. We show that, after a suitable transformation, its asymptotic
distribution is a functional of the standard Brownian motion, so that
critical values are available. However, due to the possible poor
approximation of the asymptotic critical values to the finite sample ones,
we also describe a valid bootstrap algorithm.
\end{abstract}

\maketitle

\section{\textbf{INTRODUCTION}}

\label{section:intro}

Hypothesis testing is one of the most relevant tasks in empirical work. Such
tests include situations when the null and alternative hypothesis are
assumed to belong to a parametric family of models. In a second type of
tests, known as diagnostic or lack-of-fit tests, only the null hypothesis is
assumed to belong to a parametric family leaving the alternative
nonparametric. The latter type of testing has a distinguished and long
literature starting with the work of Kolmogorov, see Stephens \cite{Stephens}, for testing the probability distribution function and in a
time series context by Grenander and Rosenblatt \cite{GrenanderRos} for
testing the white noise hypothesis. In a regression model context a new
avenue of work started in Stute \cite{Stute97}, and Andrews \cite{Andrews97} with a more econometric emphasis, using partial sums empirical
methodology, see also Stute et al. \cite{Stute_etal98b} or Koul and Stute \cite{KoulStute99}, among others. The methodology has attracted a lot of
attention and it rivals tests based on a direct comparison between a
parametric and a nonparametric fit to the regression model as first examined
in H\"{a}rdle and Mammen \cite{HardleMammen} or Hong and White \cite{HongWhite}. One advantage of tests based on partial sums empirical
methodology, when compared to the approach in \cite{HardleMammen}, is that the former does not require the choice of a bandwidth
parameter for its implementation, see also Nikabadze and Stute \cite{NikabadzeStute} for some additional advantages. However, a possible drawback
is that the asymptotic distribution depends, among other possible features,
on the estimator of the parameters of the model under the null hypothesis in
a nontrivial way, as it was shown in Durbin \cite{Durbin73}, and hence
its implementation requires either bootstrap algorithms, see Stute et al. \cite{Stute_etal98a}, or the so-called Khmaladze's \cite{Khmaladze81}
martingale transformation, see also earlier work and ideas in Brown at al. \cite{BrownDurbinEvans}.

In this paper, though, we are interested in a third type of testing where
neither the null hypothesis nor the alternative has a specific parametric
form. This type of hypothesis testing can be denoted as testing for \emph{%
qualitative }or \emph{shape} restrictions. Many shape properties (including
monotonicity, convexity/concavity, strong convexity, log-convexity) are
widespread in economics and other disciplines. It is also often of interest
to analyse statistical and economic relationships that do not have a
persistent shape pattern on the whole domain but rather switch the patterns
in the domain (for example, \emph{U-shaped} or \emph{S-shaped} relations).

To fix ideas, consider the nonparametric regression model%
\begin{align}
y_{i}& =m\left( x_{i}\right) +u_{i}\text{,}  \label{1_1} \\
E[u_{i}|x_{i}]& =0,  \notag
\end{align}%
where $x$ has bounded support $\mathcal{X=}:\left[ \underline{x},\overline{x}%
\right] $ and $m\left( \cdot \right) $ is smooth. More specific conditions
on the sequences $\left\{ u_{i}\right\} _{i\in \mathbb{Z}}$ and $\left\{
x_{i}\right\} _{i\in \mathbb{Z}}$ will be given in Condition $C1$ in Section %
\ref{section:1step}. Our main aim is testing whether the regression function 
$m\left( x\right) $ possesses shape properties captured by a general null
hypothesis 
\begin{equation}
H_{0}:\text{ \ }m\in \mathcal{M}_{0},  \label{main_null}
\end{equation}%
where the class of interest $\mathcal{M}_{0}$ is a subset of smooth
functions from $\mathcal{X}$ to $\mathbb{R}$. Our requirement on the class $%
\mathcal{M}_{0}$ is given in Condition $C0$ below, which is formulated in a way
to directly relate it to the methodology we employ later to approximate and
estimate $m(\cdot )$.

Before we introduce Condition $C0$, let us introduce some useful notations.
Let $\mathcal{B}_{q,L}$ denote the set of all \emph{B-splines} of degree $q$
with knots that split $\mathcal{X}$ into $L^{\prime }$ equally spaced
intervals.\footnote{%
As it will be clear from our exposition later, the condition that these
intervals are equally spaced is not important and is only imposed for the
simplicity of the exposition. We only need that the system of knots has to
become increasingly dense in $\mathcal{X}$. Also, in Section \ref{section:BPsplines} we shall
motivate the choice of \emph{B-splines } in comparison to other nonparametric estimation
techniques.} A generic element in this set is written as a linear
combination $m_{\mathcal{B}}(x)\equiv \sum_{\ell =1}^{L}\beta _{\ell
}p_{\ell ,L}\left( x\right) $, where $L=L^{\prime }+q$, and $\left\{ p_{\ell
,L}\left( \cdot \right) \right\} $ is the collection of the \emph{B-splines}
base for the chosen system of knots (more details will be given in Section \ref{section:BPsplines}). Any element in $\mathcal{B}_{q,L}$ can be fully
characterised by the vector $(\beta _{1},\ldots ,\beta _{L})\in \mathbb{R}%
^{L}$ and constraints on this vector can be mapped into constraints on the
B-spline. For any set $S_{q,L}\subseteq \mathbb{R}^{L}$ we can, thus, define 
\begin{equation*}
\mathcal{M}_{s_{q,L}}=:\left\{ m_{B}(\cdot )\in \mathcal{B}_{q,L}\;|\;(\beta
_{1},\ldots ,\beta _{L})\in S_{q,L}\right\} .
\end{equation*}

\vskip 0.1in

\noindent \emph{\textbf{Condition C0}. There is a set $S_{q,L} \subseteq  \mathbb{R}^{L} $ for any $L'$ (and, hence, for any $L=L'+q$)  that satisfies the following properties: 
\begin{itemize} 
\item[(a)] $S_{q,L}$ does not depend on data $\left\{ x_{i}\right\} _{i\in \mathbb{Z}}$ and, thus, is non-stochastic; 
\item[(b)] the boundary of $S_{q,L}$ consists of a finite number of smooth surfaces; 
\item[(c)] 
\begin{equation}
\label{GT2}
\mathcal{H}\left( \mathcal{M}_0,\mathcal{M}_{S_{q,L}}\right) \rightarrow
0\quad \text{as}\quad L\rightarrow \infty \text{,}
\end{equation}
where $\mathcal{H}$ is the Hausdorff distance in the supremum norm in the space of continuous functions from $\mathcal{X}$ to $\mathbb{R}$. 
\end{itemize}}

\vskip 0.1in

Condition $C0$ essentially states that the membership in $\mathcal{M}_0$ can
be captured by restrictions on parameters $\{\beta_{\ell}\}_{\ell=1}^L$
which become necessary and sufficient restrictions as the system of knots
becomes increasingly dense in $\mathcal{X}$. In addition, these restrictions
on $\{\beta_{\ell}\}_{\ell=1}^L$ do not depend on the available data, which
adds to the attractiveness of our approach for implementation purposes. 

Our idea will be to test the null hypothesis 
\begin{equation}
H_{0}^{B}:\text{ \ }(\beta _{1},\ldots ,\beta _{L})\in S_{q,L}
\label{main_nullB}
\end{equation}%
formulated in terms of the approximation for $m(\cdot )$. As may be
expected, our methodology readily extends to situations when the domain $%
\mathcal{X}$ is partitioned into several intervals with different
null hypotheses in the spirit of (\ref{main_null}) formulated on
different intervals in the partitioning. An illustration of this situation
is given in Example \ref{example:deriv2} below.

We shall now illustrate several examples of shape classes $\mathcal{M}_{0}$
that satisfy condition $C0$. In Section \ref{section:BPsplines}, after a
more thorough discussion of B-splines we indicate the corresponding sets $%
S_{q,L}$ for each of these examples.  As we shall discuss in Section \ref{section:BPsplines}, the first two examples pertain to the case where in the constraints on the coefficients of the B-splines approximation are linear\footnote{Quasi-convexity in Example \ref{example:deriv2} may be an exception depending on how exactly one approaches the testing, as is clear from further discussion in Section \ref{section:BPsplines}.}, whereas in the last two examples these constraints are nonlinear. It is important to note that these
examples are an illustration of the scope of our approach rather than a full list of properties that our methodology could test.
\label{section:BPsplines}

\begin{example}[Constraints, may hold simultaneously, on the derivatives of $%
m\left( \cdot \right) $]
\label{example:deriv1} 
\begin{equation}
H_{0}:d_{r}\cdot m^{\left( r\right) }\left( x\right) \geq c_{r}\text{,}~\ \
\ \ r\in R\text{,}  \label{eq:example1}
\end{equation}%
where $R$ is a finite subset of $\mathbb{N}^{+}$, $d_{r}\in \left\{
-1,1\right\} $ and $c_{r}$ are known constants.

In this hypothesis, we specify conditions on inequalities for several
derivatives simultaneously. If $c_{r}=0$, then we have familiar restrictions
on the sign of the $r$-th derivative. Special cases in this example include
testing for \textbf{monotonicity} ($r=1$ and $c_{1}=0$), testing for \textbf{%
convexity/concavity} ($r=2$ and $c_{2}=0$), testing for \textbf{strong $%
\lambda $-convexity} ($r=2$ and $c_{2}=\lambda >0$), testing for \textbf{%
monotonicity and concavity simultaneously}, etc. 
%These properties are of main interest in the applied work. 
\end{example}
The properties of monotonicity, convexity/concavity or the ones in terms of higher-order derivatives are so commonplace in economics and other disciplines that the applications are too many to list. For instance, a demand function is expected to be a decreasing
function of the price of a good, whereas a supply function is expected to be
increasing. In single-object auctions, the equilibria analyzed commonly are those
in which buyers play monotone strategy functions and mark-ups are often monotone functions of the bids (see e.g. \cite{Krishna}). In
other economic relationships it is often of importance whether the marginal returns
are increasing or decreasing, which naturally amounts to convexity and concavity,
respectively.

Note that in the context of isotone/monotone regressions the literature goes back 
to Brunk \cite{Brunk} and Wright \cite{Wright}, and in the context of convex regressions to 
Hildreth \cite{Hildreth}. A more detailed coverage of this literature is given in Section \ref{sec:methodology}.

As mentioned above, our methodology can be used when different properties
can be conjectured on different intervals of the partition $\mathcal{X}$.
This is illustrated in the next example.

\begin{example}[changing shape patterns  on
the domain; symmetry; quasi-convexity]
\label{example:deriv2} In this case, we can take a partitioning of $\mathcal{%
X}$ at points $\tau _{0}=\underline{x}<\tau _{1}<\ldots <\tau _{J-1}<\tau
_{J}=\overline{x}$, and formulate the null hypothesis on the signs of
various derivatives on the intervals in this partitioning. The null
hypothesis in this case would be 
\begin{align}
H_{0}& :d_{r_{j}}\cdot m^{\left( r_{j}\right) }\left( x\right) \geq 0,\quad
x\in \left( \tau _{j},\tau _{j+1}\right] ,  \label{eq:example2} \\
& \text{ for some fixed }r_{j}\in \mathbb{N},\quad j=0,\ldots ,J-1\text{,} 
\notag
\end{align}%
$d_{r_{j}}\in \left\{ -1,1\right\} $.\footnote{%
It goes without saying that we can have several constraints on each $[\tau
_{j},\tau _{j+1}],\;j=0,\ldots ,J-1$.} The partitioning points $\tau _{j}$
either have to be known or have to be consistently estimated at a suitable
rate before the testing procedure.

This type of hypotheses include, among others, the important scenarios of $U$%
\emph{-shape}, $S$\emph{-shape}. For example, \emph{U-shape} is the property of a function first decreasing
and then increasing. So, we write the null hypothesis of \emph{U-shape} with
the switch at $s_{0}$ as 
\begin{equation}
H_{0}:\text{ \ }\left\{ 
\begin{array}{c}
m\left( x^{1}\right) >m\left( x^{2}\right) \quad \text{ when }\quad
x^{1}<x^{2}\leq s_{0} \\ 
m\left( x^{1}\right) >m\left( x^{2}\right) \quad \text{ when }\quad
x^{1}>x^{2}\geq s_{0}.%
\end{array}%
\right.  \label{monotone_3}
\end{equation}

The issues associated with testing properties outlined in this example
relate to how functional pieces from different partition intervals are
connected to each other. We can incorporate various scenarios -- e.g., we
can require that different pieces are joint continuously, or smoothly at all
or some partition points.

It goes without saying that we can consider several different partitions and
test simultaneously properties formulated on various intervals of these
partitions. For example, one could consider a function defined on $[0,1]$
and conjecture that $m(\cdot)$ is convex on $(0,s_0$, concave on $(s_0,1]$,
decreasing on $(0,0.3s_0]$, increasing on $(0.3s_0, 0.5s_0+0.5]$, and
decreasing on $(0.5s_0, 1]$.

Examples of \emph{U-shaped} relationships in economics and other disciplines
can be found e.g. in \cite{CalabreseBaldwin}, \cite{Corrao_etal}, \cite{Goldin}, 
\cite{Groes_etal}. \cite{SuttonTrefler}, \cite{Weiman},   and in the
discussions in \cite{Kostyshak}, \cite{LindMelhum}  and \cite{Simonsohn}. Inverse \emph{U-shaped}
relationships include the case of the so-called single-peaked preferences,
which is an important class of preferences in psychology and economics.
Examples of \emph{S-shaped} relationships between unexpected earnings and
stock price in accounting can be found in \cite{DasLev} or \cite{FreemanTse}, among others. \emph{%
S-shaped} growth curves of the adopted population in a large society is a
generally accepted empirical feature of innovation diffusion (see
discussions in \cite{Rogers}  or \cite{Utterback}). Thus, testing for an 
\emph{S-shape} in this case would allow one to conclude whether technology
evolves as one would expect.

Another shape property in this class of shape-changing patterns are the properties of \textbf{%
quasi-convexity} and \textbf{%
quasi-concavity}. Formally, the class $\mathcal{M}_{0}$ in (\ref{main_null}) of quasi-convex functions is
defined as 
$$
\mathcal{M}_{0}=\bigg\{ \phi (\cdot ):\;\forall \,x_{1},x_{2}\in \mathcal{X}%
\; \; \forall \lambda \in \lbrack 0,1] \; \; \phi (\lambda x_{1}+(1-\lambda
)x_{2})\leq \max {\big \{}\phi (x_{1}),\phi (x_{2}){\big \}} \bigg\}.
$$
%\begin{multline*}
%\mathcal{M}_{0}=\bigg\{ \phi (\cdot ):\;\forall \,x_{1},x_{2}\in \mathcal{X}%
%\quad \forall \lambda \in \lbrack 0,1] \\ \phi (\lambda x_{1}+(1-\lambda
%)x_{2})\leq \max {\big \{}\phi (x_{1}),\phi (x_{2}){\big \}} \bigg\}.
%\end{multline*}
Since function $\phi $ is \textbf{quasi-concave} if and only if $-\phi $ is
quasi-convex, we can easily write a description of the class of
quasi-concave functions. A smooth function is quasi-convex if and only if 
it first  decreases up to some point and then increases (giving a special case of monotonicity  when such a switch point is located at one of the boundary points of the interval). 
Since this switch point may not be known a priori, it would have to be estimated. A further discussion of this is given later in Section \ref{section:BPsplines}. 

The concept of quasi-convexity/quasi-concavity is quite important in economics and other
disciplines. For example, it is known that indirect utility functions are
quasi-convex provided that the direct utility function is continuous. Thus
strict quasi-concavity (and continuity) of the production function provides
a sufficient condition for the differentiability of the cost function with
respect to input prices. 

Finally, we mention another property related to this framework, which is the \textbf{symmetry} of a function around some point $s_0$. Without a loss of generality, we can take $s_0$ to be the center of the domain  $[\underline{x}, \overline{x}]$\footnote{For a given $s_0$, one can only test for symmetry only on the interval $[s_0 - \min\{\overline{x}-s_0, s_0-\underline{x}\}, s_0+\min\{\overline{x}-s_0, s_0-\underline{x}\}]$, which is centered around $s_0$.} and write the class of symmetric functions as 
\begin{equation*}
\mathcal{M}_{0}=\left\{ \phi (\cdot ):\;\forall \,x \in [\underline{x},s_0] \quad \phi (s_0-x)=\phi (s_0+x) \right\}. 
\end{equation*}
\end{example}

\begin{example}[$r$-convexity; $\protect\rho $-convexity]
\label{example:rconvex} 
The null hypothesis of $\mathbf{r}$-\textbf{convexity}, $r \neq 0$, is
described by the following $\mathcal{M}_0$; 
\begin{multline*}
\mathcal{M}_0 = \bigg\{ \phi(\cdot): \; \phi(\cdot) \geq 0, \quad \forall \,
x_1, x_2 \in \mathcal{X} \quad \forall \lambda \in [0,1] \\
\phi(\lambda x_1+(1-\lambda )x_2)\leq \log \left(\lambda
e^{r\phi(x_1)}+(1-\lambda) e^{r\phi(x_2)}\right)^{\frac{1}{r}} \bigg\},
\end{multline*}
and for $r=0$ this property is described by 
\begin{multline*}
\mathcal{M}_0 = \bigg\{\phi(\cdot): \; \phi(\cdot) \geq 0, \quad \forall \,
x_1, x_2 \in \mathcal{X} \quad \forall \lambda \in [0,1] \\
\phi(\lambda x_1+(1-\lambda )x_2)\leq \lambda \phi(x_1) + (1-\lambda)
\phi(x_2) \bigg\}.
\end{multline*}
The definition of $\mathbf{r}$-\textbf{concave} functions would reverse the
inequalities. If we consider functions that are twice continuously
differentiable, then the class of $r$-convex functions can be described as
 (see e.g. Avriel \cite{Avriel72}) 
\begin{equation*}
\mathcal{M}_0 = \bigg\{\phi(\cdot): \; \phi(\cdot) \geq 0, \quad \forall \,
x \in \mathcal{X} \quad r \left(\phi^{\prime }(x) \right)^2 + \phi^{{\prime
\prime }}(x) \geq 0 \bigg\}.
\end{equation*}
A concept closely related to $r$-convexity/concavity is what is known as $%
\mathbf{\rho}$-convexity/concavity in the economics literature. Following \cite{CaplinNalebuff91a}, we can define the class of $\rho$-convex functions as
follows: for $\rho \neq 0$, 
\begin{multline*}
\mathcal{M}_0 = \bigg\{ \phi(\cdot): \; \phi(\cdot) > 0, \quad \forall \,
x_1, x_2 \in \mathcal{X} \quad \forall \lambda \in [0,1] \\
\phi(\lambda x_1+(1-\lambda )x_2)\leq \left(\lambda
\phi(x_1)^{\rho}+(1-\lambda) \phi(x_2)^{\rho}\right)^{\frac{1}{\rho}} \bigg\}%
,
\end{multline*}
and for $\rho=0$, 
\begin{multline*}
\mathcal{M}_0 = \bigg\{ \phi(\cdot): \; \phi(\cdot) > 0, \quad \forall \,
x_1, x_2 \in \mathcal{X} \quad \forall \lambda \in [0,1] \\
\phi(\lambda x_1+(1-\lambda )x_2)\leq \phi(x_1)^{\lambda}
\phi(x_2)^{1-\lambda}\bigg\}.
\end{multline*}
The definition of $\mathbf{\rho}$-\textbf{concave} functions would reverse
the inequalities. It is obvious that for any $r=\rho$, a real-valued
function $\phi$ is $r$-convex if and only if $e^{\phi}$ is $\rho$-convex.

It is also clear that $\rho$-convexity gives the standard definition of convexity when $\rho=1$ and gives log convexity for $\rho=0$.

In economics, $\rho $-concavity has been used to
measure the curvature of demand functions, production functions, and
distributions (i.e., density and distribution functions) of individual
characteristics. It has also been employed in the areas of imperfect
competition, auctions and mechanism design and public economics.\footnote{%
See, for instance, \cite{AndersonRenault03}, \cite{CaplinNalebuff91b}, \cite{Cowan}, \cite{Ewerhart},  \cite{Moyes}, among others.} Log-convexity
and log-concavity are of particular interest in economics.
\end{example}

Our final example illustrates further shape properties well developed in the
mathematical literature. This example first requires defining a mean
function. 
\begin{definition}[mean function; Niculescu \cite{Niculescu}]
A function $N: \mathbb{R}^{+} \times \mathbb{R}^{+} \rightarrow 
\mathbb{R}^{+}$ is called a \textit{mean function} if

\begin{itemize}
\item $N(x_1, x_2) =N(x_2, x_1)$

\item $N(x, x) =x$

\item $x_1<N(x_1, x_2) < x_2$ whenever $x_1<x_2$

\item $N(ax_1, ax_2) =aN(x_1, x_2)$ for all $a>0$
\end{itemize}	
\end{definition}

Examples of mean functions include the arithmetic mean (A), the
geometric mean (G), the harmonic mean (H), the logarithmic mean and the
identric mean.

\begin{example}[$MN$-convexity]
\label{example:meanconvex} For any two mean functions $M$ and $N$, the class of $\mathbf{MN}$-\textbf{%
convex} is defined as 
\begin{equation*}
\mathcal{M}_0 = \bigg\{ \phi(\cdot): \; \phi(\cdot) > 0, \quad \forall \,
x_1, x_2 \in \mathcal{X} \quad \phi(M(x_1, x_2)) \leq N(\phi(x_1), \phi(x_2))%
\bigg\}.
\end{equation*}
The definition of $\mathbf{MN}$-\textbf{concavity} would reverse the
inequalities.

When we have different combinations of arithmetic (A), geometric (G) and
harmonic (H) means, we end up with the following special cases of ${MN}$%
-convex functions (see e.g. \cite{Anderson_etal}):

\begin{enumerate}
\item $m$ is AG-convex if and only if $\frac{m^{\prime }(x)}{m(x)}$
is increasing 

\item $m$ is AH-convex  if and only if $\frac{m^{\prime }(x)}{m^2(x)%
}$ is increasing 

\item $m$ is GA-convex if and only if $x m^{\prime }(x)$ is
increasing 

\item $m$ is GG-convex if and only if $\frac{x m^{\prime }(x)}{m(x)%
}$ is increasing 

\item $m$ is GH-convex if and only if $\frac{x m^{\prime }(x)}{%
m^2(x)}$ is increasing 

\item $m$ is HA-convex if and only if $x^2 m^{\prime }(x)$ is
increasing 

\item $m$ is HG-convex if and only if $\frac{x^2 m^{\prime }(x)}{%
m(x)}$ is increasing 

\item $m$ is HH-convex if and only if $\frac{x^{2}m^{\prime}(x)}{%
m^{2}(x)}$ is increasing.
\end{enumerate}
\end{example}

\vskip 0.1in

We want to emphasize again that the examples given in this introduction are
meant to illustrate the scope of applicability of our testing methodology
rather than give an exhaustive list of potential application.

\vskip 0.1in

\subsection{\textbf{LITERATURE REVIEW ON TESTING FOR SHAPES}}

$\left. {}\right. $

There is a range of the literature related to testing for monotonicity or
convexity such as \cite{Bowman_etal}, Hall and
Heckman \cite{HallHeckman}, Ghosal et al. \cite{GhosalSenVV}, Juditsky and Nemirovski \cite{JudNem}, Wang and Meyer \cite{WangMeyer}, 
Chetverikov \cite{Chetverikov}, Schlee \cite{Schlee}, Diack and
Thomas-Agnan \cite{DiackTA}, Dumbgen and Spokoiny \cite{DumbgenSpok}, Abrevaya and Jiang \cite{AbrevayaJiang},
Baraud et al. \cite{Baraud_etal}, Dette et al. \cite{DetteHN}.

Some of the work referenced above (such as \cite{Baraud_etal},
\cite{DumbgenSpok}, \cite{JudNem}) focuses on the
regression function in the ideal Gaussian white noise model, while our
framework does not require such assumptions. In the aforementioned
literature some papers, such as \cite{Baraud_etal}, \cite{DiackTA} or \cite{HallHeckman}, allow the explanatory
variable to take only deterministic values. \cite{AbrevayaJiang}, \cite{Chetverikov} and \cite{GhosalSenVV} treat the
explanatory variable as random but either require its full stochastic
independence with the unobserved regression error (\cite{Chetverikov}, \cite{GhosalSenVV}) or require the distribution of the error to be
symmetric conditional on the explanatory variable, see \cite{AbrevayaJiang}, whereas we require a weaker mean independence of error
condition. Some of the above-mentioned papers, such as \cite{Bowman_etal}, \cite{HallHeckman} or \cite{WangMeyer}, do not give any asymptotic theory. Many of the
approaches are tailored to a specific type of shape and their extensions to
more general shape properties do not appear straightforward, if at all
possible. Moreover, these tests are often targeted to detecting specific
deviations from the null hypothesis. The violations of the null hypothesis,
however, can come from different sources. A parametric test for \emph{%
U-shapes} was suggested in \cite{LindMelhum}, whereas
\cite{Kostyshak}  proposes a non-parametric test based on
critical bandwidth giving sufficient conditions for its consistency. To the
best of our knowledge, there is no formal approach in the literature to
testing \emph{S-shapes} or general shape properties formulated on intervals
in a partitioning.

Also, to the best of our knowledge, this paper is the first to give a
unified framework and a nonparametric methodology for testing (possibly
simultaneously) a wide range of shape/qualitative constraints. Even for
shape properties described in terms of the signs of the derivatives of the
function, this methodology will, in particular, overcome some of the
problems discussed above. But, of course, we want to make it clear that its
applicability goes beyond such properties.

To be more specific, we propose a test based on a transformation, introduced
in Khmaladze \cite{Khmaladze81}, of the partial sums empirical process
similar to that in Stute \cite{Stute97}. Some of the properties of
the test is that it converges to a standard Brownian motion, so that
critical values of standard functionals such as Kolmogorov-Smirnov, Cram\'{e}r
-von-Mises or Anderson-Darling are readily obtained. As a consequence, our
testing procedure has the same asymptotic distribution regardless of the
shape constraints under consideration.

Another feature of our testing procedure is its flexibility as it is able to
test simultaneously for more than one shape constraint -- for instance, one
could test monotonicity and log-convexity simultaneously. Finally, the test is easy to implement as
it requires no more than \textquotedblleft recursive\textquotedblright\
least squares.

The remainder of the paper is organized as follows. Section \ref{sec:methodology} introduces and motivates our nonparametric estimator of $%
m\left( x\right) $ and it compares the methodology against rival
nonparametric estimators. It also describes the form that the sets $S_{q,L}$
in Condition $C0$ take for the examples given in the introduction. In
Section \ref{sec:Khmaladze}, we describe the test and examine its
statistical properties. Because the Monte-Carlo experiment suggests that the
asymptotic critical values are not a good approximation of the finite sample
ones, Section \ref{sec:bootstrap} introduces a bootstrap algorithm. Section \ref{sec:simulations_and_data} presents a
Monte-Carlo experiment and some empirical examples, whereas Section \ref{sec:conclusion}
concludes with a summary and possible extensions of the methodology. 
The proofs are confined to the Appendix A which employs a series of lemmas given
in Appendix B.

%All the proofs, which employ a series of lemmas, are confined to the supplementary document.
%The proofs are confined to the Appendix A which employs a series of lemmas given in Appendix B (see the supplementary document).

\section{\textbf{METHODOLOGY AND REGULARITY CONDITIONS}}

\label{sec:methodology}

Before we present the procedure for testing the hypothesis given in (\ref{main_null}), we introduce and motivate the type of nonparametric estimator
that we shall use to estimate the regression function $m\left( x\right) $
under $H_{0}$ and/or $H_{1}$.

When testing for the null hypothesis of either monotonicity or convexity, several nonparameteric estimators have been considered in the literature.
As mentioned in the introduction, early work on isotone/monotone regressions is Brunk \cite{Brunk} and Wright \cite{Wright}. Friedman and Tibshirani \cite{FT} combines local averaging and isotonic
regression, Mukerjee \cite{Mukerjee}
and Mammen \cite{Mammen91a} propose a two-step approach, which involves smoothing the data by
using kernel regression estimators in the first step, and then deals with
the isotonization of the estimator by projecting it into the space of all
isotonic functions in the second step. Hall and Huang \cite{HallHuang} proposes an alternative method based on tilting estimation which preserves
the optimal properties of the kernel regression estimator. Finally, a
different approach to isotonization is based on rearrangement methods using
a probability integral transformation, see Dette et al. \cite{DetteNP} or Chernozhukov et al. \cite{ChernozhukovFVG}. When the null
hypothesis is that of convexity, Hildreth \cite{Hildreth}  approach is based 
on estimating $m\left( \cdot \right) $ by least squares approach. Consistency, rate of convergence properties and pointwise convergence results for this estimator are established, respectively, in \cite{HansonPledger}, \cite{Mammen91b} and \cite{GroeneboomJW}. The global behaviour of
such an estimator is examined in \cite{GuntuboyinaSen}. Birke and Dette \cite{BirkeDette} examine an estimator based on first obtaining unconstrained
estimate of the derivative of the regression function which is isotonized
and then integrated.

Our main aim is to propose a testing methodology which is not only applicable for a wide range of
shape properties but also able to perform valid statistical inferences. The narrowness of the above mentioned techniques would
make it extremely difficult to make them work in our more general
context. Besides, the implementation of these techniques is not always trivial and/or
they often lack any asymptotic theory useful for the purpose of inference.

A different approach might be based on series estimation using polynomial
and, in particular, \emph{Bernstein} polynomials basis, see the survey by
Chen \cite{ChenHandbook}, for results on sieve estimation in both
nonparametric and semiparametric models. One motivation for the Bernstein
polynomials is due to a well-known result that any continuous function can
be approximated in the sup-norm by a Bernstein polynomial whose coefficients
are the values of the function at uniform knot points. This property makes
them an appealing candidate estimation method in our context. However, base
Bernstein polynomials have an undesirable property of being highly
correlated, making them difficult for our purposes -- in
particular, to obtain a valid Khmaladze's transformation, which plays a key
role in our results. This is discussed in more detail in Section \ref{section:BPsplines}.

Due to many of the aforementioned reasons, in this paper we shall employ 
\emph{B-splines} and/or penalized \emph{B-splines} known as \emph{P-splines}. Our testing will rely on the ability to write our null hypothesis in terms
of the coefficients of the \emph{B-spline} approximation, leading us to
utilizing Condition $C0$ and testing (\ref{main_nullB}). Note that monotone
regression splines (closely related to \emph{B-splines}) were introduced by Ramsay 
\cite{Ramsay} and later extended by Meyer \cite{Meyer} to estimate convex/concave function or functions that are, e.g., both
monotone and convex. A distinct feature in our approach is, first, that we
now take the use of \emph{B-splines} to a different level of generality and,
second, that we derive the asymptotic theory as the number of knots goes to
infinity, and, third, that we allow the number of constraints on the
coefficients of \emph{B-splines} to increase with the number of knots. Both
\cite{Ramsay} and \cite{Meyer} consider the
number of knots, and hence the number of constraints, to be fixed.

%The asymptotics of penalized splines estimators when the number of knots is increasing can be found e.g.in Zhou, Shen and Wolfe (1998), Huang (2003), Li and Ruppert (2008), among others.

\subsection{\textbf{B-SPLINES (or P-SPLINES)}}

\label{section:BPsplines}

$\left. {}\right. $

\emph{B-splines} or \emph{P-splines} are constructed from polynomial pieces
joined at some specific points denoted knots. Their computation is obtained
recursively, see \cite{deBoor}, for any degree of the
polynomial. It is well understood that the choice of the number of knots
determines the trade-off between overfitting when there are too many knots,
and underfitting when there are too few knots. The main difference between 
\emph{B-splines} and \emph{P-splines} is that the latter tend to employ a
large number of knots but to avoid oversmoothing they incorporate a penalty
function based on the second difference of the coefficients of adjacent 
\emph{B-splines}, in contrast to the second derivative employed in
\cite{OSullivan86}, \cite{OSullivan88}, see \cite{EilersMarx}  for details.

The methodology and applications of constrained \emph{B-splines} and \emph{%
P-splines} are discussed by many authors, too many to review here. For more 
detailed discussions, see, among others, the monographs \cite{deBoor}
and \cite{Dierckx} for  B-splines and \cite{EilersMarx}, \cite{BollaertsEM} for P-splines. Some literature on shape-preserving splines (for ordinary
shapes such as monotonicity, convexity, etc.) includes, among others, \cite{LiNaikSwetits}, \cite{MammenTA},  \cite{Meyer} and \cite{Ramsay}.

In general, the \emph{B-spline} basis of degree $q$
\begin{itemize}
\item takes positive values on the domain spanned by $q+2$ adjacent knots,
and is zero otherwise;

\item consists of $q+1$ polynomial pieces each of degree $q$, and the
polynomial pieces join at $q$ inner knots;

\item at the joining points, the $q-1$th derivatives are continuous;

\item except at the boundaries, it overlaps with $2q$ polynomials pieces of
its neighbours;

\item at a given $x$, only $q+1$ \emph{B-splines} are nonzero.
\end{itemize}

Assume that one is interested in approximating the regression function $m\left( x\right) $ in the interval $\left[\underline{x},\overline{x}\right]$. Then we split the interval $\left[\underline{x},\overline{x}\right]$ into $L^{\prime }$ equal length subintervals
with $L^{\prime }+1$ knots\footnote{%
Although it is possible to have nonequidistant subintervals, for simplicity
we consider equally spaced knots. An alternative way to locate the knots is
based on the quantiles of the $x$ distribution.}, where each subinterval
will be covered with $q+1$ \emph{B-splines} of degree $q$. The total number
of knots needed will be $L^{\prime }+2q+1$ (each boundary point $\underline{x}$, $\overline{x}$ is
a knot of multiplicity $q+1$) and the number of \emph{B-splines} is $%
L=L^{\prime }+q$. Thus, $m\left( x\right) $ is approximated by a linear
combination of \emph{B-splines} of degree $q$ with coefficients $(\beta
_{1},\ldots ,\beta _{L})$ as 
\begin{equation}
m_{\mathcal{B}}\left( x;L\right) =\sum_{\ell =1}^{L}\beta _{\ell }p_{\ell
,L}\left( x;q\right) \text{,}  \label{m_bers}
\end{equation}%
and where henceforth we shall denote the knots as $\{z^{k}\}$, $%
k=1,\ldots,L^{\prime }+2q+1$, where $\underline{x}=z^{1}=\ldots =z^{q+1}$ and $%
\overline{x}=z^{L^{\prime }+q+1}=\ldots =z^{L^{\prime }+2q+1}$.

\emph{B-splines} share some properties which turns out to be very useful for
our purpose. Indeed, 
\begin{eqnarray*}
\left( a\right) ~\sum_{\ell =1}^{L}p_{\ell ,L}\left( x;q\right) &=&1\text{ \
\ for all }x\text{ \ and }q\text{.} \\
\left( b\right) \text{ }\frac{\partial }{\partial x}m_{\mathcal{B}}\left(
x;L\right) &=&:m_{\mathcal{B}}^{\prime }\left( x;L\right) =q\sum_{\ell
=1}^{L-1}\frac{\triangle \beta _{\ell +1}}{z^{l+1+q}-z^{l+1}}p_{\ell
+1,L}\left( x;q-1\right) ,
\end{eqnarray*}%
where $\triangle \beta _{\ell }=\beta _{\ell }-\beta _{\ell -1}$. In
particular, $\left( a\right) $ indicates that \emph{B-splines}, as is the
case with \emph{Bernstein} polynomials, are a partition of $1$. The property
(b) states that the derivative of a \emph{B-spline} of degree $q$ becomes a 
\emph{B-spline} of degree $q-1$. Using this expression for the derivative
and taking into account that the knot system effectively changes with the
first and the last knots now removed (thus, the multiplicity of $\underline{x}$ and $\overline{x}$
becomes $q$ rather than $q+1$), one can derive an expression for the second
derivative, and so on.

We now describe the structure that the sets $S_{q,L}$ take in the Examples \ref{example:deriv1}-\ref{example:meanconvex} given in the previous section. One general idea, especially useful when dealing with  shape constraints in Examples \ref{example:rconvex}-\ref{example:meanconvex}, for constructing $S_{q,L}$ is to enforce the shape property of interest at the chosen knots (for simplicity we can take equidistant ones). As $L\rightarrow \infty $, the knot system 
becomes dense in $\mathcal{X}$ and, thus,  the shape property of interest becomes satisfied on an increasingly dense system of points in $\mathcal{X}$.\footnote{On a related note, \cite{Dierckx1980} shows that for cubic splines, if the sign of the second derivative is enforced at the knot points, then the respective convexity/concavity property of the approximation will hold on the whole domain. A similar property can be established for the sign of the first derivative and monotonicity.}

\vskip 0.1in

\noindent \textbf{Example \ref{example:deriv1} (Cont.).} For our general null hypothesis in (\ref{eq:example1}), the set $S_{q,L}$  is given by 
\begin{multline*}
S_{q,L}=\bigg\{(\beta _{1},\ldots ,\beta _{L})\;|\;\forall \;r\in R,\quad
\forall \;z^{k},k=1,\ldots ,L+q+1, \\
d_{r}\cdot m_{\mathcal{B}}\left( z^{k};L\right) \geq c_{r}\bigg\}.
\end{multline*}

When considering special cases of (\ref{eq:example1}), the set(s) $S_{q,L}$ takes much simpler and/or familiar structures. Indeed, consider $R=\{1\}$ and 
$c_1=0$, then the constraints become on the regression function to be
(weakly) monotone and from the property $\left( b\right) $ of B-splines, it
is easy to see that one can take 
\begin{equation*}
S_{q,L} = \left\{(\beta_1, \ldots, \beta_L) \; | \; d_1
(\beta_{\ell+1}-\beta_{\ell})\geq 0, \quad \ell=1, \ldots, L-1 \right\}.
\end{equation*}

When $R=\{r\}$, $r>1$, $c_r=0$, then the conditions imposes on the set $S_{q,L}$ can be more elaborate than those in the monotonicity scenario. However, if the interval $\left[\underline{x},\overline{x}\right]$ is split into equidistance subintervals, the structure of $S_{q,L}$ can be described quite easily as
\begin{equation}
S_{q,L}=\left\{ (\beta _{1},\ldots ,\beta
_{L})\;|\;d_{r}\sum_{k=0}^{r}(-1)^{r-k}\binom{r}{k}\beta _{\ell
+k}\geq 0,\; \ell =q,\ldots ,L-q+1-r\right\} .  \label{eq:sqlderiv}
\end{equation}%
A more refined form of $S_{q,L}$, which may be beneficial for small $L$,
would also involve constraints that capture the behaviour of $%
m_{\mathcal{B}}\left( x;L\right) $ around the boundary in a more precise way. These constraints
are linear inequalities and involve relations among coefficients corresponding to the B-splines around the boundaries.  The form of these inequalities
would, however, be slightly different from the inequalities in (\ref{eq:sqlderiv}). Just to give an example, for $r=2$, $d_{r}=1$ and $c_r=0$
(thus, convexity testing), additional inequalities around the boundary would
have the following form: 
\begin{align*}
(q-1)\triangle \beta _{q+1}& \geq q\triangle \beta _{q},\;\; \ldots ,\;\;\triangle \beta
_{3}\geq 2\triangle \beta _{2}, \\
(q-1)\triangle \beta _{L-q+1}& \leq q\triangle \beta
_{L-q+2},\;\; \ldots ,\;\; \triangle \beta _{L-1}\leq 2\triangle \beta _{L}.
\end{align*}%
As $L\rightarrow \infty $, these additional constraints becomes less and
less important as constraints (\ref{eq:sqlderiv}) essentially capture the
whole convexity property. However, in a finite sample these constraints on
the coefficients around the boundary can be important for the power of the
test. 

\vskip 0.1in

\noindent \textbf{Example \ref{example:deriv2} (Cont.).} In (\ref{eq:example2}), for each
interval $[\tau _{j},\tau _{j+1}]$ in the partition, we have its own system
of knots, so now for this interval we denote the number of unique knots as $%
L_{j}^{\prime }+1$ and the number of base B-splines for that interval as $%
L_{j}=L_{j}^{\prime }+q$. The approximating \emph{B-spline} on that interval
is denoted as $m_{\mathcal{B};j}\left( x;L_{j}\right) =\sum_{\ell
=1}^{L_{j}}\beta _{j;\ell }p_{j;\ell ,L_{j}}\left( x;q\right) $ and the set
of the constraints on coefficients is denoted as $S_{j;q,L_{j}}$. It goes
without saying that we assume that $\tau _{j}$ and $\tau _{j+1}$ are
included into the knot system for the interval as boundary knots with the
multiplicity $q+1$. Denote the complete system of knots on the interval $%
(\tau _{j},\tau _{j+1}]$ as $\left\{ z_{j}^{k}\right\} _{k=1}^{L_{j}+q+1}$.

Define 
\begin{multline*}
S_{j; q,L_j} = \bigg\{(\beta_{j; 1}, \ldots, \beta_{j; L_j}) \; | \; \forall
\; r \in R, \quad \forall \; z_j^k, \quad k=1, \ldots, L_j+ q+1, \\
d_{r_j} \cdot m_{\mathcal{B}; j}\left( z_j^k;L_j\right) \geq 0 \bigg\}
\end{multline*}
and consider this subset of $\mathbb{R}^{L_j}$ to be embedded into $\mathbb{R%
}^{\sum_{j=0}^{J-1} L_j}$, where coefficients from other intervals are taken
into account too. Then without any additional constraints on how pieces on
different partition intervals are joined together, the constraint set in $%
\mathbb{R}^{\sum_{j=0}^{J-1} L_j} $ can be written as $\bigcap_{j=0}^{J-1}
S_{j; q,L_j}$.

If one wants to impose restrictions that the curves on the adjacent subintervals are 
joined continuously, then the constraint set is 
$$S_{cont} = \left(\bigcap_{j=0}^{J-1} S_{j;
q,L_j}\right) \bigcap \left(\bigcap_{j=0}^{J-2} \left\{ \beta_{j; L_j}=
\beta_{j+1; 1}\right\} \right),$$ whereas under the  restriction that the curves are smoothly joined, the constraint set becomes 
\begin{equation*}
S_{smooth} = S_{cont} \bigcap \left( \bigcap_{j=0}^{J-2} \left\{
\beta_{j; L_j}-\beta_{j; L_j-1}= \beta_{j+1; 2}-\beta_{j+1;
1}\right\}\right).
\end{equation*}

To give a specific example, for the null hypothesis of the smooth \emph{%
U-shape} function with switch at the known or estimated $s_{0}$, the
constraint set is 
\begin{multline*}
S_{smooth} = \bigg\{(\beta _{0;1},\ldots ,\beta _{0;L_{0}},\beta _{1;1},\ldots ,\beta
_{1;L_{1}})\;|\;\beta _{0;\ell +1}-\beta _{0;\ell }\leq 0,\;\forall \ell \leq L_{0}-1,\\ \beta _{1;\ell +1}-\beta _{1;\ell }\geq 0, \; 
\forall \leq L_{1}-1,\; \;  \beta _{0;L_{0}}=\beta _{1;1},\; \; \beta
_{0;L_{0}}-\beta _{0;L_{0}-1}=\beta _{1;2}-\beta _{1;1}\bigg\}.
\end{multline*}

For testing quasi-convexity, the set
of constraints is 
\begin{multline*}
S_{q,L}=\bigg\{(\beta _{1},\ldots ,\beta _{L})\;|\;\forall \;z^{k_{1}}\leq
z^{k_{2}}\leq z^{k_{3}},\quad k_{1},k_{2},k_{3}\in \{1,\ldots ,L+q+1\}, \\
m_{\mathcal{B}}\left( z^{k_{2}};L\right) \leq \max \left\{ m_{\mathcal{B}%
}\left( z^{k_{1}};L\right) ,m_{\mathcal{B}}\left( z^{k_{3}};L\right)
\right\} \bigg\}.
\end{multline*}%
By formulating the constraints in this way, we enforce quasi-convexity at the system of knots. Note, however, that quasi-convexity can also be tested by using the above methodology for testing U-shapes but would require estimating the switch point $s_0$. In this case we can make use of the relatively large literature on estimation of the mode or the maximum of a regression function by
nonparametric methods (see \cite{Eddy80}, \cite{Eddy82} or \cite{Muller} for the case when the function is continuous at $s_0$, and \cite{DelgadoHidalgo} 
or \cite{Hidalgo} for the case when the function is discontinuous at $s_0$). 

Finally, for testing symmetry around the interval center $s_0$, given symmetric system of knots on $[\underline{x},s_0]$ and $[s_0, \overline{x}$ and, hence, the same number $L$ of B-splines of degree $q$ on these intervals,   
\begin{multline*}
S_{q,L} = \bigg\{(\beta _{0;1},\ldots ,\beta _{0;L},\beta _{1;1},\ldots ,\beta
_{1;L})\;|\;\beta_{0;\ell}= \beta_{1; L+1-l} \quad \forall l=1, \ldots, L\bigg\}.
\end{multline*}%

\vskip 0.1in 
%\footnote{Here there is a slight abuse of notation as individual sets $S_{j; q,L_j}$ are defined as subsets of $\mathbb{R}^{L_j}$ ignoring coefficients on other intervals. However, when taking the intersection $\cap_{j=1}^J  S_{j; q,L_j}$ we can think of each $\mathbb{R}^{L_j}$ being embedded into $\mathbb{R}^{\sum_{j=1}^J L_j}$.} 

\noindent \textbf{Example \ref{example:rconvex} (Cont.).}
For the sake of brevity, we write the constraints for $r$-convexity only
(it would be straightforward for a reader to write the set of constraints
for $\rho$-convexity as these two notions are related): 
\begin{multline*}
S_{q,L} = \bigg\{ (\beta_1, \ldots, \beta_L) \; | \; \forall \; z^{k}  \quad k \in \{1, \ldots, L+ q+1\}, \\
r \left(m_{\mathcal{B}}'\left( z^{k};L\right)\right)^{2} + m_{\mathcal{B}}''\left( z^{k};L\right) \geq 0; \quad  
\beta_{\ell} > 0, \; \ell=1, \ldots, L \bigg\}.
\end{multline*}

\vskip 0.1in

\noindent \textbf{Example \ref{example:meanconvex} (Cont.).} Convexity in mean for any two given
mean functions $M$ and $N$ can be tested by using the set of constraints 
\begin{multline*}
S_{q,L}=\bigg\{(\beta _{1},\ldots ,\beta _{L})\;|\;\forall
\;z^{k_{1}},z^{k_{2}},\quad k_{1},k_{2}\in \{1,\ldots ,L+q+1\}, \\
m_{\mathcal{B}}\left( M\left( z^{k_{1}},z^{k_{2}}\right) ;L\right) \leq
N\left( m_{\mathcal{B}}\left( z^{k_{1}};L\right) ,m_{\mathcal{B}}\left(
z^{k_{2}};L\right) \right) , \\
\beta _{\ell }>0,\;\ell =1,\ldots ,L\bigg\}.
\end{multline*}%
This is a generic form of constraints for ${MN}$-convexity. Once we choose a
specific form of mean functions $M$ and $N$, it is possible to write these
constraints in a modified way. For illustration purposes, below we present
alternative ways to write the constraints for GA-convex and HG-convex 
functions.

For GA-convexity, we can consider 
\begin{multline*}
S_{q,L} = \bigg\{ (\beta_1, \ldots, \beta_L) \; | \; \forall \; z^{k_1} <
z^{k_2} , \quad k_1,k_2 \in \{1, \ldots, L+ q+1\}, \\
z^{k_1} m^{\prime }_{\mathcal{B}}\left(z^{k_1};L\right) \leq z^{k_2}
m^{\prime }_{\mathcal{B}}\left(z^{k_2};L\right), \quad \beta_{\ell} > 0, \;
\ell=1, \ldots, L\bigg\}.
\end{multline*}

For HG-convexity, we can consider 
\begin{multline*}
S_{q,L} = \bigg\{ (\beta_1, \ldots, \beta_L) \; | \; \forall \; z^{k_1} <
z^{k_2} , \quad k_1,k_2 \in \{1, \ldots, L+ q+1\}, \\
\frac{\left(z^{k_1}\right)^2 m^{\prime }_{\mathcal{B}}\left(z^{k_1};L\right)%
}{m^2_{\mathcal{B}}\left(z^{k_1};L\right)} \leq \frac{\left(z^{k_2}\right)^2
m^{\prime }_{\mathcal{B}}\left(z^{k_2};L\right)}{m^2_{\mathcal{B}%
}\left(z^{k_2};L\right)}, \text{ if }, \quad \beta_{\ell} > 0, \; \ell=1,
\ldots, L \bigg\}. 
\end{multline*}

\vskip 0.1in

Note that in Examples \ref{example:deriv1}-\ref{example:deriv2} the constraints on the coefficients of \emph{%
B-splines} happen to be linear\footnote{Quasi-convexity may be an exception depending on the exact approach.} and, therefore, especially easy to
implement. We want to emphasize that it is a consequence of the type of the
approximation basis that we employ. In other words, it is precisely the
properties of \emph{B-splines} that make the testing of these hypotheses
particularly easy. In Examples \ref{example:rconvex}-\ref{example:meanconvex} the constraints will in general be
non-linear in such coefficients, even though in some special cases -- such
as GA-convexity, HA-convexity, or AA-convexity (which of course results in
the usual convexity) -- the constraints will still happen to be linear.

\vskip 0.1in

\noindent \textbf{Bernstein polynomials and other sieve bases.} We now discuss our motivation to not use other sieves bases in our testing approach. As already mentioned above, a strong candidate for approximating the
regression function subject to qualitative properties is a \emph{Bernstein}
polynomial given by 
\begin{equation*}
m_{\mathcal{B}}\left( x;L\right) =\sum_{\ell =0}^{L}\beta _{\ell }\mathcal{B}%
_{\ell ,L}\left( \frac{x-\underline{x}}{\overline{x}-\underline{x}}\right) 
\text{, \ \ \ }x\in \mathcal{X}\text{,}
\end{equation*}%
where $\mathcal{B}_{\ell ,L}\left( \tilde{x}\right) =:\left( 
\begin{array}{c}
L \\ 
\ell%
\end{array}%
\right) \left( 1-\tilde{x}\right) ^{L-\ell }\tilde{x}^{\ell }$, $\tilde{x}%
\in \lbrack 0,1]$, denotes the $\ell th$ base Bernstein polynomial. What
makes it a strong candidate is the fact that we can take $\beta _{\ell
}=m\left( \underline{x}+\ell /L(\overline{x}-\underline{x})\right) $, and
hence, translate constrains on $m\left( x\right) $ into constraints on
coefficients $\{\beta _{\ell }\}_{\ell =1}^{L}$. However, our main
motivation not to use \emph{Bernstein} polynomials comes from the
observation that, contrary to \emph{B-splines }or \emph{P-splines}, \emph{%
Bernstein} polynomials are highly correlated. Indeed, using results in \cite{LeeParkYoo}, 
the eigenvalues of the matrix $\left\{ E\left( 
\mathcal{B}_{\ell _{1},L}\left( x\right) \mathcal{B}_{\ell _{2},L}\left(
x\right) \right) =a_{\ell _{1},\ell _{2}}\right\} _{\ell _{1},\ell
_{2}=0}^{L}$ are 
\begin{equation*}
\lambda _{k}=\frac{1}{2L+1}\left( 
\begin{array}{c}
2L+1 \\ 
L-k%
\end{array}%
\right) /\left( 
\begin{array}{c}
2L \\ 
L%
\end{array}%
\right) \text{,}
\end{equation*}%
which implies that $\lambda _{L}\leq \left( 2L+1\right) ^{-1}L^{1/2}2^{1-2L}$
since $\left( 
\begin{array}{c}
2L \\ 
L%
\end{array}%
\right) \geq L^{-1/2}2^{2L-1}$. In addition, it is not difficult to see that 
$a_{\ell _{1},\ell _{2}}^{2}/a_{\ell _{1},\ell _{1}}a_{\ell _{2},\ell
_{2}}\rightarrow e^{-1}$ if $\left\vert \ell _{1}-\ell _{2}\right\vert
=o\left( L^{1/2}\right) $, which yields some adverse and important technical
consequences for the test proposed in Section \ref{sec:Khmaladze}.

Finally, other sieve bases could potentially be used but, the implementation
of the test, even for case in Examples 1 and 2 would be laborious. One
popular sieve basis is the power series $1,x,\ldots ,x^{L}$ for any $L$.
However, the formulation of constraints becomes increasingly complicated
when $L$ increases even if one restricts their attention to such cases as in
Examples 1 and 2. For instance, for testing an increasing $m_{\mathcal{B}%
}\left( x,L\right) =:\sum_{\ell =0}^{L}\beta _{\ell }x^{\ell }$, one may require all coefficients in the first derivative polynomial $\sum_{\ell
=1}^{L} l \beta _{\ell }x^{\ell -1 } $ to be non-negative. These conditions
are clearly sufficient but do not become necessary even as $L \rightarrow
\infty$. This could lead to the test rejecting with probability one the null
hypothesis even if it is true, with the main reason for this being the
imposition of constraints that are not true or suitable. A better approach
would be to use the one-to-one relationship between the Bernstein
polynomials basis and the power basis $1,x,\ldots ,x^{L}$ for any $L$. A
similar comment applies if we want to employ the Legendre polynomial base --
as it can be seen from the relationship between the Legendre and Bernstein
polynomials given e.g. in Propositions 1 and
2 in \cite{Farouki}, and this, again, seems an unnecessary step.\footnote{%
The relations between the Bernstein basis and some other polynomial bases
have been addressed in the approximation literature. E.g., \cite{LiZhang} discusses not only the relation between the Bernstein basis and the
Legendre basis but also the relation between the Bernstein basis and
Chebyshev orthogonal bases.}

We shall mention nonetheless that there is no reason to believe that the
results or the methodology introduced in this paper cannot be implemented
using a power series approximation or potentially some other approximation
basis. However, this route seems more arduous than using \emph{B-splines}.
Moreover, in the context of testing shape properties we believe it is more
natural to employ local objects like \emph{B-splines} than the objects more
a global type, like most of the polynomial bases.  

\subsection{\textbf{First step in the testing methodology and regularity
conditions}}

\label{section:1step}

$\left. {}\right. $

In the remainder of the paper, we shall assume, without loss of generality, that  $\mathcal{X}=\left[ 0,1\right]$. 

%This is without loss of generality for a bounded $x\in \left[\underline{x},\overline{% x}\right] $ as we can conduct a simple affine transformation of theregressor to define $x^{\prime }=\left(X-\underline{x}\right) /\left(  \overline{x}-\underline{x}\right) $ to attain the property $x^{\prime }\in \left[ 0,1\right] $ and without affecting any shape property.

We now describe the first step in our methodology of testing $H_{0}$ in (\ref{main_null}) by giving estimators of $m(\cdot )$ under the alternative and
the null hypothesis. To that end, write the \emph{B-splines} as a vector of
functions 
\begin{equation}
\boldsymbol{P}_{L}\left( x\right) =:\left( p_{1,L}\left( x;q\right) ,\ldots
,p_{L,L}\left( x;q\right) \right) ^{\prime }\text{; \ \ }\boldsymbol{P}_{k}=:%
\boldsymbol{P}_{L}\left( x_{k}\right) \text{.}  \label{p_appro}
\end{equation}%
Then, the standard series estimator of $m\left( x\right) $ is defined as the
projection of $y$ onto the space spanned by $\boldsymbol{P}_{L}\left(
x\right) $, that is%
\begin{align}
\widetilde{m}_{\mathcal{B}}\left( x_{i};L\right) & =\widetilde{b}^{\prime }%
\boldsymbol{P}_{i}\text{, }  \label{uncons} \\
\widetilde{b}& =\left( \widetilde{b}_{1},\ldots ,\widetilde{b}_{L}\right)
^{\prime }=\left( \frac{1}{n}\sum_{k=1}^{n}\boldsymbol{P}_{k}\boldsymbol{P}%
_{k}^{\prime }\right) ^{+}\frac{1}{n}\sum_{k=1}^{n}\boldsymbol{P}_{k}y_{k}%
\text{,}  \notag
\end{align}%
where $B^{+}$ denotes the Moore-Penrose inverse of the matrix $B$.

To obtain an estimator under the null hypothesis, we conduct a linear
projection subject to suitable constraints. For general null hypothesis (\ref{main_null}), we consider estimation under the constraints in (\ref{main_nullB}), where $S_{q,L}$ satisfies condition $C0$: 
\begin{equation}
\widehat{b}=\left( \widehat{b}_{1},\ldots ,\widehat{b}_{L}\right) =:\arg
\min _{\substack{ b_{1},\ldots ,b_{L}  \\ \text{s.t. }(b_{1},b_{2},\ldots
,b_{L})\in S_{q,L}}}\sum_{i=1}^{n}\left( y_{i}-\sum_{\ell =1}^{L}b_{\ell
}p_{\ell ,L}\left( x_{i};q\right) \right) ^{2}  \label{est_1}
\end{equation}%
This gives us the estimator of $m\left( \cdot \right) $ under $\left( \ref%
{main_null}\right) $: 
\begin{equation}
\widehat{m}_{\mathcal{B}}\left( x_{i};L\right) =\widehat{b}^{\prime }%
\boldsymbol{P}_{i}\text{.}  \label{m_estimate}
\end{equation}
When the estimation in (\ref{m_estimate})  involves only constraints linear in the parameters, such as in the case of the null hypothesis of an increasing function: 
\begin{equation}
	\label{conest_mon}
\widehat{b}=\left( \widehat{b}_{1},\ldots ,\widehat{b}_{L}\right) =:\arg
\min _{\substack{ b_{1},\ldots ,b_{L}  \\ \text{s.t. } b_{1}\leq b_{2}\leq
\ldots \leq b_{L}}}\sum_{i=1}^{n}\left( y_{i}-\sum_{\ell =1}^{L}b_{\ell
}p_{\ell ,L}\left( x_{i};q\right) \right) ^{2}, 
\end{equation}
then the constrained optimization problem becomes a quadratic programing problem with linear constraints. When the constraints are nonlinear, the constrained estimation may be slightly more challenging to implement but various global optimization techniques can be utilized or even some local optimization techniques as the range of plausible initial values 
of the coefficients of the B-spline approximation can be assessed from the data/model.\footnote{If the unconstrained least squares estimator is in the interior of $S_{q,L}$, then, of course, none of the constraints are binding and the constrained estimation is standard. The computational complications may only happen when some of the constraints are binding.}   
 
To get a better idea of what nonlinear constraints may look like,
take e.g. Example \ref{example:meanconvex}. GA-convexity and HA-convexity
there result in sets $S_{q,L}$ that have linear constraints only, whereas 
AG-, AH-, GG-, GH-, HG-, HH-convexities result in the quadratic
constraints of the form 
\begin{equation}
\beta ^{\prime }Q_{j}\beta \leq 0,\quad j=1,\ldots ,L+1-q.
\label{ex:quadratic_1}
\end{equation}

Indeed, take for example AH-convexity. Under enough smoothness, this
property can be written $\frac{d\,}{d\,x}\left( \frac{m^{\prime }\left(
x\right) }{m^{2}\left( x\right) }\right) \geq 0$, which is
equivalent to $m^{\prime \prime }\left( x\right) \geq \frac{2\left(
m^{\prime }\left( x\right) \right) ^{2}}{m\left( x\right) }$. Taking
into account that $m\left( x\right)$ is known to be positive, as in many economics examples, rewrite the last displayed
inequality as 
\begin{equation*}
m^{\prime \prime }\left( x\right) m\left( x\right) \geq 2\left( m^{\prime
}\left( x\right) \right) ^{2}\text{.}
\end{equation*}%
Since all $m_{\mathcal{B}}(x;L)$, $m_{\mathcal{B}}^{\prime
}(x;L)$, $m_{\mathcal{B}}^{\prime \prime }(x;L)$ are
linear at in $\beta $, then for any $x$, this constraint
is quadratic in $\beta $. When this inequality is enforced at $
L+1-q=L^{\prime }+1$ unique knots,\footnote{Recall an earlier discussion in Section \ref{section:BPsplines} of a general approach to constructing $S_{q,L}$ by enforcing the shape property of interest at the knots, which guarantees that  the shape property becomes satisfied on an increasingly dense system of points in $\mathcal{X}$ as $L \rightarrow \infty$.} it gives us a system of
quadratic inequalities (\ref{ex:quadratic_1}). It is worth noting that due to the nature of the B-splines, the matrices $Q_{j}$ 
are quite sparse. In fact, each $Q_{j}$  contains a non-zero $q\times q$ diagonal block while all the
other elements in $Q_{j}$ are zero. This means that each
inequality in (\ref{ex:quadratic_1}) will effectively contain only $q$ adjacent coefficients in $\beta $. The analogous
conclusion will imply to other properties in Example \ref{example:meanconvex}
which result in non-linear constraints.\footnote{%
There is a rather rich literature on the quadratic optimization subject to
quadratic constraints. A review can be found, e.g., in Park and Boyd (2017).}

\vskip 0.1in

We now introduce our regularity conditions.

\begin{description}
\item[\emph{Condition C1}] $\left\{ (x_{i},u_{i})^{\prime }\right\} _{i\in 
\mathbb{Z}}$ is a sequence of independent and identically distributed random
vectors, where $x_{i}$ has support on $\mathcal{X}=:\left[ 0,1\right] $ and
its probability density function, $f_{X}\left( x\right) $, is bounded away
from zero. In addition, $E[u_{i}|x_{i}]=0$, $E[u_{i}^{2}|x_{i}]=\sigma
_{u}^{2}$, and $u_{i}$ has finite $4th$ moments. \vskip 0.1in

%$\left\{ x_{i}\right\} _{i\in \mathbb{Z}}$ is asequence of independent and identically distributed random variables withsupport on $X=:\left[ 0,1\right] $. $\left\{ u_{i}\right\} _{i\in \mathbb{Z}%} $ is a sequence of zero mean errors mutually independent of $\left\{x_{i}\right\} _{i\in \mathbb{Z}}$ with variance $\sigma _{u}^{2}$ and finite$4th$ moments.

\item[\emph{Condition C2}] $m\left( x\right) $ is three times continuously
differentiable on $\left[ 0,1\right] $. \vskip0.1in

\item[\emph{Condition C3}] As $n\rightarrow \infty $, $L$ satisfies $%
L^{2}/n+n/L^{4}=o\left( 1\right) $.
\end{description}

\vskip 0.1in

Condition $C1$ can be weakened to allow for heteroscedasticity, e.g. $E\left[
u_{i}^{2}\mid x\right] =\sigma _{u}^{2}\left( x\right) $ as it was done in
\cite{Stute97}. However, the latter condition complicates the
technical arguments and for expositional simplicity we omit a detailed
analysis of this case. However, in our empirical applications we present
examples with heteroscedastic errors and illustrate how to deal with it in
practice. Condition $C3$ bounds the rate at which $L$ increases to infinity
with $n$.

Condition $C2$ is a smoothness condition on the regression function $m\left(
x\right) $. It guarantees that the approximation error or bias 
\begin{equation}
m^{bias}\left( x\right) =:m_{\mathcal{B}}\left( x;L\right) -m\left( x\right)
\label{m_bias}
\end{equation}%
is $O(L^{-3})$, see see Agarwal and Studden's \cite{AgarwalStudden} Theorems 3.1
and 4.1 or Zhou et al. \cite{Zhou}. It can be weakened to say that
the second derivatives are H\"{o}lder continuous of degree $\eta >0$. In
that case, $C3$ had to be modified to $L^{2}/n+n/L^{2+2\eta }=o\left(
1\right) $. In case of using \emph{P-splines} we refer to Claeskens \cite{Claeskens_etal} Theorem 2.

\subsection{\textbf{THE TESTING METHODOLOGY}}

$\left. {}\right. $

Having presented our estimators of $m\left( x\right) $ using or not the
constraints induced by the null hypothesis, we now discuss the main part of
the methodology for testing shape restrictions outlined in the introduction.

We shall focus on the null hypothesis $\left( \ref{main_nullB}\right) $ in
terms of the coefficients $\{ \beta _{\ell }\} _{\ell =1}^{L}$, with the
alternative hypothesis being the negation of the null. Since the boundary of 
$S_{q,L}$ is taken to consist of a finite number of smooth surfaces, $%
S_{q,L} $ can be described by a finite number of restrictions. In this way, our testing problem translates  into the more familiar testing scenario when the null hypothesis is given
as a set of constraints on the parameters of the model. However the main and
key difference is that, in our scenario, the number of such constraints
increases with the sample size.

When testing for constraints among the parameters in a regression model, one
possibility is via the (Quasi) Likelihood Ratio principle which compares the
fits obtained by constrained and unconstrained estimates $\widehat{m}_{%
\mathcal{B}}\left( x_{i};L\right) $ and $\widetilde{m}_{\mathcal{B}}\left(
x_{i};L\right) $, respectively. That is, 
\begin{equation*}
\mathcal{LR}_{n}\left( L\right) =\frac{1}{n}\sum_{i=1}^{n}\left( y_{i}-%
\widehat{m}_{\mathcal{B}}\left( x_{i};L\right) \right) ^{2}-\frac{1}{n}%
\sum_{i=1}^{n}\left( y_{i}-\widetilde{m}_{\mathcal{B}}\left( x_{i};L\right)
\right) ^{2}\text{.}
\end{equation*}%
Although results in \cite{ChenChristensenJOE} might be used
to obtain its asymptotic distribution, this route appears difficult to
implement, since there is a high probability that both estimators $\widehat{m%
}_{\mathcal{B}}\left( x_{i};L\right) $ and $\widetilde{m}_{\mathcal{B}%
}\left( x_{i};L\right) $ coincide numerically. The latter is the case when
the constraints are not binding.

A second possibility is to employ the Wald principle, which involves
checking if the constraints in $\left( \ref{main_nullB}\right) $ hold true
for the unconstrained estimator $\widetilde{b}$ of the parameters $\beta
=:\left\{ \beta _{\ell }\right\} _{\ell =1}^{L}$ given in $\left( \ref%
{uncons}\right) $, or, in other words, if the data supports the set of
constraints 
\begin{equation*}
(\widetilde{b}_{1},\widetilde{b}_{2},\ldots ,\widetilde{b}_{L})\in S_{q,L}.
\end{equation*}%
Generally $S_{q,L}$ will involve some inequalities, and even when $L$ is
fixed, the test would be difficult to implement or even to compute the
critical values based on its asymptotic distribution. In our scenario,
however, we would have two potential further technical complications. First,
the number of constraints increases with the sample size $n$, which makes,
from both a theoretical and practical point of view, this route very
arduous, if at all feasible. Second, when some constraints describing $%
S_{q,L}$ are binding, we are then dealing with estimation at the boundary
which implies that the asymptotic distribution cannot be Gaussian.

A third way to test for the null hypothesis $\left( \ref{main_nullB}\right) $
is to implement the Lagrange Multiplier,\emph{\ LM}, test -- that is to
check if the residuals%
\begin{equation}
\widehat{u}_{i}=y_{i}-\widehat{m}_{\mathcal{B}}\left( x_{i};L\right) \text{,}%
\quad i=1,\ldots ,n\text{.}  \label{res_con}
\end{equation}%
and $x_{i}$ satisfy the orthogonality condition imposed by Condition $C1$.
That is, we might base our test on whether or not the set of moment
conditions 
\begin{equation}
\mathcal{K}_{n}\left( x\right) =\frac{1}{n}\sum_{i=1}^{n}\mathcal{I}%
_{i}\left( x\right) \widehat{u}_{i}\text{, \ \ \ \ }x\in \left[ 0,1\right]
\label{T_n}
\end{equation}%
are significantly different from zero, where $\mathcal{I}$ is the indicator function 
and we abbreviate $\mathcal{I}\left( x_{i}<x\right) $ as $\mathcal{I}%
_{i}\left( x\right) $. This approach was described and examined in \cite{Stute97} or \cite{Andrews97} with a more
econometric emphasis. Tests based on $\mathcal{K}_{n}\left( x\right) $ are
known as testing using partial sum empirical processes. Recall that in a
standard regression model the \emph{LM} test is based on the first order
conditions%
\begin{equation*}
\mathcal{LM}_{n}\left( L\right) =\frac{1}{n}\sum_{i=1}^{n}p_{\ell ,L}\left(
x_{i};q\right) \widehat{u}_{i}\text{,}
\end{equation*}%
which has the interpretation of whether the residuals and regressors, $%
p_{\ell ,L}\left( x_{i};q\right) $, satisfies the orthogonality moment
condition induced by Condition $C1$.

However to motivate the reasons to employ a transformation of $\mathcal{K}%
_{n}\left( x\right) $, given in $\left( \ref{k_1}\right) $ or $\left( \ref%
{k_2}\right) $ below, as the basis for our test statistic, it is worth
examining the structure of $\mathcal{K}_{n}\left( x\right) $ given in $%
\left( \ref{T_n}\right) $. For that purpose, we observe that%
\begin{equation}
\mathcal{K}_{n}\left( x\right) =\frac{1}{n}\sum_{i=1}^{n}\mathcal{I}%
_{i}\left( x\right) u_{i}-\sum_{\ell =1}^{L}\left( \widehat{b}_{\ell }-\beta
_{\ell }\right) \mathcal{P}_{n,\ell }\left( x;q\right) +\frac{1}{n}%
\sum_{i=1}^{n}\mathcal{I}_{i}\left( x\right) m^{bias}\left( x_{i}\right) 
\text{,}  \label{T_n1}
\end{equation}%
where $m^{bias}\left( x_{i}\right) $ was given in $\left( \ref{m_bias}%
\right) $ and 
\begin{equation}
\mathcal{P}_{n,\ell }\left( x;q\right) =:\frac{1}{n}\sum_{i=1}^{n}\mathcal{I}%
_{i}\left( x\right) p_{\ell ,L}\left( x_{i};q\right) \text{.}  \label{p_L}
\end{equation}%
Now, the third term on the right of $\left( \ref{T_n1}\right) $ is $O\left(
L^{-3}\right) $ by Agarwal and Studden's \cite{AgarwalStudden} Theorems 3.1
and 4.1, and then Condition $C1$. On the other hand, standard arguments and
Condition $C1$ imply that 
\begin{equation*}
n^{1/2}\mathcal{U}_{n}\left( x\right) =\frac{1}{n^{1/2}}\sum_{i=1}^{n}%
\mathcal{I}_{i}\left( x\right) u_{i}\overset{weakly}{\Rightarrow }\mathcal{U}%
\left( x\right) =:\sigma _{u}\mathcal{B}\left( F_{X}\left( x\right) \right) 
\text{,}
\end{equation*}%
where $\mathcal{B}\left( z\right) $ denotes the standard Brownian motion and 
$F_{X}\left( x\right) $ the distribution function of $x_{i}$.

Next, we discuss the contribution due to $\sum_{\ell =1}^{L}\left( \widehat{b%
}_{\ell }-\beta _{\ell }\right) \mathcal{P}_{n,\ell }\left( x;q\right) $. In
standard lack-of-fit testing problems with $L$ finite, when $\mathcal{K}%
_{n}\left( x\right) =O_{p}\left( n^{-1/2}\right) $, its contribution is
nonnegligible as first showed by \cite{Durbin73} and later by
\cite{Stute97} in a regression model context. However the proof
of Theorem \ref{M_n} suggests that $\sum_{\ell =1}^{L}\left( \widehat{b}%
_{\ell }-\beta _{\ell }\right) \mathcal{P}_{n,\ell }\left( x;q\right)
=O_{p}\left( \left( L/n\right) ^{1/2}\right) $, so that as $L$ increases
with the sample size, it yields that the normalization factor for $\mathcal{K%
}_{n}\left( x\right) $ is of order $n^{\alpha }$ for some $\alpha <1/2$.

Thus the previous arguments suggest that under our conditions, we would have 
\begin{equation*}
\left( \frac{n}{L}\right) ^{1/2}\mathcal{K}_{n}\left( x\right) =:-\left( 
\frac{n}{L}\right) ^{1/2}\sum_{\ell =1}^{L}\left( \widehat{b}_{\ell }-\beta
_{\ell }\right) \mathcal{P}_{n,\ell }\left( x;q\right) \left( 1+o_{p}\left(
1\right) \right) \text{.}
\end{equation*}%
Results by Newey \cite{Newey97}, Lee and Robinson \cite{LeeRobinson}
or Chen and Christensen \cite{ChenChristensenJOE} in a more general context,
might suggest that the left side of the last displayed expression might
converge to a Gaussian process when $\beta $ is in the interior of the set $%
S_{q,L}$. However, when $\beta $ is at the boundary of $S_{q,L}$, then the
asymptotic distribution is not Gaussian, and so to obtain the asymptotic
distribution of $\left( n/L\right) ^{1/2}\sum_{\ell =1}^{L}\left( \widehat{b}%
_{\ell }-\beta _{\ell }\right) \mathcal{P}_{n,\ell }\left( x;q\right) $ for
inference purposes appears quite difficult, if at all possible.

So, the purpose of the next section is to examine a transformation of $%
\mathcal{K}_{n}\left( x\right) $ such that its statistical behaviour will be
free from $\sum_{\ell =1}^{L}\left( \widehat{b}_{\ell }-\beta _{\ell
}\right) \mathcal{P}_{n,\ell }\left( x;q\right) $. The consequence of the
transformation would then be twofold. First, we would obtain that the
transformation of $n^{1/2}\mathcal{K}_{n}\left( x\right) $ is $O_{p}\left(
1\right) $, which leads to better statistical properties of the test, and
secondly and more importantly, the test will be pivotal in the sense that $%
\sigma _{u}^{2}$ becomes the only unknown (although easy to estimate) of its
asymptotic distribution. One consequence of our results is that the
asymptotic distribution becomes independent of the null hypothesis under
consideration.

\section{\textbf{KHMALADZE'S TRANSFORMATION}}

\label{sec:Khmaladze}

This section examines a transformation of $\mathcal{K}_{n}\left( x\right) $
whose asymptotic distribution is free from the statistical behaviour of $%
\left\{ \widehat{b}_{\ell }\right\} _{\ell =1}^{L}$. To that end, we propose
a \textquotedblleft \emph{martingale}\textquotedblright\ transformation
based on ideas by \cite{Khmaladze81}, see also \cite{BrownDurbinEvans} for earlier work. The (linear) transformation, denoted $%
\mathcal{T}$, should satisfy that 
\begin{eqnarray}
\left( \emph{i}\right) \text{ }n^{1/2}\left( \mathcal{TU}_{n}\right) \left(
x\right) \overset{weakly}{\Rightarrow }\mathcal{U}\left( x\right) &=&:\sigma
_{u}\mathcal{B}\left( F_{X}\left( x\right) \right)  \notag \\
\left( \emph{ii}\right) \text{ \ \ \ \ \ \ \ \ \ \ \ \ \ \ \ \ }%
n^{1/2}\left( \mathcal{TP}_{L}\right) \left( x\right) &=&0
\label{properties} \\
\left( \emph{iii}\right) \text{ \ \ \ \ \ \ \ \ \ \ \ \ }n^{1/2}\left( 
\mathcal{T}m^{bias}\right) \left( x\right) &=&o\left( 1\right) \text{,} 
\notag
\end{eqnarray}%
where%
\begin{equation*}
\mathcal{P}_{L}\left( x\right) =:E\left( \mathcal{P}_{n}\left( x;L\right)
\right) =\int_{0}^{x}\boldsymbol{P}_{L}\left( z\right) f_{X}\left( z\right)
dz\text{ \ \ }\left( \mathcal{P}_{n}\left( x;L\right) =:\left\{ \mathcal{P}%
_{n,\ell }\left( x;q\right) \right\} _{\ell =1}^{L}\right) \text{,}
\end{equation*}%
with $\boldsymbol{P}_{L}\left( x\right) $ and $\mathcal{P}_{n,\ell }\left(
x;q\right) $ given respectively in $\left( \ref{p_appro}\right) $ and $%
\left( \ref{p_L}\right) $.

\subsection{\textbf{ALL THE CONSTRAINTS ON $\protect\beta_{\ell}$ ARE LINEAR}%
}

\label{sec:Khmlinear}

$\left. {}\right. $

We first present the form of Khmaladze's transformation when all the
constraints characterizing $S_{q,L}$ are linear and, thus, the surfaces forming the boundary of $%
S_{q,L}$ are hyperplanes.

For any $x<1$, denote 
\begin{equation}
A_{L}\left( x\right) =\int_{x}^{1}\left( \boldsymbol{P}_{L}\left( z\right) 
\boldsymbol{P}_{L}^{\prime }\left( z\right) \right) f_{X}\left( z\right) dz%
\text{.}  \label{L_x}
\end{equation}%
We then define the transformation $\mathcal{T}$ as%
\begin{equation*}
\left( \mathcal{TW}\right) \left( x\right) =\mathcal{W}\left( x\right)
-\int_{0}^{x}\boldsymbol{P}_{L}^{\prime }\left( z\right) A_{L}^{+}\left(
z\right) \left( \int_{z}^{1}\boldsymbol{P}_{L}\left( w\right) \mathcal{W}%
\left( dw\right) \right) f_{X}\left( z\right) dz\text{, \ \ }x<1\text{.}
\end{equation*}%
It is easy to see that the transformation $\mathcal{T}$ satisfies condition $%
\left( \emph{ii}\right) $ in $\left( \ref{properties}\right) $, so that the
main concern will be to show that $\left( \emph{i}\right) $ and $\left( 
\emph{iii}\right) $ hold true. It is important to note that the constraints that proved to be binding in the estimation under the null hypothesis have to be incorporated when conducting the Khmaladze's transformation and, thus, $\boldsymbol{P}_{L}$ in $\left( \ref{p_appro}\right) $ has to be redefined. To give a specific example, consider the monotonicity testing and, thus, the constrained estimation as in (\ref{conest_mon}). Suppose the estimation results in one binding constraint  $\widehat{b}_{\ell _{0}}=\widehat{b}_{\ell _{0}+1} $.  Then, we redefine $P_k$ in $\left( \ref{p_appro}\right)$ as  
\begin{eqnarray*}
\boldsymbol{P}_{k} &=&:\boldsymbol{P}_{L}\left( x_{k}\right), \quad \text{where} \\
\boldsymbol{P}_{L}\left( x\right) &=&:\left( p_{1,L}\left( x;q\right)
,\ldots ,p_{\ell _{0},L}\left( x;q\right) +p_{\ell _{0}+1,L}\left(
x;q\right) ,p_{\ell _{0}+2,L}\left( x;q\right) ,...p_{L,L}\left( x;q\right)
\right) ^{\prime }\text{,}
\end{eqnarray*}%
so that we effectively incorporate the binding constraint. The analogous methodology applies with several binding constraints. 

In what follows then, we shall make no distinction among various cases of binding constraints to
simplify the exposition of the arguments. However, we shall emphasize that the use of the
\textquotedblleft correct\textquotedblright $P_{L}\left( x\right) $ is crucial for the power of the test. Using $P_{L}\left( x\right) $ 
as given in $\left( \ref{p_appro}\right)$ without taking into account the binding the constraints will make the test to have only trivial
power (see the discussion in Section \ref{sec:power}).

The transformation $\mathcal{T}$ has only a theoretical value and
as such, from an inferential point of view, we need to replace it by its
sample analogue, which we shall denote by $\mathcal{T}_{n}$. To that end,
for any $x\in \mathcal{X}$, define 
\begin{eqnarray}
F_{n}\left( x\right) &=&\frac{1}{n}\sum_{k=1}^{n}\mathcal{I}_{k}\left(
x\right)  \label{notation} \\
\text{\ }C_{n}\left( x\right) &=&\frac{1}{n}\sum_{k=1}^{n}\boldsymbol{P}%
_{k}u_{k}\mathcal{J}_{k}\left( x\right) ;~\ \ \ A_{n}\left( x\right) =\frac{1%
}{n}\sum_{k=1}^{n}\boldsymbol{P}_{k}\boldsymbol{P}_{k}^{\prime }\mathcal{J}%
_{k}\left( x\right)  \notag
\end{eqnarray}%
where $\mathcal{I}\left( x\leq x_{k}\right) =:\mathcal{J}_{k}\left( x\right)
=1-\mathcal{I}_{k}\left( x\right) $.

In what follows,  we shall abbreviate 
\begin{equation}
C_{n,i}=:C_{n}\left( \widetilde{x}_{i}\right) ;\text{ \ }A_{n,i}=:A_{n}%
\left( \widetilde{x}_{i}\right) \text{,}  \label{A_n_i}
\end{equation}%
where $\widetilde{x}_{i}=x_{i}$ if $x_{i}+n^{-\varsigma }<z^{k\left(
x_{i}\right) }$ and $=z^{k\left( x_{i}\right) }$ otherwise, with $z^{k\left(
x\right) }$ denoting the closest knot $z^{k}$, $k=1,\ldots ,L$, bigger than $%
x$ and $1/2<\varsigma <1$. The motivation to make this \textquotedblleft
trimming\textquotedblright\ is because when $x_{i}$ is too close to $%
z^{k\left( x_{i}\right) }$, the \emph{B-spline} is close but not equal to
zero, which induces some technical complications in the proof of our main
results. However, in small samples this\emph{\ }\textquotedblleft
trimming\textquotedblright\ does not appear to be needed, becoming a purely
technical argument.

We define the sample analogue of $\left( \mathcal{TW}\right) \left( x\right) 
$ as 
\begin{equation}
\left( \mathcal{T}_{n}\mathcal{W}\right) \left( x\right) =\mathcal{W}\left(
x\right) -\frac{1}{n}\sum_{i=1}^{n}\boldsymbol{P}_{i}^{\prime
}A_{n,i}^{+}\int_{\widetilde{x}_{i}}^{1}\boldsymbol{P}_{L}\left( w\right) 
\mathcal{W}\left( dw\right) \mathcal{I}_{i}\left( x\right) \text{.}
\label{Khma_n}
\end{equation}

The transformation in $\left( \ref{Khma_n}\right) $ has a rather simple
motivation. Suppose that we have ordered the observations according to $%
x_{i} $, that is $x_{i-1}\leq x_{i}$, $i=2,...,n$, which would not affect
the statistical behaviour of $\mathcal{K}_{n}\left( x\right) $. The latter
follows by the well known argument that%
\begin{equation}
\sum_{i=1}^{n}g\left( x_{i}\right) =\sum_{i=1}^{n}g\left( x_{\left( i\right)
}\right) \text{,}  \label{g_i}
\end{equation}%
where $x_{\left( i\right) }$ is the $i-th$ order statistic of $\left\{
x_{i}\right\} _{i=1}^{n}$. So, we have that 
\begin{equation}
\mathcal{K}_{n}\left( x_{j}\right) =\frac{1}{n}\sum_{i=1}^{n}\widehat{u}_{i}%
\mathcal{I}_{i}\left( x_{j}\right) =:\frac{1}{n}\sum_{i=1}^{j}\widehat{u}_{i}%
\text{.}  \label{cusum}
\end{equation}%
Now, $\widehat{u}_{i}=u_{i}-\boldsymbol{P}_{i}^{\prime }A_{n}^{+}\left(
0\right) C_{n}\left( 0\right) $, where 
\begin{eqnarray*}
C_{n}\left( 0\right) &=&\frac{1}{n}\sum_{k=1}^{n}\boldsymbol{P}_{k}u_{k}%
\mathcal{I}_{k}\left( x_{i}\right) +\frac{1}{n}\sum_{k=1}^{n}\boldsymbol{P}%
_{k}u_{k}\mathcal{J}_{k}\left( x_{i}\right) \\
&=&\frac{1}{n}\sum_{k=1}^{i-1}\boldsymbol{P}_{k}u_{k}+\frac{1}{n}%
\sum_{k=i}^{n}\boldsymbol{P}_{k}u_{k}\text{.}
\end{eqnarray*}%
So, if instead of $C_{n}\left( 0\right) $ and $A_{n}\left( 0\right) $, we
employed $C_{n,i}=n^{-1}\sum_{k=i}^{n}\boldsymbol{P}_{k}u_{k}$ and $%
A_{n,i}=n^{-1}\sum_{k=i}^{n}\boldsymbol{P}_{k}\boldsymbol{P}_{k}^{\prime }$,
we would replaced $\widehat{u}_{i}$ by 
\begin{equation}
v_{i}=u_{i}-\boldsymbol{P}_{i}^{\prime }A_{n,i}^{+}C_{n,i}  \label{u_rec}
\end{equation}%
in $\left( \ref{cusum}\right) $, so that it has a martingale difference
structure as $E\left[ v_{i}\mid past\right] =0$, in comparison with $%
\widehat{u}_{i}$ where $E\left[ \widehat{u}_{i}\mid past\right] \neq 0$.
This is the idea behind the so-called (recursive) \emph{Cusum} statistic
first examined in \cite{BrownDurbinEvans} and
developed and examined in length in \cite{Khmaladze81}. Observe that $\left( %
\ref{u_rec}\right) $ becomes the \textquotedblleft prediction
error\textquotedblright of $u_{i}$ when we use the \textquotedblleft
last\textquotedblright\ $j=i,\ldots,n$ observations.

%{\Large Finally, it is worth mentioned that due to thefact that our functionals given in }$\left( \ref{KSAD}\right) ${\Large \ anthat neither }$K_{n}\left( x_{j}\right) ${\Large \ nor its transformationgiven in }$\left( \ref{k_1}\right) ${\Large \ is affected by reordering theobservations according to the rank of }$x_{i}${\Large , we shall forsimplicity to assume that we have implemented this ordering. That is }$x_{1}$%{\Large \ becomes the smaller value of the observed regressor and so on.(Leave this out based on what I said before and I repeat below?)}

Thus, the preceding argument yields the Khmaladze's transformation 
\begin{equation}
\left( \mathcal{T}_{n}\mathcal{K}_{n}\right) \left( x\right) =\mathcal{M}%
_{n}\left( x\right) =:\frac{1}{n^{1/2}}\sum_{i=1}^{n}v_{i}\mathcal{I}%
_{i}\left( x\right) \text{.}  \label{k_1}
\end{equation}%
Observe that, using $\left( \ref{g_i}\right) $, we could write $\mathcal{M}%
_{n}\left( x\right) $ as 
\begin{equation*}
\mathcal{M}_{n}\left( x\right) =:\frac{1}{n^{1/2}}\sum_{i=1}^{n}v_{\left(
i\right) }\mathcal{I}_{\left( i\right) }\left( x\right) \text{.}
\end{equation*}%
Now, Lemma \ref{bias} in Appendix B implies that the condition $\left( \emph{ii}\right) $ in $%
\left( \ref{properties}\right) $ holds true when 
\begin{equation*}
\mathcal{W}_{n}\left( x\right) =\frac{1}{n}\sum_{k=1}^{n}\boldsymbol{P}_{k}%
\mathcal{I}_{k}\left( x\right) \text{,}
\end{equation*}%
so the technical problem is to show that (asymptotically) conditions $\left( 
\emph{i}\right) $ and $\left( \emph{iii}\right) $ in $\left( \ref{properties}%
\right) $ also hold true. That is, to show that 
\begin{eqnarray*}
\left( 1\right) &&\mathcal{M}_{n}\left( x\right) \overset{weakly}{%
\Rightarrow }\mathcal{U}\left( x\right) \text{ \ \ }x\in \left[ 0,1\right] \\
\left( 2\right) &&\left( \mathcal{T}_{n}m^{bias}\right) \left( x\right)
=o_{p}\left( n^{-1/2}\right) \text{.}
\end{eqnarray*}

Finally, it is worth mentioning that in $\left( \ref{Khma_n}\right) $ we
might have employed $\mathcal{J}_{k}\left( x\right) =\mathcal{I}\left(
x<x_{k}\right) $ instead of our definition $\mathcal{J}_{k}\left( x\right) =%
\mathcal{I}\left( x\leq x_{k}\right) $. However because by definition of 
\emph{B-splines}, the matrix $A_{n,i}$, and hence $A_{L}\left( x_{i}\right) $%
, might be singular, if we employed $\mathcal{J}_{k}\left( x\right) =%
\mathcal{I}\left( x<x_{k}\right) $, then it would not be guaranteed that%
\begin{equation*}
\boldsymbol{P}_{i}^{\prime }-\boldsymbol{P}_{i}^{\prime }A_{n,i}^{+}A_{n,i}=0%
\text{.}
\end{equation*}%
On the other hand, Theorem 12.3.4 in \cite{Harville} yields
that the last displayed equation holds true when $\mathcal{J}_{k}\left(
x\right) =\mathcal{I}\left( x\leq x_{k}\right) $. Now, $\left( 1\right) $
will be shown in the next theorem, whereas $\left( 2\right) $ is shown in
Theorem \ref{M_nest}.

\begin{theorem}
\label{M_n}Under Conditions $C1-C3$, we have that 
\begin{equation*}
\mathcal{M}_{n}\left( x\right) \overset{weakly}{\Rightarrow }\mathcal{U}%
\left( x\right) ;\text{ \ \ }x\in \left[ 0,1\right] \text{.}
\end{equation*}
\end{theorem}

Unfortunately, we do not observe $u_{i}$, so that to implement the
transformation we replace $v_{i}$ by $\widehat{v}_{i}$, where $\widehat{v}%
_{i}$ is defined as $v_{i}$ in $\left( \ref{u_rec}\right) $ but where we
replace $u_{i}$ by $\widehat{u}_{i}$ as defined in $\left( \ref{res_con}%
\right) $, yielding the statistic%
\begin{equation}
\widetilde{\mathcal{M}}_{n}\left( x\right) =:\frac{1}{n^{1/2}}\sum_{i=1}^{n}%
\widehat{v}_{i}\mathcal{I}_{i}\left( x\right) \text{.}  \label{k_2}
\end{equation}

\begin{theorem}
\label{M_nest} Assuming that $H_{0}$ holds true, under Conditions $C1-C3$,
we have that 
\begin{equation*}
\widetilde{\mathcal{M}}_{n}\left( x\right) \overset{weakly}{\Rightarrow }%
\mathcal{U}\left( x\right) \text{.}
\end{equation*}
\end{theorem}

Denote the estimator of the variance of $u_{i}$, $\sigma _{u}^{2}$, by 
\begin{equation*}
\widehat{\sigma }_{u}^{2}=\frac{1}{n}\sum_{i=1}^{n}\widehat{u}_{i}^{2}\text{.%
}
\end{equation*}

\begin{proposition}
\label{sigma_est}Under Conditions $C1-C3$, we have that $\widehat{\sigma }%
_{u}^{2}\overset{P}{\rightarrow }\sigma _{u}^{2}$.
\end{proposition}

We then have the following corollary.

\begin{corollary}
Under $H_{0}$ and assuming Conditions $C1-C3$, for any continuous functional 
$g:\mathbb{R\rightarrow R}^{+}$, 
\begin{equation*}
g\left( \widetilde{\mathcal{M}}_{n}\left( x\right) /\widehat{\sigma }%
_{u}\right) \overset{d}{\rightarrow }g\left( \mathcal{U}\left( x\right)
/\sigma _{u}\right) \text{.}
\end{equation*}
\end{corollary}

\begin{proof}
The proof is standard using Theorem $\ref{M_nest}$, Proposition \ref{sigma_est} and the continuous mapping theorem, so it is omitted.
\end{proof}

Denoting $\tilde{n}=:n-L-2$ and $\widetilde{\mathcal{M}}_{n}\left(
x^{q}\right) =\widetilde{\mathcal{M}}_{n,q}$, where $x^{q}=q/n$, standard
functionals are the Kolmogorov-Smirnov, Cram\'{e}r-von-Mises and
Anderson-Darling tests defined respectively as{\ 
\begin{eqnarray}
\mathcal{KS}_{n} &\mathcal{=}&\sup_{q=1,...,\tilde{n}}\left\vert \frac{%
\widetilde{\mathcal{M}}_{n,q}}{\widehat{\sigma }_{u}}\right\vert \overset{d}{%
\rightarrow }\sup_{x\in \left( 0,1\right) }\left\vert \mathcal{B}\left(
F_{X}\left( x\right) \right) \right\vert  \notag \\
\mathcal{C}v\mathcal{M}_{n} &\mathcal{=}&\frac{1}{\tilde{n}}\sum_{q=1}^{%
\tilde{n}}\frac{\widetilde{\mathcal{M}}_{n,q}^{2}}{\widehat{\sigma }_{u}^{2}}%
\overset{d}{\rightarrow }\int_{0}^{1}\mathcal{B}^{2}\left( F_{X}\left(
x\right) \right) dx\text{,}  \label{KSAD}
\end{eqnarray}%
}%
\begin{equation*}
\mathcal{AD}_{n}\mathcal{=}\frac{1}{\tilde{n}}\sum_{q=1}^{\tilde{n}}\frac{%
\widetilde{\mathcal{M}}_{n,q}^{2}}{\widehat{\sigma }_{u}^{2}x^{q}\left(
1-x^{q}\right) }\overset{d}{\rightarrow }\int_{0}^{1}\frac{\mathcal{B}%
^{2}\left( F_{X}\left( x\right) \right) }{F_{X}\left( x\right) \left(
1-F_{X}\left( x\right) \right) }dx\text{.}
\end{equation*}

\vskip 0.1in %$\left. {}\right. $

\subsection{\textbf{NONLINEAR CONSTRAINTS ON $\protect\beta_{\ell}$}}

$\left. {}\right. $

We turn now our attention to describing the Khmaladze's transformation in situations
when some constraints describing $S_{q,L}$ may be non-linear, as it happens to be in our Examples 3 and 4. If the
constrained estimate $\widehat{b}$ lies in the interior of $S_{q,L}$ (and, thus, coincides with the unconstrained estimate), then
the transformation is conducted in the same way as in Section \ref{sec:Khmlinear}. However, the transformation will have a modified form if $%
\widehat{b}$ lies on the boundary of $S_{q,L}$. We give its form and, in particular, discuss what objects plays the role of $\boldsymbol{P}_{L}\left( x\right)$ described in the previous section.

For expositional simplicity suppose that $\widehat{b}$ belongs only to one
of the smooth surfaces describing the boundary of $S_{q,L}$. Usually these
surfaces will be defined by implicit functions but, applying the implicit
function theorem, the explicit representation of this surface can be obtained either analytically or numerically
(even if local, which would suffice since $\widehat{b}$ is consistent) 
with respect to one parameter expressed as a function of
other parameters: suppose that for some $\ell _{0}$ we can express it as $%
\beta_{\ell _{0}}=h(\beta _{1},\ldots ,\beta _{\ell _{0}-1},\beta _{\ell
_{0}+1},\ldots ,\beta _{L})$. In cases such as AG-, AH-convexity and other
similar properties in Example \ref{example:meanconvex}, for the reasons discussed in Section \ref{section:1step} this would be a
restriction $\beta_{\ell _{0}}=h(\beta _{\ell _{0}-2},\beta _{\ell _{0}-1})$
obtained from an implicit function which is a polynomial of degree 2. Then for the
purpose of conducting the Khmaladze's transformation, instead of approximating 
$m(\cdot )$ by the linear function $\sum_{k=0}^{\ell
-1}\beta _{k}p_{k}\left( x_{i}\right) $, we consider the approximation given by 
\begin{equation}
g\left( x_{i};\beta_{-\ell _{0}}\right) =:\sum_{k=0}^{\ell _{0}-1}\beta
_{k}p_{k}\left( x_{i}\right) +h(\beta_{-\ell _{0}})p_{\ell _{0}}\left(
x_{i}\right) +\sum_{k=\ell _{0}+1}^{L}\beta _{k}p_{k}\left( x_{i}\right) ,
\label{2}
\end{equation}%
where $\beta_{-\ell _{0}}=(\beta _{1},\ldots ,\beta _{\ell _{0}-1},\beta
_{\ell _{0}+1},\ldots ,\beta _{L})$. Although this approximation function is nonlinear in
parameters,  it has a simple and useful structure, as it will
become evident from our analysis below.

To define the Khmaladze's transformation, denote the vector of first derivatives
of $g\left( x;\beta_{-\ell _{0}}\right) $ with respect to the parameters as 
\begin{eqnarray*}
\widetilde{P}\left( x;\beta_{-\ell _{0}}\right) &=&:\frac{\partial }{%
\partial \beta_{-\ell _{0}}}g\left( x;\beta_{-\ell _{0}}\right) \\
&=&:\left\{ \widetilde{p}_{\ell }\left( x;\beta_{-\ell _{0}}\right)
\right\} _{\ell =1}^{L}\text{,}
\end{eqnarray*}%
where 
\begin{equation*}
\widetilde{p}_{\ell }\left( x;\beta_{-\ell _{0}}\right) =:p_{\ell }\left(
x\right) +\frac{\partial h\left( \beta_{-\ell _{0}}\right) }{\partial \beta
_{\ell }}p_{\ell _{0}}\left( x\right) ,\quad \ell \neq \ell _{0}.
\end{equation*}

Then, the Khmaladze's transformation of the test statistic $\mathcal{K}%
_{n}\left( x\right) $ in (\ref{T_n}) has the following form: 
\begin{equation*}
\mathcal{M}_{n}\left( x\right) =:\frac{1}{n^{1/2}}\sum_{i=1}^{n}v_{i}%
\mathcal{I}_{i}\left( x\right) ,
\end{equation*}%
where 
\begin{equation*}
v_{i}=u_{i}-\widetilde{P}_{i}^{\prime }\left( \widehat{b}_{-\ell
_{0}}\right) \mathcal{D}_{n}^{+}\left( \widehat{b}_{-\ell
_{0}};i\right) \sum_{k=1}^{n}\widetilde{P}_{k}\left( \widehat{b}_{-\ell
_{0}}\right) {u}_{k}\mathcal{J}_{k}\left( \widetilde{x}_{i}\right),
\end{equation*}
$\widetilde{P}_{i}\left( \beta_{-\ell _{0}}\right) =:\widetilde{P}%
\left( x_{i};\beta_{-\ell _{0}}\right) $, and $\mathcal{D}_{n}\left(x; \beta
_{-\ell _{0}}\right) =\sum_{k=1}^{n}\widetilde{P}_{k}\left( \beta _{-\ell
_{0}}\right) \widetilde{P}_{k}^{\prime }\left( \beta_{-\ell _{0}}\right) 
\mathcal{J}_{k}\left(x\right) $, and $\mathcal{D}^{+}_{n}\left(x; \beta
_{-\ell _{0}}\right)$ is the Moore-Penrose inverse of $\mathcal{D}_{n}\left(x; \beta
_{-\ell _{0}}\right)$, and  $\mathcal{D}_{n}\left( \beta
_{-\ell _{0}};i\right) =\mathcal{D}_{n}\left(\widetilde{x}_{i}; \beta
_{-\ell _{0}}\right) $ with $\widetilde{x}_{i}$ defined 
in the same way as in Section \ref{sec:Khmlinear}. Note that by employing $\widetilde{p}_{\ell }\left( x_i;\beta_{-\ell _{0}}\right)$ instead of $p_{\ell }\left(
x_i\right)$, we have automatically incorporated in the transformation our binding restriction. 

In practice, instead of $v_{i}$ we use 
\begin{equation*}
\widehat{v}_{i}=\widehat{u}_{i}-\widetilde{P}_{i}^{\prime }\left( \widehat{b}_{-\ell _{0}}\right) \mathcal{D}_{n}^{+}\left( \widehat{b}
_{-\ell _{0}};i\right) \sum_{k=1}^{n}\widetilde{P}_{k}\left( \widehat{b}
_{-\ell _{0}}\right) \widehat{u}_{k}\mathcal{J}_{k}\left( \widetilde{x}%
_{i}\right)
\end{equation*}%
and consider a feasible version of the Khmaladze transformation: 
\begin{equation*}
\widetilde{\mathcal{M}}_{n}\left( x\right) =:\frac{1}{n^{1/2}}%
\sum_{i=1}^{n}\widehat{v}_{i}\mathcal{I}_{i}\left( x\right) ,
\end{equation*}%
Just like in Section \ref{sec:Khmlinear}, we have the following result.

\begin{theorem}
\label{M_nest-nonlinear} Assuming that $H_{0}$ holds true, under Conditions $%
C1-C3$, we have that 
\begin{equation*}
\widetilde{\mathcal{M}}_{n}\left( x\right) \overset{weakly}{%
\Rightarrow }\mathcal{U}\left( x\right) \text{.}
\end{equation*}
\end{theorem}

This methodology can be generalized, of course, to the situation when the
constrained estimate belongs to the intersection of several boundary
surfaces. In this case we would express several parameter values as
functions of other parameters and plug these functions into the linear
approximation obtaining a nonlinear expression in the remaining parameters.
We would define the derivative of the new approximation and adjust the
Khmaladze transformation to account now for several nonlinear equality
constraints. The rest of the methodology and the asymptotic result would
remain the same as above.

%{\Huge It is important to remark that altho}ugh in this section we have described the Khmaladze transformation when the transformation is conducted without imposing any constrains induced by the null hypothesis, to ensure the power of our testing procedures the transformation must be conducted under the constraints given by the null hypothesis. Thus, effectively, we would consider constrained forward residuals. In fact this is exactly how we have computed the transformation, so is the test, in Section \ref% {sec:simulations_and_data}. DO\ WE\ NEED\ THIS\ PARAGRAPH?

%\paragraph*{\textbf{Power and local alternatives.}}

\subsection{\textbf{COMPUTATIONAL ISSUES}}

$\left. {}\right. $

This section is devoted at how we can compute our statistic. In view of the
CUSUM interpretation, we shall rely on the standard recursive residuals. We will illustrate 
the computational issues using notations from the case of all linear constraints but, of course, 
the methodology in the case of non-linear constraints will be the same. 

Note that since $f\left( x\right) $ is continuous the probability of a
tie is zero, so that we can always consider the case $x_{i}<x_{i+1}$. Now with this view we have that
\begin{equation*}
\mathcal{M}_{n}\left( x\right) =:\frac{1}{n^{1/2}}\sum_{i=1}^{n}v_{i}%
\mathcal{I}_{i}\left( x\right) ;
\end{equation*}%
can be written with $v_{i}$ replaced by $v_{i}=u_{i}-\boldsymbol{P}%
_{i}^{\prime }A_{n}^{+}\left( x\right) C_{n}\left( x\right) $ and now
\begin{equation*}
\left( \frac{1}{n}\sum_{k=1}^{n}\boldsymbol{P}_{k}\boldsymbol{P}_{k}^{\prime
}\mathcal{J}_{k}\left( x_{i}\right) \right) ^{+}\frac{1}{n}\sum_{k=1}^{n}%
\boldsymbol{P}_{k}u_{k}\mathcal{J}_{k}\left( x_{i}\right)
=:A_{n,i}^{+}C_{n,i}\text{.}
\end{equation*}

Then from a computational point of view is worth observing that
\begin{equation*}
A_{n,k}^{+}=A_{n,k+1}^{+}-\frac{A_{n,k+1}^{+}\boldsymbol{P}_{k}\boldsymbol{P}%
_{k}^{\prime }A_{n,k+1}^{+}}{n+\boldsymbol{P}_{k}^{\prime }A_{n,k+1}^{+}%
\boldsymbol{P}_{k}}
\end{equation*}%
and%
\begin{equation*}
A_{n.k}^{+}C_{n,k}=A_{n,k+1}^{+}C_{n.k+1}+A_{n,k}^{+}\boldsymbol{P}%
_{k}\left( u_{k}-\boldsymbol{P}_{k}^{\prime }A_{n,k+1}^{+}C_{n,k+1}\right)
\end{equation*}%
see \cite{BrownDurbinEvans} for similar arguments.
Alternatively, we could have considered the Cusum of backward recursive
residuals, in which case we would have use the computational formulae,%
\begin{equation*}
\bar{A}_{n,k+1}^{+}=\bar{A}_{n,k}^{+}-\frac{\bar{A}_{n,k}^{+}\boldsymbol{P}%
_{k+1}\boldsymbol{P}_{k+1}^{\prime }\bar{A}_{n,k}^{+}}{n+\boldsymbol{P}%
_{k+1}^{\prime }\bar{A}_{n,k}^{+}\boldsymbol{P}_{k+1}}
\end{equation*}%
and%
\begin{equation*}
\overline{A}_{n,k+1}^{+}C_{n,k+1}=\overline{A}_{n,k}^{+}C_{n,k}+A_{n,k+1}^{+}%
\boldsymbol{P}_{k+1}\left( u_{k+1}-\boldsymbol{P}_{k+1}^{\prime }\overline{A}%
_{n,k}^{+}C_{n,k}\right) .
\end{equation*}

Of course in the previous formulas one would replace $u_{i}$ by $\widehat{u}%
_{i}$ or $y_{i}$.

\subsection{\textbf{POWER AND LOCAL ALTERNATIVES}}
\label{sec:power}

$\left. {}\right. $

Here we describe the local alternatives for which the tests based on $%
\widetilde{\mathcal{M}}_{n}\left( x\right) $ have no trivial power. For that
purpose, assume that the true model is such that 
\begin{equation*}
H_{a}=m\left( x\right) +n^{-1/2}m_{1}\left( x\right) \text{,}
\end{equation*}%
where $m\left( x\right) \in \mathcal{M}_0$ and $m_{1}\left( x\right) $ is
incompatible with $\mathcal{M}_0$ at least on some subset $\mathcal{X}%
_{1}\subset \mathcal{X=}\left( 0,1\right) $ with the Lebesgue measure
strictly greater than zero. Then, we have the result of Proposition \ref{Loc_alt}.

\begin{proposition}
\label{Loc_alt}Assuming Conditions $C1-C3$, under $H_{a}$ we have that 
\begin{equation*}
\widetilde{\mathcal{M}}_{n}\left( x\right) +\mathcal{L}\left( x\right) 
\overset{weakly}{\Rightarrow }\mathcal{U}\left( x\right) \text{, \ \ }x\in %
\left[ 0,1\right]
\end{equation*}%
where 
\begin{equation*}
\mathcal{L}\left( x\right) =\int_{0}^{x}\left\{ m_{1}\left( v\right) -%
\mathbf{P}_{L}^{\prime }\left( v\right) A_{L}^{+}\left( v\right) \int_{v}^{1}%
\mathbf{P}_{L}\left( w\right) m_{1}\left( w\right) f_{X}\left( w\right)
dw\right\} f_{X}\left( v\right) dv\text{.}
\end{equation*}
\end{proposition}

One consequence of Proposition \ref{Loc_alt} is that not only tests based on 
$\widetilde{\mathcal{M}}_{n}\left( x\right) $ are consistent since $\mathcal{%
L}\left( x\right) $ is a nonzero function, but that they also have a
nontrivial power against local alternatives converging to the null at the
\textquotedblleft parametric\textquotedblright\ rate $n^{-1/2}$.

\paragraph*{\textbf{Role of constraints for the power of the test}} In our discussion of the Khmaladze's transformation we stated that in that transformation one has to enforce the binding constraints obtained in the constrained estimation under the null hypothesis. Here we present a simple example to illustrate what happens to the power of the test if such constraints are not enforced. 

To keep arguments simple, consider the case where $L=2$ when we use the B-splines to approximate the function and the null hypothesis, and take the null hypothesis to be that of an  increasing regression function. That is, we employ the approximation $m_{\mathcal{B}}\left( x;2\right) = \beta _{1}p_{1,2}\left( x;q\right) + \beta _{2}p_{2
,2}\left( x;q\right)$ and
rewrite it in the following way: 
\begin{eqnarray*}
m_{\mathcal{B}}\left( x;2\right)  &=&\alpha _{1}\left( p_{1,2}\left(
x;q\right) +p_{2,2}\left( x;q\right) \right) +\alpha _{2}p_{2,2}\left(
x;q\right)  \\
&=&\alpha _{1}p\left( x;q\right) +\alpha _{2}p_{2,2}\left( x;q\right) ,
\end{eqnarray*}%
where $\alpha _{1}=\beta _{1}$ $\alpha _{2}=\beta _{2}-\beta _{1}$ and $%
p\left( x;q\right) =:p_{1,2}\left( x;q\right) +p_{2,2}\left( x;q\right) $. Of course, because of the B-splines properties, we have $p\left( x;q\right)=1$ but for the sake of presenting a more general argument, we will keep the more general notation  $p\left( x;q\right)$. The null hypothesis is then written as $\alpha _{2}\geq 0$. Suppose the null
is not true and the true value is $\alpha _{2}<0$. When we estimate the
model, we should expect that our estimator is of the form $\left( \widehat{%
\alpha }_{1},0\right) $. We then define constrained residuals as 
\begin{equation*}
\widehat{u}_{i}=y_{i}-\widehat{\alpha }_{1}p\left( x;q\right) .
\end{equation*}%
According to our methodology, at each step of the transformation we should
only be projecting the residuals on $p\left( x;q\right)$.

Let us analyze what happens if in the transformation we use as our $\mathbf{P%
}_{k}$ that comes from the vector $\left( p\left( x_{k};q\right) ,{p}%
_{2,2}\left( x_{k};q\right) \right) $ instead of using only $p\left( x_{k};q\right)$. In
other words, we use the test statistic 
\begin{equation*}
\widetilde{\mathcal{M}}_{n}\left( x\right) =:\frac{1}{n^{1/2}}\sum_{i=1}^{n}%
\widehat{v}_{i}\mathcal{I}_{i}\left( x\right) \text{,}
\end{equation*}%
where $\widehat{u}_{i}=y_{i}-\widehat{\alpha }_{1}p\left( x;q\right) $ and
the transformation is defined as follows: 
\begin{eqnarray*}
\widehat{v}_{i} &=&\widehat{u}_{i}-\boldsymbol{P}_{i}^{\prime
}A_{n,i}^{+}C_{n,i} \\
C_{n,i} &=&n^{-1}\sum_{k=i}^{n}\boldsymbol{P}_{k}\widehat{u}_{k};~\ \ \
A_{n,i}=n^{-1}\sum_{k=i}^{n}\boldsymbol{P}_{k}\boldsymbol{P}_{k}^{\prime }%
\text{.}
\end{eqnarray*}
Rewrite the test statistic as 
\begin{equation*}
\widetilde{\mathcal{M}}_{n}\left( x\right) =\frac{1}{n^{1/2}}%
\sum_{i=1}^{n}v_{i}\mathcal{I}_{i}\left( x\right) +\frac{1}{n^{1/2}}%
\sum_{i=1}^{n}\left( \widehat{v}_{i}-v_{i}\right) \mathcal{I}_{i}\left( x\right)
\end{equation*}%
and notice that the first term on the right-hand side converges to the
Brownian motion. The second term is negligible under the null, as
established earlier, and partly this is due to the result of Lemma \ref{bias} in Appendix B. It is rather obvious that the power of the test, as
usual, comes from that term. For the test to have power we need that under
the alternative the term $\frac{1}{n}\sum_{i=1}^{n}\left( \widehat{v}%
_{i}-v_{i}\right) \mathcal{I}_{i}\left( x\right) $ converges somewhere
different than zero on a set of a positive measure, so that once we multiply it
by $n^{1/2}$, it diverges to $\pm \infty $. We have that%
\begin{eqnarray}
\widehat{u}_{i}-u_{i} &=&\widehat{m}_{\mathcal{B}}\left( x_{i};2\right) -m_{%
\mathcal{B}}\left( x_{i};2\right) +\left( m_{\mathcal{B}}\left(
x_{i};2\right) -m(x)\right)   \notag \\
&=&\left( \widehat{\alpha }_{1}-\alpha _{1}\right) p\left( x_{i};q\right)
-\alpha _{2}{p}_{2,2}\left( x_{i};q\right) +\left( m_{\mathcal{B}}\left(
x_{i};2\right) -m(x)\right) .  \label{eq:eqpower}
\end{eqnarray}
Lemma \ref{bias} implies that when the transformation uses $\mathbf{P}_{k}$
that comes from the vector $\left( p\left( x_{k};q\right) ,{p}_{2,2}\left(
x_{k};q\right) \right) $, any linear combination $\boldsymbol{%
\mathring{p}}\left( x\right) =:a_{1}p\left( x;q\right) +a_{2}{p}_{2,2}\left(
x;q\right) $ satisfies $\left( \mathcal{T}_{n}\mathcal{B}_{n,2}\right)
\left( x\right) =0$, where $\mathcal{B}_{n,2}\left( x\right)
=n^{-1}\sum_{k=1}^{n}\boldsymbol{\mathring{p}}\left( x_{k}\right) \mathcal{I}%
_{k}\left( x\right) $. This, taking into account representation (\ref{eq:eqpower}), implies that the part of $n^{-1}\sum_{i=1}^{n}\left( \widehat{%
v}_{i}-v_{i}\right) \mathcal{I}_{i}\left( x\right) $ corresponding to $%
\left( \widehat{\alpha }_{1}-\alpha _{1}\right) p\left( x_{i};q\right)
-\alpha _{2}{p}_{2,2}\left( x_{i};q\right) $ in (\ref{eq:eqpower}) will be
zero, and its part corresponding to the bias term $m_{\mathcal{B}}\left(
x_{i};2\right) -m(x)$ in (\ref{eq:eqpower}) will be asymptotically
negligible even when multiplied by $n^{1/2}$. Thus, the power of the test
equals the trivial power because we failed to embed the binding constraint
in the transformation.

On the other hand, if in the transformation we use  only the ``explanatory''
polynomial $p\left( x;q\right) $, we would now have that the
contribution due to $\left( \widehat{\alpha }_{1}-\alpha _{1}\right) p\left(
x_{i};q\right) -\alpha _{2}{p}_{2,2}\left( x_{i};q\right) $ in $%
n^{-1}\sum_{i=1}^{n}\left( \widehat{v}_{i}-v_{i}\right) \mathcal{I}%
_{i}\left( x\right) $ is exactly as that from $\alpha _{2}{p}_{2,2}\left(
x_{i};q\right) $ and it is not equal to zero as ${p}_{2,2}\left( x_{i};q\right) $ is
not in the space generated by $p\left( x_{i};q\right) $.  In other words, %
\begin{equation*}
\frac{1}{n}\sum_{k=1}^{n}\left\{ {p}_{2,2}\left( x_{k};q\right) -\boldsymbol{%
P}_{k}^{\prime }A_{n,k}^{+}\frac{1}{n}\sum_{j=1}^{n}\boldsymbol{P}_{j}{p}%
_{2,2}\left( x_{j};q\right) \mathcal{J}_{j}\left( \widetilde{x}_{k}\right)
\right\} \mathcal{I}_{k}\left( x\right) \neq 0
\end{equation*}%
unless ${p}_{2,2}\left( x_{i};q\right) $ is proportional to ${p}_{2,2}\left(
x_{i};q\right) $, which is not the case.

\section{\textbf{BOOTSTRAP ALGORITHM}}
\label{sec:bootstrap}

One of our motivations to introduce a bootstrap algorithm for our test(s) is
that although it is pivotal, our Monte Carlo experiment suggests that they
suffer from small sample biases. When the asymptotic distribution does not
provide a good approximation to the finite sample one, a standard approach
to improve its performance is to employ bootstrap algorithms, as they
provide small sample refinements. In fact, our Monte Carlo simulation does
suggest that the bootstrap, to be described below, does indeed give a
better finite sample approximation. The notation for the bootstrap is as
usual and we shall implement the fast algorithm of WARP by \cite{GiacominiPW} in the Monte Carlo experiment.

The bootstrap is based on the following \emph{3 STEPS}.

\begin{description}
\item[\emph{STEP 1}] Compute the unconstrained residuals%
\begin{equation*}
\widetilde{u}_{i}=y_{i}-\widetilde{m}_{\mathcal{B}}\left( x_{i};L\right) 
\text{, \ \ \ }i=1,...,n
\end{equation*}%
with $\widetilde{m}_{\mathcal{B}}\left( x_{i};L\right) $ as defined in $%
\left( \ref{uncons}\right) $.

\vskip 0.1in

\item[\emph{STEP 2}] Obtain a random sample of size\emph{\ }$n$ from the
empirical distribution of $\left\{ \widetilde{u}_{i}-\frac{1}{n}%
\sum_{i=1}^{n}\widetilde{u}_{i}\right\} _{i=1}^{n}$. Denote such a sample as 
$\left\{ u_{i}^{\ast }\right\} _{i=1}^{n}$ and compute the bootstrap
analogue of the regression model using $\widehat{m}_{\mathcal{B}}\left(
x_{i};L\right) $, that is%
\begin{equation}
y_{i}^{\ast }=\widehat{m}_{\mathcal{B}}\left( x_{i};L\right) +u_{i}^{\ast }%
\text{, \ }i=1,...,n\text{.}  \label{reg_boot}
\end{equation}

\vskip0.1in

\item[\emph{STEP 3}] Compute the bootstrap analogue of $\widetilde{\mathcal{M%
}}_{n}\left( x\right) $ as 
\begin{equation*}
\widetilde{\mathcal{M}}_{n}^{\ast }\left( x\right) =:\frac{1}{n^{1/2}}%
\sum_{i=1}^{n}\widehat{v}_{i}^{\ast }\mathcal{I}_{i}\left( x\right)
\end{equation*}%
where%
\begin{equation*}
\widehat{v}_{i}^{\ast }=\widehat{u}_{i}^{\ast }-\boldsymbol{P}_{i}^{\prime
}A_{n,i}^{+}C_{n,i}^{\ast };\text{ \ }C_{n,i}^{\ast }=:C_{n,i}^{\ast }\left( 
\widetilde{x}_{i}\right) =\frac{1}{n}\sum_{k=1}^{n}\boldsymbol{P}_{k}%
\widehat{u}_{k}^{\ast }\mathcal{J}_{k}\left( \widetilde{x}_{i}\right)
\end{equation*}%
with $\widehat{u}_{i}^{\ast }=y_{i}^{\ast }-\boldsymbol{P}_{i}^{\prime
}A_{n}^{+}\left( 0\right) C_{n}^{\ast }\left( 0\right) $, $i=1,...,n$.
\end{description}

\begin{theorem}
\label{Mb_n}Under Conditions $C1-C3$, we have that for any continuous
function $g:\mathbb{R\rightarrow R}^{+}$, (in probability), 
\begin{equation*}
g\left( \widetilde{\mathcal{M}}_{n}^{\ast }\left( x\right) \right) \overset{d%
}{\Rightarrow }g\left( \mathcal{U}\left( x\right) \right) \text{.}
\end{equation*}
\end{theorem}

Finally, we can replace $\widehat{u}_{i}^{\ast }$ by $y_{i}^{\ast }$ in the
computation of $\widetilde{\mathcal{M}}_{n}^{\ast }\left( x\right) $. That
is,

\begin{corollary} 
\label{cor:bootstrap}	
Under Conditions $C1-C3$, we have that 
\begin{equation*}
\widetilde{\mathcal{M}}_{n}^{\ast }\left( x\right) -\widetilde{\widetilde{%
\mathcal{M}}}_{n}^{\ast }\left( x\right) =0\text{,}
\end{equation*}%
where%
\begin{equation*}
\widetilde{\widetilde{\mathcal{M}}}_{n}^{\ast }\left( x\right) =:\frac{1}{%
n^{1/2}}\sum_{i=1}^{n}\left( y_{i}^{\ast }-\boldsymbol{P}_{i}^{\prime
}A_{n,i}^{+}\frac{1}{n}\sum_{k=1}^{n}\boldsymbol{P}_{k}y_{k}^{\ast }\mathcal{%
J}_{k}\left( x_{i}\right) \right) \mathcal{I}_{i}\left( x\right) \text{.}
\end{equation*}
\end{corollary}

The proof of Corollary \ref{cor:bootstrap} is immediate by Lemma \ref{bias} in the Appendix and therefore is omitted.

\section{\textbf{MONTE CARLO EXPERIMENTS AND EMPIRICAL EXAMPLES}}

\label{sec:simulations_and_data}

\subsection{\textbf{MONTE CARLO EXPERIMENTS}}

$\left. {}\right. $

In this section we present the results of several computational experiments.
All the results in this section are given for cubic splines with different
number of knots. We present the results for \emph{B-splines} as well as for 
\emph{P-splines} with penalties on the second differences of coefficients.
The penalty parameter is chosen by cross-validation in the unconstrained
estimation. 
In the tables \textquotedblleft \emph{KS}\textquotedblright\ refers to the
Kolmogorov-Smirnov test statistic, \textquotedblleft \emph{CvM}%
\textquotedblright\ refers to the \emph{Cram\'{e}r-von Mises} test statistic
and \textquotedblleft \emph{AD}\textquotedblright\ to the Anderson-Darling
integral test statistic. All three test statistics are based on a Brownian
bridge. $L^{\prime }+1$ denotes the number of equidistant knots (including
the boundary points) on the interval of interest. For example, when $%
L^{\prime }=6$ and the interval is $[0,1]$, we consider knots $%
0,1/6,1/3,1/2,2/3,5/6,1$. In the implementation of P-splines in simulations,
every simulation draw will give a different cross-validation parameter (we
use ordinary cross validation described in Eilers and Marx, $1996$). In our
simulation results for each $L^{\prime }$ we use a modal value of these
cross-validation parameters.

In all the scenarios below 
\begin{equation*}
X\sim \mathcal{U}[0,1],\quad U\sim \mathcal{N}(0,\sigma ^{2}),\quad U\perp X.
\end{equation*}
In Scenarios 1, 3-5 the interval of interest if $[0,1]$ whereas in Scenario
2 of U-shape we consider individually intervals $[0,s_0]$ and $[s_0,1]$ with 
$s_0$ being the switch point.

In the WARP bootstrap implementation, the demeaned residuals and $x$ are
drawn independently.

\vskip 0.1in

\noindent \textbf{Scenario 1 (test for monotonicity)}. We take the following regression function: 
\begin{equation*}
m\left( x\right) =x^{\frac{13}{4}}\text{,}\quad x\in \left[ 0,1\right] \text{%
.}
\end{equation*}%
The results are summarized in Table \ref{MonteCarlo1}.

\vskip 0.05in

\begin{table}[tbp]
%[ht!]
\par
\begin{center}
\begin{tabular}{llllll}
\hline\hline
Setting & Method \hspace*{0.2in} & \multicolumn{2}{c}{B-splines} & 
\multicolumn{2}{c}{P-splines} \\ 
&  & 10\% & 5\% & 10\% & 5\% \\ \hline
$L^{\prime }=6$ & KS bootstrap & 0.113 & 0.054 & 0.0965 & 0.051 \\ 
$N=1000$ & CvM bootstrap & 0.1035 & 0.0475 & 0.093 & 0.0465 \\ 
$\sigma=0.25$ & AD bootstrap & 0.1065 & 0.054 & 0.1005 & 0.054 \\ \hline
%%%%%%
$L^{\prime }=9$ & KS bootstrap & 0.102 & 0.044 & 0.119 & 0.052 \\ 
$N=1000$ & CvM bootstrap & 0.101 & 0.048 & 0.11 & 0.0455 \\ 
$\sigma=0.25$ & AD bootstrap & 0.0985 & 0.042 & 0.1045 & 0.048 \\ \hline
%%%%%%
$L^{\prime }=14$ & KS bootstrap & 0.0945 & 0.043 & 0.105 &  0.0555\\ 
$N=1000$ & CvM bootstrap & 0.098 & 0.0425 & 0.0955 &  0.045\\ 
$\sigma=0.25$ & AD bootstrap & 0.093 & 0.043 &  0.096 & 0.049 \\ \hline
%%%%%%%
$L^{\prime }=19$ & KS bootstrap & 0.089 & 0.0485 & 0.1055  & 0.059 \\ 
$N=1000$ & CvM bootstrap & 0.105 & 0.0555 & 0.1035 & 0.0505 \\ 
$\sigma=0.25$ & AD bootstrap & 0.1085 & 0.0545 &  0.106 & 0.0495 \\ \hline
\end{tabular}
\vspace*{0.05in}
\end{center}
\caption{{\protect\small Tests for monotonically increasing regression
function in Scenario 1. Rejection rates in 2000 simulations. $L^{\prime }+1$
denotes the number of equidistant knots on $[0,1]$. $N$ denotes the number
of observations in each simulation. $\protect\sigma$ is the standard
deviation in the error distribution.}}
\label{MonteCarlo1}
\end{table}

\vskip 0.1in

\noindent \textbf{Scenario 2 (test for U-shape)}. The regression function is defined
as 
\begin{align*}
m(x) & =10 \left(\log(1+x)-0.33\right)^2.
\end{align*}
The graph of this function is U-shaped with the switch point at $%
s_0=e^{0.33}-1$. In simulations $s_0$ is taken to be known.

The results are summarized in Table \ref{MonteCarloU2}. We use two different 
\emph{B-splines} -- one on $[0,s_{0}]$ and the other on $[s_{0},1]$. We
analyze the properties of the testing procedure in two approaches. In the
first approach additional restrictions are imposed for the two 
B-splines to be joined continuously at $s_{0}$, and in the second approach
these two B-splines are 
joined smoothly at $s_{0}$ (see details in Example 2).

\begin{table}[tbp]
%[p]
\par
\begin{center}
\begin{tabular}{llllllllll}
\hline\hline
&  & \multicolumn{4}{c}{Continuously joined} & \multicolumn{4}{c}{Smoothly
joined} \\ \hline
Setting & Method \hspace*{0.2in} & \multicolumn{2}{c}{B-splines} & 
\multicolumn{2}{c}{P-splines} & \multicolumn{2}{c}{B-splines} & 
\multicolumn{2}{c}{P-splines} \\ 
&  & 10\% & 5\% & 10\% & 5\% & 10\% & 5\% & 10\% & 5\% \\ \hline
$L^{\prime }=4$ & KS & 0.0935 & 0.048 & 0.113 & 0.051 & 0.098 & 0.062 & 0.1105
& 0.055 \\ 
$N=1000$ & CvM & 0.106 & 0.046 & 0.1 & 0.0515 & 0.101 & 0.0575 &  0.1005 & 0.05
 \\ 
$\sigma=0.25$ & AD & 0.107 & 0.048 & 0.098 & 0.055 & 0.1015 & 0.0615 & 0.094
&  0.0505\\ 
\vspace*{0.02in} &  &  &  &  &  &  &  &  &  \\ 
$L^{\prime }=6$ & KS & 0.1105 & 0.0555 & 0.112 & 0.051 & 0.107 & 0.0575 & 0.0935
 &  0.0495\\ 
$N=1000$ & CvM & 0.101 & 0.0585 & 0.099 & 0.0525 & 0.1 & 0.0575 & 0.0955 & 0.048
\\ 
$\sigma=0.25$ & AD & 0.1015 & 0.0565 & 0.098 & 0.047 & 0.1055 & 0.055 & 0.0965 & 0.0485 \\ 
 \hline
\end{tabular}%
\end{center}
\caption{{\protect\small Tests for U-shape with the switch at $%
s_0=e^{0.33}-1 $ in Scenario 2. Rejection rates in 2000 simulations. $%
L^{\prime }+1$ denotes the number of equidistant knots on each subinterval $%
[0,s_0]$ and $[s_0,1]$ . $N$ denotes the number of observations in each
simulation. $\protect\sigma $ is the standard deviation in the error
distribution.}}
\label{MonteCarloU2}
\end{table}

\vskip 0.1in

\noindent \textbf{Scenario 3 (analysis of power of the test)}. Take the
regression function 
\begin{align*}
m(x) & =(10x-0.5)^3 - \exp(-100(x-0.25)^2))\cdot \mathcal{I}(x<0.5) \\
& + (0.1(x-0.5)-\exp(-100(x-0.25)^2))\cdot \mathcal{I}(x>=0.5)
\end{align*}
and depicted in Figure \ref{fig:dip}. As expected, the power of the test depends on the variance of the error. The
results are summarized in Table \ref{MonteCarlodip}.
\begin{figure}[ht!]
\includegraphics[width=0.4\textwidth]{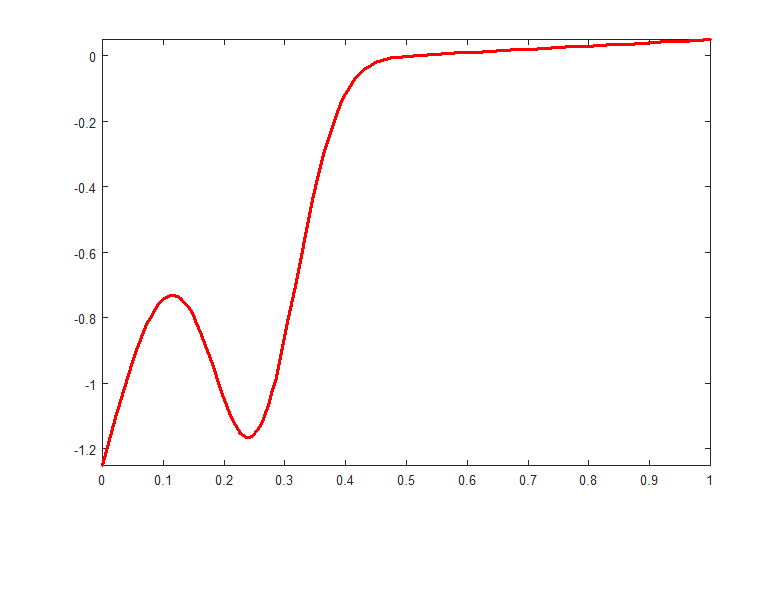}
\caption{{\protect\small Plot of the regression function in Scenario 3.}}
\label{fig:dip}
\end{figure}

\begin{table}[tbp]
%[ht]
\par
\begin{center}
\begin{tabular}{llllll}
\hline\hline
Setting & Method \hspace*{0.2in} & \multicolumn{2}{c}{B-splines} & 
\multicolumn{2}{c}{P-splines} \\ 
&  & 10\% & 5\% & 10\% & 5\% \\ \hline
$L^{\prime }=6$ & KS bootstrap & 1 & 0.998 & 1 & 0.9885 \\ 
$N=1000$ & CvM bootstrap & 0.9155 & 0.7335 & 0.9975 & 0.9895 \\ 
$\sigma=0.5$ & AD bootstrap & 0.985 & 0.9375 &  0.9995 & 0.998 \\ 
\vspace*{0.05in} &  &  &  &  &  \\ 
$L^{\prime }=9$ & KS bootstrap & 0.93 & 0.823 & 0.999 & 0.998 \\ 
$N=1000$ & CvM bootstrap & 0.862 & 0.766 & 0.998 & 0.9935 \\ 
$\sigma=0.5$ & AD bootstrap & 0.9195 & 0.8085 &  0.9985 & 0.9935 \\ 
\vspace*{0.05in} &  &  &  &  &  \\ 
$L^{\prime }=12$ & KS bootstrap & 0.864 & 0.8175 & 0.996 & 0.9885 \\ 
$N=1000$ & CvM bootstrap & 0.8505 & 0.7895 & 0.989  &  0.9725\\ 
$\sigma=0.5$ & AD bootstrap & 0.8655 & 0.799 & 0.9885 & 0.974 \\ 
\vspace*{0.05in} &  &  &  &  &  \\ 
$L^{\prime }=19$ & KS bootstrap & 0.639 & 0.5375 & 0.9815 & 0.9515 \\ 
$N=1000$ & CvM bootstrap & 0.5295 & 0.395 & 0.9435 & 0.8795 \\ 
$\sigma=0.5$ & AD bootstrap & 0.571 & 0.428 & 0.9445 & 0.892 \\ \hline
$L^{\prime }=6$ & KS bootstrap & 1 & 1 & 1 & 1 \\ 
$N=1000$ & CvM bootstrap & 1 & 1 & 1 & 1 \\ 
$\sigma=0.25$ & AD bootstrap & 1 & 1 & 1 & 1 \\ 
\vspace*{0.05in} &  &  &  &  &  \\ 
$L^{\prime }=9$ & KS bootstrap & 1 & 1 & 1 & 1 \\ 
$N=1000$ & CvM bootstrap & 1 & 1 & 1 & 1 \\ 
$\sigma=0.25$ & AD bootstrap & 1 & 1 & 1 & 1 \\ 
\vspace*{0.05in} &  &  &  &  &  \\ 
$L^{\prime }=12$ & KS bootstrap & 0.9995 & 0.9995 & 1 & 1 \\ 
$N=1000$ & CvM bootstrap & 0.996 & 0.9945 & 1 & 1 \\ 
$\sigma=0.25$ & AD bootstrap & 1 & 0.996 & 1 & 1 \\ 
\vspace*{0.05in} &  &  &  &  &  \\ 
$L^{\prime }=19$ & KS bootstrap & 0.996 & 0.986 & 1 & 1 \\ 
$N=1000$ & CvM bootstrap & 0.9155 & 0.836 & 1 & 1 \\ 
$\sigma=0.25$ & AD bootstrap & 0.955 & 0.891 & 1 & 1 \\ \hline
\end{tabular}
\vspace*{0.05in}
\end{center}
\caption{{\protect\small Tests for monotonicity in Scenario 3. Rejection
rates in 2000 simulations. $L^{\prime }+1$ denotes the number of equidistant
knots on $[0,1]$. $N$ denotes the number of observations in each simulation. 
$\protect\sigma$ is the standard deviation in the error distribution. }}
\label{MonteCarlodip}
\end{table}

The power of monotonicity tests based on this regression function is
considered in \cite{GhosalSenVV} and a similar
regression function is considered in \cite{HallHeckman}. Note that \cite{GhosalSenVV} considers smaller sample sizes and also smaller
standard deviation of noise with $\sigma=0.1$.

\vskip 0.1in

\noindent \textbf{Scenario 4 (analysis of power of the test)}. The
regression function 
\begin{equation*}
m(x) = x + 0.415\exp(-ax^2), \quad a>0.
\end{equation*}
and depicted in Figure \ref{fig:dipother}. The left-hand side graph in
Figure \ref{fig:dipother} is for the case $a=50$ and the right-hand side
graph in Figure \ref{fig:dipother}. In the latter case the non-monotonicity
dip is smaller. These situations are considered to be challenging for
monotonicity tests as these functions are somewhat close to the set of
monotone functions (in any conventional metric). As expected, the power of
the test depends on the value of parameter $a$ and also depends on the
variance of the error. The results are summarized in Table \ref{MonteCarlodipother}.

\begin{figure}[!h]
\begin{center}
\begin{minipage}{\linewidth}
\begin{minipage}{0.5\linewidth}
% \hspace*{0.05in}
  \includegraphics[width=0.95\linewidth]{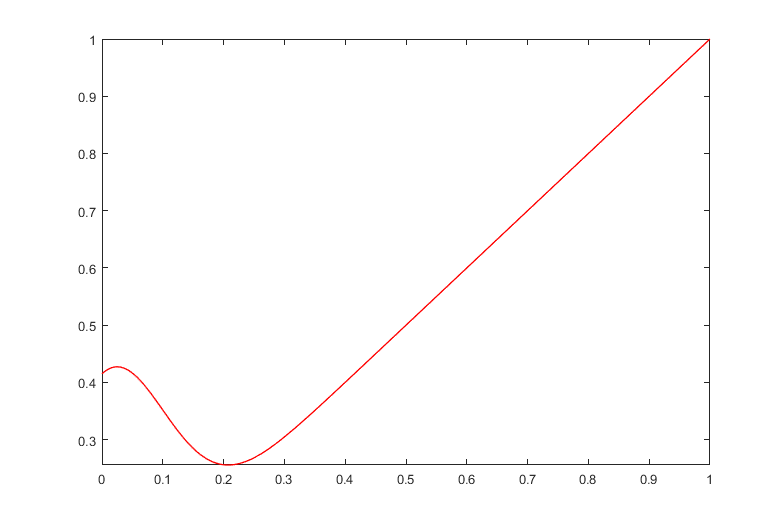}
\end{minipage}
\begin{minipage}{0.5\linewidth}
% \hspace*{0.05in}
\includegraphics[width=0.95\linewidth]{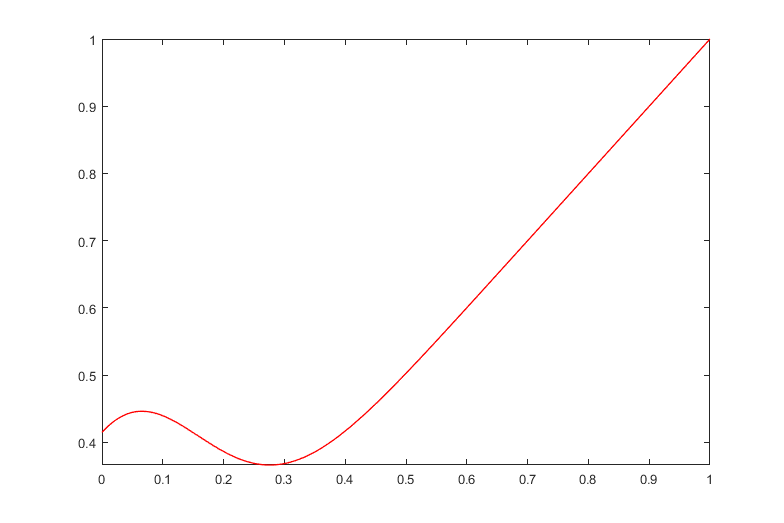}
\end{minipage}
\end{minipage}
\end{center}
\caption{{\protect\small Plot of the regression function in Scenario 4. The
left-hand side graph is for $a=50$ and the right-hand side graph is for $%
a=20 $.}}
\label{fig:dipother}
\end{figure}

\begin{table}[tbp]
%[p]
\par
\begin{center}
\begin{tabular}{llllllllll}
\hline\hline
&  & \multicolumn{4}{c}{$a=50$} & \multicolumn{4}{c}{$a=20$} \\ \hline
Setting & Method \hspace*{0.2in} & \multicolumn{2}{c}{B-splines} & 
\multicolumn{2}{c}{P-splines} & \multicolumn{2}{c}{B-splines} & 
\multicolumn{2}{c}{P-splines} \\ 
&  & 10\% & 5\% & 10\% & 5\% & 10\% & 5\% & 10\% & 5\% \\ \hline
$L^{\prime }=6$ & KS & 0.382 & 0.2425 & 0.575 & 0.423 & 0.25 & 0.139 & 0.338 
 &  0.216 \\ 
$N=1000$ & CvM & 0.3385 & 0.2245 & 0.5165 &  0.384 & 0.2465 & 0.1335 & 0.311  & 0.212 
 \\ 
AD $\sigma=0.5$ & AD & 0.434 & 0.228 & 0.6015 & 0.457 & 0.2475 & 0.1395 & 0.3355 
 &  0.221 \\ 
\vspace*{0.02in} &  &  &  &  &  &  &  &  &  \\ 
$L^{\prime }=9$ & KS & 0.381 & 0.2545 &  0.572 & 0.4655 & 0.2 & 0.124 & 0.3405 
&  0.221\\ 
$N=1000$ & CvM & 0.363 & 0.2355 & 0.536 & 0.3915 & 0.2055 & 0.132 &  0.3385 & 0.206 
 \\ 
$\sigma=0.5$ & AD & 0.4735 & 0.3165 & 0.6095 & 0.4555 & 0.228 & 0.145 & 0.348 
&  0.206 \\ 
\vspace*{0.02in} &  &  &  &  &  &  &  &  &  \\ 
$L^{\prime }=12$ & KS & 0.3975& 0.27 & 0.5995 & 0.4745 & 0.22 & 0.1375 & 0.343 
 &  0.2185 \\ 
$N=1000$ & CvM & 0.3905 & 0.2565 & 0.5545 & 0.4075 & 0.2335 & 0.152 &  0.3215 & 0.2065 
 \\ 
$\sigma=0.5$ & AD & 0.486 & 0.3505 &  0.614 & 0.5035 & 0.2525 & 0.1625 & 0.3415 & 0.2195 
 \\ 
\vspace*{0.02in} &  &  &  &  &  &  &  &  &  \\ 
$L^{\prime }=19$ & KS & 0.4045 & 0.284 & 0.5905 & 0.476 & 0.2185 & 0.1235 & 0.3455 
 &  0.2265 \\ 
$N=1000$ & CvM & 0.418 & 0.308 & 0.5525 & 0.414 & 0.232 & 0.145 &  0.318 & 0.1995 
 \\ 
$\sigma=0.5$ & AD & 0.497 & 0.37 & 0.6125 & 0.4915 & 0.2485 & 0.155 &  0.341 & 0.2205 
 \\ \hline
 %%%%%%%%%%%%%%%%%%%%%%%%%%%
$L^{\prime }=6$ & KS & 0.9 & 0.8405 & 0.986 & 0.9625 & 0.5795 & 0.4605 & 0.756 
 &  0.6415 \\ 
$N=1000$ & CvM & 0.8295 & 0.7025 & 0.9615 & 0.913 & 0.5895 & 0.4195 &  0.741 & 0.626 
 \\ 
$\sigma=0.25$ & AD & 0.939 & 0.854 & 0.9835  & 0.9665 & 0.608 & 0.4505 & 0.7275
&  0.6295 \\ 
\vspace*{0.02in} &  &  &  &  &  &  &  &  &  \\ 
$L^{\prime }=9$ & KS & 0.919 & 0.8385 & 0.986 & 0.9685 & 0.484 & 0.3355 & 0.6805 
 & 0.5645 \\ 
$N=1000$ & CvM & 0.8355 & 0.7275 & 0.966 & 0.911 & 0.46 & 0.347 &  0.6495 & 0.5065 
 \\ 
$\sigma=0.25$ & AD & 0.937 & 0.8695 & 0.99 & 0.9615 & 0.4995 & 0.374 &  0.6655 &  0.512 \\ 
\vspace*{0.02in} &  &  &  &  &  &  &  &  &  \\ 
$L^{\prime }=12$ & KS & 0.9235 & 0.842 &  0.9865 & 0.9705 & 0.461 & 0.334 & 0.6785 
 &  0.5585 \\ 
$N=1000$ & CvM & 0.863 & 0.756 & 0.97 & 0.9325 & 0.436 &  0.318 & 0.6525 
&  0.4965 \\ 
$\sigma=0.25$ & AD & 0.951 & 0.8895 & 0.9865 & 0.9705 & 0.492 & 0.3575 & 0.658 
 &  0.517 \\ 
\vspace*{0.02in} &  &  &  &  &  &  &  &  &  \\ 
$L^{\prime }=19$ & KS & 0.925 & 0.8505 & 0.9865 & 0.974 & 0.4835 & 0.3505 & 0.698 
 & 0.585  \\ 
$N=1000$ & CvM & 0.8835 & 0.802 & 0.9745 &  0.9355 & 0.4595 & 0.345 & 0.6625 & 0.495 \\ 
$\sigma=0.25$ & AD & 0.941 & 0.8985 & 0.987 & 0.9695 & 0.486 & 0.3665 &  0.666 & 0.5105 
 \\ \hline
%%%%%%%%%%%%%%%%%%%%
$L^{\prime }=6$ & KS & 1 & 1 &  1 & 1  & 1 & 0.998 &  1 & 1 \\ 
$N=1000$ & CvM & 1 & 1 &  1 & 1  & 0.9985 & 0.9965 & 1  & 1 \\ 
$\sigma=0.1$ & AD & 1 & 1 &  1 & 1  & 0.9995 & 0.9975 & 1 & 1 \\ 
\vspace*{0.02in} &  &  &  &  &  &  &  &  &  \\ 
$L^{\prime }=9$ & KS & 1 & 1 &  1 & 1  & 0.9935 & 0.9875 & 1 & 1 \\ 
$N=1000$ & CvM & 1 & 1 &  1 & 1  & 0.9915 & 0.9825 & 1 & 0.9955 \\ 
$\sigma=0.1$ & AD & 1 & 1 &  1 & 1  & 0.9915 & 0.9865 &  1 &  0.997 \\ 
\vspace*{0.02in} &  &  &  &  &  &  &  &  &  \\ 
$L^{\prime }=12$ & KS & 1 & 1 &  1 & 1  & 0.968 & 0.9535 &  0.9995 &  0.998\\ 
$N=1000$ & CvM & 1 & 1 &  1 & 1  & 0.9539 & 0.932 & 0.9975 &  0.995 \\ 
$\sigma=0.1$ & AD & 1 & 1 &  1 & 1  &  0.969 & 0.951 & 0.998 & 0.997 \\ 
\vspace*{0.02in} &  &  &  &  &  &  &  &  &  \\ 
$L^{\prime }=19$ & KS & 1 & 1 &  1 & 1  & 0.973 & 0.9515 &  0.9995 &  0.9995 \\ 
$N=1000$ & CvM & 1 & 1 &  1 & 1  & 0.9525 & 0.9285 &  0.9985 &  0.9955\\ 
$\sigma=0.1$ & AD & 1 & 1 &  1 & 1  & 0.966 & 0..948 &  0.999 &  0.997\\ \hline
\end{tabular}%
\end{center}
\caption{{\protect\small Tests for monotonicity in Scenario 4. Rejection
rates in 2000 simulations. $L^{\prime }+1$ denotes the number of equidistant
knots on $[0,1]$ (including 0 and 1). $N$ denotes the number of observations
in each simulation. $\protect\sigma $ is the standard deviation in the error
distribution.}}
\label{MonteCarlodipother}
\end{table}
The power of monotonicity tests based on this regression function is
examined in \cite{GhosalSenVV} and a similar regression function was considered in \cite{Bowman_etal}. 
Note that \cite{GhosalSenVV} uses smaller sample sizes and also only $a=50$ and $\sigma=0.1$ to analyze power implications.

\vskip 0.1in

\noindent \textbf{Scenario 5 (test for log-convexity)}. We take the following regression function: 
\begin{equation*}
m\left( x\right) =\exp(x^2), \quad x\in \left[ 0,1\right] \text{%
.}
\end{equation*}%
The results are summarized in Table \ref{MonteCarlo:logconvex}. In this case, the results for P-splines are the same as for B-splines as the cross-validation criterion indicated 0 as the optimal penalty parameter in the overwhelming majority of simulations. 

\vskip 0.05in

\begin{table}[tbp]
%[ht!]
\par
\begin{center}
\begin{tabular}{llllll}
\hline\hline
Setting & Method \hspace*{0.2in} & \multicolumn{2}{c}{B-splines} \\ 
&  & 10\% & 5\%  \\ \hline
$L^{\prime }=6$ & KS bootstrap & 0.1045 & 0.0535  \\ 
$N=1000$ & CvM bootstrap & 0.1015 & 0.054  \\ 
$\sigma=0.25$ & AD bootstrap & 0.1015 & 0.0495  \\ \hline
%%%%%%
$L^{\prime }=9$ & KS bootstrap & 0.098 & 0.0435  \\ 
$N=1000$ & CvM bootstrap & 0.1 & 0.0445  \\ 
$\sigma=0.25$ & AD bootstrap & 0.1 & 0.048  \\ \hline
%%%%%%
$L^{\prime }=12$ & KS bootstrap & 0.1135 & 0.053  \\ 
$N=1000$ & CvM bootstrap & 0.0935 & 0.0465  \\ 
$\sigma=0.25$ & AD bootstrap & 0.095 & 0.0495  \\ \hline
%%%%%%%
\end{tabular}
\vspace*{0.05in}
\end{center}
\caption{{\protect\small Tests for log-convexity in Scenario 5. Rejection rates in 2000 simulations. $L^{\prime }+1$
denotes the number of equidistant knots on $[0,1]$. $N$ denotes the number
of observations in each simulation. $\protect\sigma$ is the standard
deviation in the error distribution.}}
\label{MonteCarlo:logconvex}
\end{table}

\subsection{\textbf{APPLICATIONS}}

$\left. {}\right.$

\noindent 1. \textbf{Hospital data} Here we use data on hospital finance for
332 hospitals in California in 2003.\footnote{%
The dataset is from 
https://www.kellogg.northwestern.edu/faculty/dranove/\newline
htm/dranove/coursepages/mgmt469.htm. The original dataset is for 333
hospitals but we had to remove the observation that had missing information
about administrative expenses.} The data include many variables related to
hospital finance and hospital utilization. We are interested in analyzing
the effect of revenue derived from patients on administrative expenses.

Figure \ref{fig:Hospital} is a scatter plot of the logarithm of patient
revenue and the logarithm of administrative expenses with the fitted curve
obtained using cubic \emph{B-splines} with $L^{\prime }+1=5$ uniform knots
in the range of values of the log of patient revenue. The fitted cure is
obtained under the monotonicity restriction.

\begin{figure}[!h]
\includegraphics[width=0.6\textwidth]{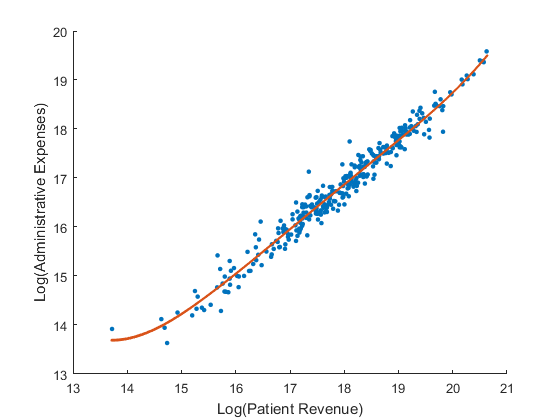}
\caption{{\protect\small Hospital data. Plot of the log of administrative
expenses and the log of patient revenue and the constrained (under
monotonicity) fit by cubic B-splines with $L^{\prime }+1=5$ uniform knots in
the domain of the log of patient revenue.}}
\label{fig:Hospital}
\end{figure}

We conduct tests for the following hypotheses: (\textbf{a}) monotonicity; (%
\textbf{b}) convexity; (\textbf{c}) monotonicity and convexity.

In order to correct for heteroscedasticity of the errors, we estimate the
scedastic function $\widehat{\sigma }^{2}(x)$ using residuals obtained in
the unconstrained estimation using cubic \emph{B-splines} with the same set
of knots. The scedastic function $\widehat{\sigma }^{2}(x)$ is estimated by
regressing the logarithm of the squared unconstrained residuals on a linear
combination of first-order \emph{B-splines} with $6$ knots in the domain of
the log of patient revenue.

We then consider the constrained residuals divided by $\widehat{\sigma }(x)$
when calculating \emph{KS}, \emph{CvM} and \emph{AD} test statistics and
unconstrained residuals divided by $\widehat{\sigma }(x)$ when drawing
bootstrap samples. After a bootstrap sample of residuals is drawn, we
multiply each residual by the corresponding $\widehat{\sigma }(x)$ when
generating a bootstrap sample of observations of the dependent variable.

We implement the testing procedure by conducting the Khmaladze
transformation both from the right end of the support (as is described
theoretically in this paper) and from the left end of the support and
obtained extremely similar results. More specifically, we only report
results when the transformation is conducted from the right end of the
support.

In the case of $P$\emph{-splines}, we use the same \emph{B-spline} basis,
take the second-order penalty and choose the penalization constant using the
ordinary cross-validation criterion as in Eilers and Marx (1996). The
penalty enters unconstrained optimization problems as well as constrained
ones.

Tables \ref{table:HospA}-\ref{table:HospC} present results of our testing.
Namely, Table \ref{table:HospA} shows test statistics for the null
hypothesis of the monotonically increasing regression function and also
bootstrap critical values using both \emph{B-splines} and \emph{P-splines}.
Table \ref{table:HospB} presents analogous results for the null hypothesis
of convexity of the regression function. Table \ref{table:HospC} gives
results for the joint null hypothesis of monotonicity and convexity.

%\FloatBarrier
\begin{table}[tbp]
%[!h]
\par
\begin{center}
\begin{tabular}{ccccccccc}
\hline\hline
\quad & \multicolumn{4}{c}{B-splines} & \multicolumn{4}{c}{P-splines} \\ 
\quad & Test statistic & \multicolumn{3}{c}{Bootstrap critical value} & Test
statistic & \multicolumn{3}{c}{Bootstrap critical value} \\ 
Method & \quad & 10\% & 5\% & 1\% & \quad & 10\% & 5\% & 1\% \\ \hline
KS & 1.0365 & 1.1801 & 1.3028 & 1.6393 & 0.7965 & 0.8502 & 0.9246 & 1.1044 \\ 
CvM & 0.2842 & 0.3428 & 0.45 & 0.7656 & 0.0614 & 0.1113 & 0.1356 & 0.2138  \\ 
AD & 1.5477 & 1.739 & 2.2068 & 3.5445 & 0.4691 & 0.664 & 0.8037 & 1.1031 \\ \hline
\end{tabular}
%\vspace*{0.05in}
\end{center}
\par
%\par
\caption{{\protect\small Hospital data. Test statistics and bootstrap
critical values under the null hypothesis of monotonicity of the regression
function. Bootstrap critical values are from 400 bootstrap replications.}}
\label{table:HospA}
\end{table}

\begin{table}[tbp]
%[!h]
\par
\begin{center}
\begin{tabular}{ccccccccc}
\hline\hline
\quad & \multicolumn{4}{c}{B-splines} & \multicolumn{4}{c}{P-splines} \\ 
\quad & Test statistic & \multicolumn{3}{c}{Bootstrap critical value} & Test
statistic & \multicolumn{3}{c}{Bootstrap critical value} \\ 
Method & \quad & 10\% & 5\% & 1\% & \quad & 10\% & 5\% & 1\% \\ \hline
KS & 0.9064 & 1.3918 & 1.5949 & 1.9306 & 0.7678 & 1.3462 & 1.5503 & 1.9284 \\ 
CvM & 0.1188 & 0.5137 & 0.6783 & 1.1675 & 0.1302 & 0.4994 & 0.6404 & 1.1974  \\ 
AD & 0.7103 & 2.7328 & 3.6416 & 7.1427 & 0.7241 & 2.5773 & 3.3495 & 7.119 \\ \hline
\end{tabular}
%\vspace*{0.05in}
\end{center}
%\vspace*{0.05in}
\par
%\par
\caption{{\protect\small Hospital data. Test statistics and bootstrap
critical values under the null hypothesis of convexity of the regression
function. Bootstrap critical values are from 400 bootstrap replications.}}
\label{table:HospB}
\end{table}

\begin{table}[tbp]
%[!h]
\par
\begin{center}
\begin{tabular}{ccccccccc}
\hline\hline
\quad & \multicolumn{4}{c}{B-splines} & \multicolumn{4}{c}{P-splines} \\ 
\quad & Test statistic & \multicolumn{3}{c}{Bootstrap critical value} & Test
statistic & \multicolumn{3}{c}{Bootstrap critical value} \\ 
Method & \quad & 10\% & 5\% & 1\% & \quad & 10\% & 5\% & 1\% \\ \hline
KS & 1.218 & 1.4347 & 1.6063 & 1.9643 & 1.092 & 1.3568 & 1.6256 & 1.903 \\ 
CvM & 0.2497 & 0.5488 & 0.7375 & 1.4689 & 0.1895 & 0.4785 & 0.7751 & 1.221  \\ 
AD & 1.2513 & 3 & 3.8429 & 7.0125 & 0.9857 & 2.6499 & 3.8105 & 6.9091 \\ \hline
\end{tabular}
%\vspace*{0.05in}
\end{center}
\par
%\par
\caption{{\protect\small Hospital data. Test statistics and bootstrap
critical values under the null hypothesis of both monotonicity and convexity
of the regression function. Bootstrap critical values are from 400 bootstrap
replications.}}
\label{table:HospC}
\end{table}

As we can see from Tables \ref{table:HospA}--\ref{table:HospC}, we do not
reject any of the three hypotheses even at 10\% level.

$\left. {}\right.$

\noindent 2. \textbf{Energy consumption in the Southern region of Russia.}

The data are on daily energy consumption (in MWh) and average daily
temperature (in Celsius) in the Southern region of Russia in the period from
February 1, 2016 till January 31, 2018. The data have been downloaded from
the official website of System Operator of the Unified Energy System of
Russia.\footnote{%
http://so-ups.ru/}

We provide tests for U-shape with a switch at ${17.6}^{\circ }$ and
convexity using the approaches outlined in section \ref{section:BPsplines}
in the context of Examples 2 and 1, respectively. In order to correct for
heteroscedasticity of the errors, we estimate the scedastic function $%
\widehat{\sigma }^{2}(x)$ using residuals obtained in the unconstrained
estimation using \emph{B-splines} (or \emph{P-splines}, respectively). The
scedastic function is estimated using cubic \emph{B-splines} with 6 uniform
knots and in the form of 
\begin{equation*}
\sigma ^{2}(x)=\left( \sum_{k=1}^{8}c_{k}B_{k,8}\right) ^{2}.
\end{equation*}

Figure \ref{fig:SFO} gives scatter plots of the data together with fitted
curves obtained under the U-shape constraint with the switch at $s_{0}={17.6}%
^{\circ }$. This constraint fit is obtained in accordance with the technique
in section \ref{section:BPsplines}. Namely, we consider individual \emph{%
B-spline} fits on intervals $[\underline{x},s_{0}]$ and $[s_{0},\overline{x}%
] $, where $\underline{x}$ and $\overline{x}$ are respectively lowest and
highest values of the temperature in the sample. On each interval we use $%
L^{\prime }+1=5$ uniform knots. The left-hand side figure only imposes the
continuity of the fitted curve at the switch point, whereas the right-hand
side figure imposes continuous differentiability.

\begin{figure}[!h]
\begin{center}
\vspace{0in} 
\begin{minipage}{\linewidth}
\begin{minipage}{0.5\linewidth}
% \hspace*{0.05in}
  \includegraphics[width=0.95\linewidth]{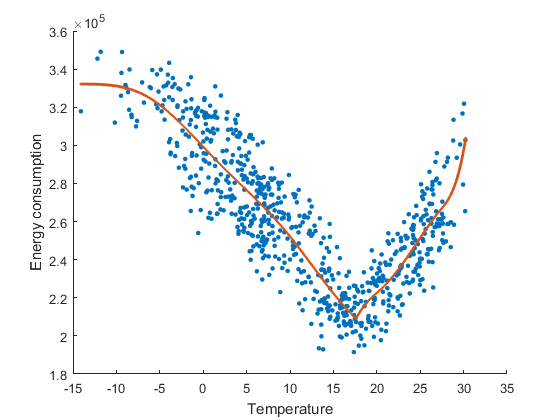}
\end{minipage}
\begin{minipage}{0.5\linewidth}
% \hspace*{0.05in}
\includegraphics[width=0.95\linewidth]{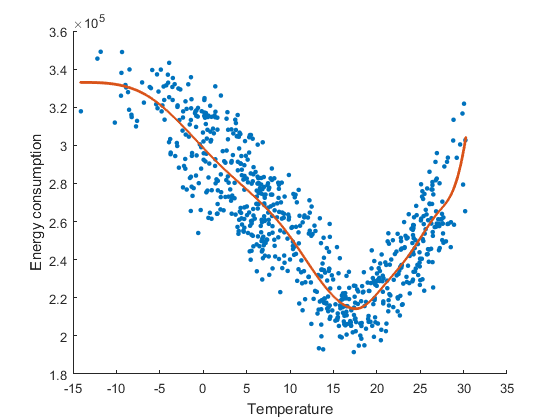}
\end{minipage}
\end{minipage}
\end{center}
\caption{{\protect\small Energy consumption data. Plot of temperature and
energy consumption and the constrained fit (under U-shape with the switch at 
${17.6}^{\circ}$) using cubic B-spline with 5 uniform knots on each
subinterval of temperature values. On the left-hand side the fitted curve is
continuous at the switch point. On the right-hand side the fitted curve is
continuously differentiable at the switch point.}}
\label{fig:SFO}
\end{figure}

Tables \ref{table:SFOA}-\ref{table:SFOC} present results of our testing.
Namely, Table \ref{table:SFOA} shows test statistics for the null hypothesis
of \emph{U-shaped} regression function and also bootstrap critical values
using both \emph{B-splines} and \emph{P-splines} in case when two \emph{%
B-spline} curves are joined at the switch point in a continuous way. Table %
\ref{table:SFOB} presents analogous results for the null hypothesis of \emph{%
U-shaped} regression function when two \emph{B-spline} curves are joined at
the switch point in a continuously differentiable way. Table \ref{table:SFOC}
gives results for the null hypothesis of convexity. In all the cases
Khmaladze's transformation is conducted from the right end of support. he
bootstrap critical values obtained on the basis of $400$ bootstrap
replications. As we can see, the null hypothesis of a \emph{U-shaped}
relationship with the switch point at ${17.6}^{\circ }$ is not rejected at
the $5\%$ level by any type of the test, whereas convexity is rejected. When
testing convexity we use cubic splines with $L^{\prime }+1=7$ uniform knots
on $[\underline{x},\overline{x}]$.

\begin{table}[tbp]
%[!h]
\par
\begin{center}
\begin{tabular}{ccccccccc}
\hline\hline
\quad & \multicolumn{4}{c}{B-splines} & \multicolumn{4}{c}{P-splines} \\ 
\quad & Test statistic & \multicolumn{3}{c}{Bootstrap critical value} & Test
statistic & \multicolumn{3}{c}{Bootstrap critical value} \\ 
Method & \quad & 10\% & 5\% & 1\% & \quad & 10\% & 5\% & 1\% \\ \hline
KS & 0.6204 & 1.0997 & 1.2139 & 1.4053 & 0.2482 & 0.5898 & 0.6199 & 0.7207 \\ 
CvM & 0.0698 & 0.3098 & 0.4048 & 0.5885 & 0.0038 & 0.0434 & 0.0498 & 0.0621  \\ 
AD & 0.5119 & 1.6541 & 2.1098 & 2.925 & 0.0647 & 0.397 & 0.4312 & 0.5324 \\ \hline
\end{tabular}
%\vspace*{0.05in}
\end{center}
\par
%\par
\caption{{\protect\small Energy consumption data. Test statistics and
bootstrap critical values under the null hypothesis of U-shaped regression
function with the switch as ${17.6}^{\circ}$. Two B-spline curves are joined
continuously at the switch point. Bootstrap critical values are from 400
bootstrap replications.}}
\label{table:SFOA}
\end{table}

\begin{table}[tbp]
%[htp]
\par
\begin{center}
\begin{tabular}{ccccccccc}
\hline\hline
\quad & \multicolumn{4}{c}{B-splines} & \multicolumn{4}{c}{P-splines} \\ 
\quad & Test statistic & \multicolumn{3}{c}{Bootstrap critical value} & Test
statistic & \multicolumn{3}{c}{Bootstrap critical value} \\ 
Method & \quad & 10\% & 5\% & 1\% & \quad & 10\% & 5\% & 1\% \\ \hline
KS & 0.8481 & 1.023 & 1.1749 & 1.407 & 0.505 & 0.6143 & 0.6524 & 0.7747 \\ 
CvM & 0.1472 & 0.2211 & 0.3247 & 0.5553 & 0.0469 & 0.0476 & 0.0542 & 0.0752  \\ 
AD & 0.9114 & 1.3676 & 1.8002 & 2.9199 & 0.2878 & 0.3713 & 0.4151 & 0.5547 \\ \hline
\end{tabular}
%\vspace*{0.05in}
\end{center}
\par
%\par
\caption{{\protect\small Energy consumption data. Test statistics and
bootstrap critical values under the null hypothesis of U-shaped regression
function with the switch as ${17.6}^{\circ}$. Two B-spline curves are joined
smoothly at the switch point. Bootstrap critical values are from 400
bootstrap replications.}}
\label{table:SFOB}
\end{table}

\begin{table}[tbp]
%[!h]
\par
\begin{center}
\begin{tabular}{ccccccccc}
\hline\hline
\quad & \multicolumn{4}{c}{B-splines} & \multicolumn{4}{c}{P-splines} \\ 
\quad & Test statistic & \multicolumn{3}{c}{Bootstrap critical value} & Test
statistic & \multicolumn{3}{c}{Bootstrap critical value} \\ 
Method & \quad & 10\% & 5\% & 1\% & \quad & 10\% & 5\% & 1\% \\ \hline
KS & 2.9812 & 1.1537 & 1.2652 & 1.5579 & 3.4626 & 0.5478 & 0.6891 & 1.0121 \\ 
CvM & 2.7713 & 0.3327 & 0.4389 & 0.71 & 3.2622 & 0.0402 & 0.0716 & 0.2076  \\ 
AD & 14.361 & 1.823 & 2.3809 & 3.856 & 17.23 & 0.3332 & 0.484 & 1.2158 \\ \hline
\end{tabular}
%\vspace*{0.05in}
\end{center}
\par
%\par
\caption{{\protect\small Energy consumption data. Test statistics and
bootstrap critical values under the null hypothesis of convexity of the
regression function. Bootstrap critical values are from 400 bootstrap
replications.}}
\label{table:SFOC}
\end{table}

\section{\textbf{CONCLUSION}}

\label{sec:conclusion}

$\left. {}\right.$

This paper proposes a methodology for testing a wide range of shape
properties of a regression function. The methodology relies on applying
the Khmaladze's transformation to the partial sums empirical process in a
nonparametric setting where \emph{B-splines} have been used to approximate
the functional space under the null hypothesis. We establish that the
proposed Khmaladze's transformation eliminates the effect of nonparametric
estimation and results in asymptotically pivotal testing. To the best of our
knowledge, this paper is the first implementation of the Khmaladze's
transformation in a nonparametric setting.

In our main examples we considered shape constraints that can be written as
inequality constraints on the coefficients of the approximating regression
splines. The generality of our procedure allows to test several shape
properties simultaneously. The implementation is especially easy when the
inequality constraints are linear, which is the case for shape properties
expressed as linear inequality constraints on the derivatives.

\section{\textbf{APPENDIX A}}

For a $M\times M$ matrix $C=\left\{ c_{kj}\right\} _{k,j=1}^{M}$, $%
\left\Vert C\right\Vert _{E}=\sum_{k,j=1}^{M}c_{kj}^{2}$ denotes the
Euclidean norm and for a $M\times N$ matrix we define the spectral norm $%
\left\Vert G\right\Vert $ as $\overline{\lambda }^{1/2}\left( G^{\prime
}G\right) $, where $\overline{\lambda }\left( H\right) $ denotes the maximum
eigenvalue of the matrix $H$. It is worth recalling the following inequality 
$\left\Vert CG\right\Vert _{E}\leq \left\Vert C\right\Vert _{E}\left\Vert
G\right\Vert $. Finally $K$ denotes a generic finite and positive constant.

We now introduce the following notation. We shall denote the fourth cumulant
of a random variable $z$ by $\kappa _{4}\left( z\right) $ and for any $x_{i}$%
, $i=1,...,n$, $\mathcal{I}\left( x^{1}<x_{i}<x^{2}\right) =:\mathcal{I}%
_{i}\left( x^{1},x^{2}\right) $. Also we define 
\begin{eqnarray}
A_{L,i} &=&A_{L}\left( \widetilde{x}_{i}\right) \text{, \ \ \ \ }\widetilde{%
\boldsymbol{P}}_{i}=A_{L,i}^{+1/2}\boldsymbol{P}_{i}\text{; }\ i=1,...,n
\label{q_1} \\
\overline{A}_{L}\left( x\right) &=&D_{L}\left( x\right) A_{L}\left( x\right)
D_{L}\left( x\right) \text{; \ \ }\overline{A}_{L,i}=:\overline{A}_{L}\left( 
\widetilde{x}_{i}\right) \text{,}  \notag
\end{eqnarray}%
where $A_{L}\left( x\right) $ was given in $\left( \ref{L_x}\right) $ and%
\begin{equation}
D_{L}\left( x\right) =diag\left( d_{1}\left( x\right) ,...,d_{L}\left(
x\right) \right) \text{, \ }d_{\ell }\left( x\right) =\left\{ 
\begin{array}{c}
0\text{ \ \ \ \ \ \ \ \ \ \ \ \ \ \ \ \ \ \ \ \ if \ \ \ \ \ \ }x<z^{\ell -1}
\\ 
L^{-q}\left( z^{\ell }-x\right) ^{-q-1/2}\ \ \ \ \text{if \ \ \ \ \ }z^{\ell
-1}\leq x<z^{\ell } \\ 
L^{1/2}\text{ \ \ \ \ \ \ \ \ \ \ \ \ \ \ if \ \ \ \ \ \ \ }z^{\ell }\leq x%
\text{.}%
\end{array}%
\right.  \label{q_11}
\end{equation}

Observe that when $x=:z^{k}$, that is a knot, Condition $C1$ yields that%
\begin{equation*}
\overline{A}_{L}\left( z^{k}\right) =:L~A_{L}\left( z^{k}\right) =diag\left( 
\underline{0},B_{L}\left( z^{k}\right) \right) \text{,}
\end{equation*}%
where $\underline{0}$ is a square $k-1$ matrix of zeroes and the $L-k+1$
matrix $B_{L}\left( z^{k}\right) $ is positive definite, where the elements $%
B_{L,\ell _{1},\ell _{2}}\left( z^{k}\right) $ of the matrix $B_{L}\left(
z^{k}\right) $ are zero if $\left\vert \ell _{1}-\ell _{2}\right\vert >q$.
The latter follows because there are only $q$ splines different than zero at
a given value $x$. Finally, it is worth mentioning that for $x\in \left(
z^{\ell -1},z^{\ell }\right) $, $d_{\ell }\left( x\right) =Kp_{\ell
,L}^{-1}\left( x;q\right) \left( z^{\ell }-x\right) ^{-1/2}$ and recalling
that Harville's \cite{Harville} Section 20.2 yields that $\overline{A}%
_{L}^{+}\left( z^{k}\right) =diag\left( \underline{0},B_{L}^{-1}\left(
z^{k}\right) \right) $, where we abbreviate $p_{\ell ,L}\left( x;q\right) $
by $p_{\ell ,L}\left( x\right) $ in what follows.

\subsection{\textbf{PROOF OF THEOREM \protect\ref{M_n}}}

$\left. {}\right.$

We need to show $\left( \mathbf{a}\right) $ the finite dimensional
distributions converge to a Gaussian random variable with covariance
structure given by that of $\mathcal{U}\left( x\right) $ and $\left( \mathbf{%
b}\right) $ the tightness of the sequence. We begin the proof of part $%
\left( \mathbf{a}\right) $ showing the structure of the covariance structure
of $\mathcal{M}_{n}\left( x\right) $. That is, for any $0\leq x^{1}\leq
x^{2}\leq 1$, 
\begin{eqnarray}
E\left( \mathcal{M}_{n}\left( x^{1}\right) \mathcal{M}_{n}\left(
x^{2}\right) \mid X\right) &=&\frac{1}{n}\sum_{i,j=1}^{n}E\left(
v_{i}v_{j}\mid X\right) \mathcal{I}_{i}\left( x^{1}\right) \mathcal{I}%
_{j}\left( x^{2}\right)  \label{pth1_1} \\
&&\overset{P}{\rightarrow }\sigma _{u}^{2}F_{X}\left( x^{1}\right) \text{.} 
\notag
\end{eqnarray}

Consider $i<j$ first and assume, without loss of generality, that $%
x_{i}<x_{j}$. When $x_{i}+n^{-\varsigma }\leq z^{k\left( x_{i}\right) }$,
that is $\widetilde{x}_{i}=x_{i}$, Condition $C1$ implies that%
\begin{eqnarray}
E\left( C_{n,j}u_{i}\mid X\right) &=&0;\text{ }E\left( C_{n,i}u_{j}\mid
X\right) =\frac{\sigma _{u}^{2}}{n}\boldsymbol{P}_{j}  \notag \\
E\left( C_{n,i}C_{n,j}^{\prime }\mid X\right) &=&\frac{\sigma _{u}^{2}}{n}%
A_{n,j}\text{,}  \label{th_1_a}
\end{eqnarray}%
and hence we obtain that%
\begin{equation}
E\left( v_{i}v_{j}\mid X\right) =\frac{\sigma _{u}^{2}}{n}\left\{ 
\boldsymbol{P}_{i}^{\prime }A_{n,i}^{+}A_{n,j}A_{n,j}^{+}\boldsymbol{P}_{j}-%
\boldsymbol{P}_{i}^{\prime }A_{n,i}^{+}\boldsymbol{P}_{j}\right\} =0
\label{th_1_b}
\end{equation}%
because Harville's \cite{Harville} Theorem 12.3.4 yields that $%
A_{n,j}A_{n,j}^{+}\boldsymbol{P}_{j}=\boldsymbol{P}_{j}$. On the other hand
when $x_{i}+n^{-\varsigma }>z^{k\left( x_{i}\right) }$, that is $\widetilde{x%
}_{i}\not=x_{i}$ and $x_{j}>z^{k\left( x_{i}\right) }$, $\left( \ref{th_1_a}%
\right) $ holds true which implies $\left( \ref{th_1_b}\right) $. Finally
when $z^{k\left( x_{i}\right) }-n^{-\varsigma }<x_{i},x_{j}\leq z^{k\left(
x_{i}\right) }$, we obtain that 
\begin{eqnarray*}
E\left( v_{i}v_{j}\mid X\right) &=&\frac{\sigma _{u}^{2}}{n}\boldsymbol{P}%
_{i}^{\prime }A_{n}^{+}\left( z^{k\left( x_{i}\right) }\right) \boldsymbol{P}%
_{j}\mathcal{I}\left( z^{k\left( x_{i}\right) }-\frac{1}{n^{\varsigma }}%
<x_{i},x_{j}\leq z^{k\left( x_{i}\right) }\right) \\
&=&:\vartheta _{n}\left( i,j;X\right) \text{.}
\end{eqnarray*}%
Observe that in this case we have 
\begin{equation}
E\left\Vert E\left( v_{i}v_{j}\mid X\right) \right\Vert =O\left(
Ln^{-1-2\varsigma }\right)  \label{e_v_i_vj}
\end{equation}%
because Lemmas \ref{Est_A_i} and \ref{Est_A_uni} imply that 
\begin{equation}
\underline{\lambda }\left( \inf_{x\in \left[ 0,1\right] }A_{L}^{+1/2}\left(
x\right) A_{n}\left( x\right) A_{L}^{+1/2}\left( x\right) \right) >\delta >0
\label{A_inv}
\end{equation}%
with probability approaching one, where \underline{$\lambda $}$\left(
G\right) $ denotes the minimum eigenvalue of the matrix $G$, Lemma \ref%
{Nor_P_i} implies that $E\left\Vert \widetilde{\boldsymbol{P}}%
_{i}\right\Vert ^{2}=O\left( L\right) $ and 
\begin{equation}
E\left( \mathcal{I}\left( z^{k\left( x_{i}\right) }-n^{-\varsigma
}<x_{i}\leq z^{k\left( x_{i}\right) }\right) \right) =O\left( n^{-\varsigma
}\right) \text{.}  \label{e_ind}
\end{equation}

Next, when $i=j$, proceeding as above we obtain that 
\begin{equation*}
E\left( v_{i}v_{j}\mid X\right) =\sigma _{u}^{2}\left( 1-\frac{\boldsymbol{P}%
_{i}^{\prime }A_{n,i}^{+}\boldsymbol{P}_{i}}{n}\right) \text{.}
\end{equation*}%
Observe that, denoting $A_{n,i}^{\ast }=A_{n,i}-\boldsymbol{P}_{i}%
\boldsymbol{P}_{i}^{\prime }/n$, we have that 
\begin{equation*}
1-\frac{\boldsymbol{P}_{i}^{\prime }A_{n,i}^{+}\boldsymbol{P}_{i}}{n}=1-2%
\frac{\boldsymbol{P}_{i}^{\prime }A_{n,i}^{+}\boldsymbol{P}_{i}}{n}+\left( 
\frac{\boldsymbol{P}_{i}^{\prime }A_{n,i}^{+}\boldsymbol{P}_{i}}{n}\right)
^{2}+\frac{\boldsymbol{P}_{i}^{\prime }A_{n,i}^{+}A_{n,i}^{\ast }A_{n,i}^{+}%
\boldsymbol{P}_{i}}{n}>0\text{.}
\end{equation*}

So, we conclude that the left side of $\left( \ref{pth1_1}\right) $ is, for
any $0\leq x^{1}\leq x^{2}\leq 1$, 
\begin{equation}
\frac{\sigma _{u}^{2}}{n}\sum_{i=1}^{n-1}\left( 1-\frac{\boldsymbol{P}%
_{i}^{\prime }A_{n,i}^{+}\boldsymbol{P}_{i}}{n}\right) \mathcal{I}_{i}\left(
x^{1}\right) +\frac{\sigma _{u}^{2}}{n}\sum_{i\neq j}^{n-1}\vartheta
_{n}\left( i,j;X\right) \mathcal{I}_{i}\left( x^{1}\right) \text{.}
\label{Theo1_1}
\end{equation}

The second term of $\left( \ref{Theo1_1}\right) $ is $o_{p}\left( 1\right) $
because $\left( \ref{e_v_i_vj}\right) $ implies that its first absolute
moment is $O\left( L/n^{2\varsigma }\right) =o\left( 1\right) $ since $%
\varsigma >1/2$ and Condition $C3$. The first term of $\left( \ref{Theo1_1}%
\right) $ is also $o_{p}\left( 1\right) $ as we now show. Using the
inequality $\mathcal{I}_{i}\left( x^{1}\right) \leq \sum_{k=1}^{k\left(
x^{1}\right) -1}\mathcal{I}_{i}\left( z^{k};z^{k+1}\right) $, we obtain that 
\begin{eqnarray}
&&E\left( \frac{1}{n^{2}}\sum_{i=1}^{n-1}\widetilde{\boldsymbol{P}}%
_{i}^{\prime }\left( A_{L,i}^{+1/2}A_{n,i}A_{L,i}^{+1/2}\right) ^{+}%
\widetilde{\boldsymbol{P}}_{i}\mathcal{I}_{i}\left( x^{1}\right) \right) 
\notag \\
&\leq &\sum_{k=1}^{k\left( x^{1}\right) -1}E\left( \frac{1}{n^{2}}%
\sum_{i=1}^{n-1}\widetilde{\boldsymbol{P}}_{i}^{\prime }\left(
A_{L,i}^{+1/2}A_{n,i}A_{L,i}^{+1/2}\right) ^{+}\widetilde{\boldsymbol{P}}_{i}%
\mathcal{I}_{i}\left( z^{k};z^{k+1}\right) \right)  \notag \\
&=&\frac{K}{n^{2}}\sum_{k=1}^{k\left( x^{1}\right)
-1}\sum_{i=1}^{n}E\left\Vert \widetilde{\boldsymbol{P}}_{i}\mathcal{I}%
_{i}\left( z^{k};z^{k+1}\right) \right\Vert ^{2}=O\left( \frac{L}{n}\right) 
\text{,}  \label{eg}
\end{eqnarray}%
by $\left( \ref{A_inv}\right) $ and because $E\left\Vert \widetilde{%
\boldsymbol{P}}_{i}\mathcal{I}_{i}\left( z^{k};z^{k+1}\right) \right\Vert
^{2}=O\left( 1\right) $ by Lemma \ref{Nor_P_i}, cf. $\left( \ref{Lemma2_4}%
\right) $. So, we conclude that $\left( \ref{pth1_1}\right) $ holds true as
it is well known that, under Conditions $C1$ and then $C3$, 
\begin{equation*}
\sigma _{u}^{2}\frac{1}{n}\sum_{i=1}^{n}\mathcal{I}_{i}\left( x\right)
\left( 1+\frac{L}{n}\right) \overset{P}{\rightarrow }\sigma
_{u}^{2}F_{X}\left( x\right) \text{.}
\end{equation*}

To complete part $\left( \mathbf{a}\right) $, it suffices to show the
asymptotic Gaussianity of $\mathcal{M}_{n}\left( x\right) $ for a fixed $x$
due to Cram\'{e}r-Wold device. First, we have already shown that%
\begin{equation*}
\sum_{i=1}^{n}E\left( \upsilon _{in}^{2}\left( x\right) \mid X\right) 
\overset{P}{\rightarrow }1\text{,}
\end{equation*}%
where by construction $\upsilon _{in}\left( x\right) =:\sigma
_{u}^{-1}F_{X}^{-1/2}\left( x\right) \mathcal{I}_{i}\left( x\right)
v_{i}/n^{1/2}$ is a martingale difference triangular array of r.v.'s. So, it
remains to show the Lindeberg's condition for which a sufficient condition
is 
\begin{equation}
\sum_{i=1}^{n}E\left\vert \upsilon _{in}\left( x\right) \right\vert
^{4}=o\left( 1\right) \text{.}  \label{v_in}
\end{equation}%
But the latter holds true because $\left( \ref{A_inv}\right) $ yields that 
\begin{equation*}
E\left\Vert \boldsymbol{P}_{i}^{\prime }A_{n,i}^{+}C_{n,i}\mathcal{I}%
_{i}\left( x\right) \right\Vert ^{4}\leq E\left\Vert \widetilde{\boldsymbol{P%
}}_{i}\right\Vert ^{4}E\left\Vert \frac{1}{n}\sum_{j=1}^{n}\widetilde{%
\boldsymbol{P}}_{j}u_{j}\mathcal{J}_{j}\left( \widetilde{x}_{i}\right)
\right\Vert ^{4}=O\left( \frac{L^{4}}{n^{2}}\right) \text{,}
\end{equation*}%
since Lemma \ref{Nor_P_i} implies that $E\left\Vert \widetilde{\boldsymbol{P}%
}_{i}\right\Vert ^{4}=O\left( L^{2}\right) $. So, $\left( \ref{v_in}\right) $
holds true by Condition $C3$, which concludes the proof of part $\left( 
\mathbf{a}\right) $.

$\left( \mathbf{b}\right) $ Tightness follows by Lemmas \ref{Est_A_uni} and %
\ref{tight} and using the same arguments as those for the proof of
Billingsley's \cite{Billingsley} Theorem 22.1, see also Wu's \cite{Wu} Lemma 14, which completes the proof of the theorem. Observe
that, for any $\delta >0$, $\delta \sum_{k\in \mathbb{Z}}\sup_{x\in \left(
k\delta ,\left( k+1\right) \delta \right) }f_{X}\left( x\right) <K$
following Lemma 4 of Wu \cite{Wu}.

\subsection{PROOF OF THEOREM \protect\ref{M_nest}}

$\left. {}\right. $

We shall show that, uniformly in $x\in \left[ 0,1\right] $, 
\begin{equation}
\widetilde{\mathcal{M}}_{n}\left( x\right) -\mathcal{M}_{n}\left( x\right)
=o_{p}\left( 1\right) \text{.}  \label{m_dif}
\end{equation}%
Because $\left( \ref{res_con}\right) $ yields that $\widehat{u}%
_{i}-u_{i}=m^{bias}\left( x_{i}\right) -\boldsymbol{P}_{i}^{\prime }\left( 
\widehat{b}-\beta \right) $, we have that the left side of $\left( \ref%
{m_dif}\right) $ is%
\begin{equation}
\frac{1}{n^{1/2}}\sum_{i=1}^{n}m^{bias}\left( x_{i}\right) \mathcal{I}%
_{i}\left( x\right) -\frac{1}{n^{1/2}}\sum_{i=1}^{n}\left\{ \frac{%
\boldsymbol{P}_{i}^{\prime }A_{n,i}^{+}}{n}\sum_{k=1}^{n}\boldsymbol{P}%
_{k}m^{bias}\left( x_{k}\right) \mathcal{J}_{k}\left( \widetilde{x}%
_{i}\right) \right\} \mathcal{I}_{i}\left( x\right)  \label{m_dif2}
\end{equation}%
since the contribution due to $\boldsymbol{P}_{i}^{\prime }\left( \widehat{b}%
-\beta \right) $ is zero by Lemma \ref{bias}.

Now the first term of $\left( \ref{m_dif2}\right) $ is $o_{p}\left( 1\right) 
$ by Condition $C3$ and that Agarwal and Studden's \cite{AgarwalStudden} 
Theorems 3.1 and 4.1, see also Zhou et al. $\left( 1998\right) $, yields
that $Em^{bias}\left( x_{i}\right) =O\left( L^{-3}\right) $, whereas
together with Lemma \ref{Est_A_uni} and $\left( \ref{A_inv}\right) $, we
obtain that the second term is, with probability approaching $1$, bounded by 
\begin{equation*}
\frac{K}{L^{3}n^{1/2}}\sum_{i=1}^{n}\left\{ \left\Vert \widetilde{%
\boldsymbol{P}}_{i}\right\Vert \frac{1}{n}\sum_{k=1}^{n}\left\Vert 
\widetilde{\boldsymbol{P}}_{k}\right\Vert \right\} =O_{p}\left(
n^{1/2}/L^{2}\right) =o\left( 1\right)
\end{equation*}%
by Markov's inequality and Lemma \ref{Nor_P_i}\ and then Condition $C3$.
This completes the proof of the theorem.\hfill $\blacksquare $

\subsection{PROOF OF PROPOSITION \protect\ref{sigma_est}}

$\left. {}\right. $

First, using $\left( \ref{res_con}\right) $, we have that 
\begin{eqnarray*}
\widehat{u}_{i}^{2} &=&u_{i}^{2}+m^{bias}\left( x_{i}\right) ^{2}+\left( 
\widehat{b}-\beta \right) ^{\prime }\boldsymbol{P}_{i}\boldsymbol{P}%
_{i}^{\prime }\left( \widehat{b}-\beta \right) \\
&&+2u_{i}m^{bias}\left( x_{i}\right) +2u_{i}\boldsymbol{P}_{i}^{\prime
}\left( \widehat{b}-\beta \right) +2m^{bias}\left( x_{i}\right) \boldsymbol{P%
}_{i}^{\prime }\left( \widehat{b}-\beta \right) \text{.}
\end{eqnarray*}%
So, by Cauchy-Schwarz's inequality, it suffices to show that 
\begin{equation}
\left( \widehat{b}-\beta \right) ^{\prime }\frac{1}{n}\sum_{i=1}^{n}%
\boldsymbol{P}_{i}\boldsymbol{P}_{i}^{\prime }\left( \widehat{b}-\beta
\right) \text{ \ and }\frac{1}{n}\sum_{i=1}^{n}m^{bias}\left( x_{i}\right)
^{2}  \label{a1}
\end{equation}%
are$\ o_{p}\left( 1\right) $ because $n^{-1}\sum_{i=1}^{n}u_{i}^{2}-\sigma
_{u}^{2}=o_{p}\left( 1\right) $. Now, the norm of the first expression in $%
\left( \ref{a1}\right) $ is $O_{p}\left( L^{1/2}/n\right) =o_{p}\left(
1\right) $ because 
\begin{equation*}
E\left( \left( A_{L}\left( 0\right) ^{-1/2}\frac{1}{n}\sum_{i=1}^{n}\mathbf{P%
}_{i}u_{i}\right) \left( A_{L}\left( 0\right) ^{-1/2}\frac{1}{n}%
\sum_{i=1}^{n}\mathbf{P}_{i}u_{i}\right) ^{\prime }\mid X\right) =\frac{%
\sigma _{u}^{2}}{n}I_{L}
\end{equation*}%
and $\left\Vert I_{L}\right\Vert _{E}^{2}=L$. Note that by Lemma \ref%
{Est_A_i}, $A_{n}\left( 0\right) $ is invertible with probability
approaching $1$ since $A_{L}\left( 0\right) $ is. On the other hand, the
second expression in $\left( \ref{a1}\right) $ is $o_{p}\left( 1\right) $ by
Agarwal and Studden's \cite{AgarwalStudden} Theorems 3.1 and 4.1 and then
Condition $C3$. This completes the proof of the proposition.\hfill $%
\blacksquare $

\subsection{PROOF OF THEOREM \protect\ref{M_nest-nonlinear}}

$\left. {} \right.$

In this proof, for the sake of notation simplicity, we will write $\beta$
and $\widehat{b}$ instead of $\beta_{-\ell_0}$ and $\widehat{b}%
_{-\ell_0}$, respectively.

Let 
\begin{equation*}
G\left( x;\beta \right) =:\frac{\partial \widetilde{P}_{i}\left( x;\beta
\right) }{\partial \beta ^{\prime }}\quad \text{and}\quad G_{i}\left( \beta
\right) =:G_{i}\left( x_{i};\beta \right) .
\end{equation*}%
From the definition of $\widetilde{P}_{i}\left( x;\beta \right) $, it follows that $G\left( x;\beta \right) =H(\beta )p_{\ell _{0}}(x)$, where $%
H(\beta )$ is the $(L-1)\times (L-1)$ Hessian of $h(\beta )$, whose norm is
finite.

Then, employing the form of our model and the approximation, by Taylor's approximation to $g(\cdot)$, we write 
\begin{eqnarray*}
\widehat{u}_{i} &=&u_{i}-\left( g\left( x_{i};\widehat{b}\right)
-g\left( x_{i};\beta \right) \right) +(m(x_{i})-g\left( x_{i};\beta \right) )
\\
&=&u_{i}-\left( \widehat{b}-\beta \right) ^{\prime }\widetilde{P}%
_{i}\left( \beta \right) +\frac{1}{2}\left( \widehat{b}-\beta \right)
^{\prime }G_{i}\left( \widetilde{b}\right) \left( \widehat{b}%
-\beta \right) +\left( m(x_{i})-g\left( x_{i};{\beta }\right) \right) \text{,
}
\end{eqnarray*}
where $\widetilde{b}$ is an intermediate point between $\beta$ and $\widehat{b}$. 

So, it suffices to examine the contribution due to the following four terms: 
\begin{eqnarray*}
\mathcal{A}_{1,n}\left( x;\widehat{b}\right) &=&:\frac{1}{n^{1/2}}%
\sum_{i=1}^{n}\left\{ u_{i}-\widetilde{P}_{i}^{\prime }\left( \widehat{b 
}\right) \mathcal{D}_{n}^{+}\left( \widehat{b};i\right) \sum_{k=1}^{n}%
\widetilde{P}_{k}\left( \widehat{b}\right) u_{k}\mathcal{J}_{k}\left( 
\widetilde{x}_{i}\right) \right\} \mathcal{I}_{i}(x) \\
\mathcal{A}_{2,n}\left( x;\widehat{b}\right) &=&:\frac{1}{n^{1/2}}%
\sum_{i=1}^{n}\left\{ \widetilde{P}_{i}^{\prime }\left( \beta \right) -%
\widetilde{P}_{i}^{\prime }\left( \widehat{b}\right) \mathcal{D}%
_{n}^{+}\left( \widehat{b};i\right) \sum_{k=1}^{n}\widetilde{P}%
_{k}\left( \widehat{b}\right) \widetilde{P}_{k}^{\prime }\left( \beta
\right) \mathcal{J}_{k}\left( \widetilde{x}_{i}\right) \right\} \mathcal{I}%
_{i}(x)\left( \widehat{b}-\beta \right) \\
\mathcal{A}_{3,n}\left( x;\widehat{b}\right) &=&:\left( \widehat{b}%
-\beta \right) ^{\prime }\frac{1}{n^{1/2}}\sum_{i=1}^{n}\mathcal{I}%
_{i}(x)G_{i}\left( \widetilde{b}\right) \left( \widehat{b}-\beta
\right) \\
&&-\frac{1}{n^{1/2}}\sum_{i=1}^{n}\widetilde{P}_{i}^{\prime }\left( \widehat{b }\right) \mathcal{D}_{n}^{+}\left( \widehat{b};i\right) \frac{1}{n%
}\left\{ \sum_{k=1}^{n}\widetilde{P}_{k}\left( \widehat{b}\right)
\left( \widehat{b}-\beta \right) ^{\prime }G_{k}\left( \widetilde{b
}\right) \left( \widehat{b}-\beta \right) \mathcal{J}_{k}\left( 
\widetilde{x}_{i}\right) \right\} \mathcal{I}_{i}(x)\text{,} \\
\mathcal{A}_{4,n}\left( \widehat{b}\right) &=&:\frac{1}{n^{1/2}}%
\sum_{i=1}^{n}\left\{ m^{bias}\left( x_{i}\right) -\widetilde{P}_{i}^{\prime
}\left( \widehat{b}\right) \mathcal{D}_{n}^{+}\left( \widehat{b}%
;i\right) \frac{1}{n}\sum_{k=1}^{n}\left( m^{bias}\left( x_{k}\right)
\right) \mathcal{J}_{k}\left( \widetilde{x}_{i}\right) \right\} \mathcal{I}%
_{i}(x),
\end{eqnarray*}%
with $m(x_{i})-g\left( x_{i};{\beta }\right) =:m^{bias}\left( x_{i}\right) $. Also introduce the following: 
\begin{eqnarray*}
\mathcal{A}_{1,n}\left( x;{\beta }\right) &=&:\frac{1}{n^{1/2}}%
\sum_{i=1}^{n}\left\{ u_{i}-\widetilde{P}_{i}^{\prime }\left( {\beta }%
\right) \mathcal{D}_{n}^{+}\left( {\beta };i\right) \sum_{k=1}^{n}\widetilde{%
P}_{k}\left( {\beta }\right) u_{k}\mathcal{J}_{k}\left( \widetilde{x}%
_{i}\right) \right\} \mathcal{I}_{i}(x), \\
\mathcal{A}_{2,n}\left( x;\beta \right) &=&:\frac{1}{n^{1/2}}%
\sum_{i=1}^{n}\left\{ \widetilde{P}_{i}^{\prime }\left( \beta \right) -%
\widetilde{P}_{i}^{\prime }\left( {\beta }\right) \mathcal{D}_{n}^{+}\left( {%
\beta };i\right) \sum_{k=1}^{n}\widetilde{P}_{k}\left( {\beta }\right) 
\widetilde{P}_{k}^{\prime }\left( \beta \right) \mathcal{J}_{k}\left( 
\widetilde{x}_{i}\right) \right\} \mathcal{I}_{i}(x)\left( \widehat{b}%
-\beta \right) , \\
\mathcal{A}_{3,n}\left( x;{\beta }\right) &=&:\left( \widehat{b}-\beta
\right) ^{\prime }\frac{1}{n^{1/2}}\sum_{i=1}^{n}\mathcal{I}%
_{i}(x)G_{i}\left( \widetilde{b}\right) \left( \widehat{b}-\beta
\right) \\
&&-\frac{1}{n^{1/2}}\sum_{i=1}^{n}\widetilde{P}_{i}^{\prime }\left( {\beta }%
\right) \mathcal{D}_{n}^{+}\left( {\beta };i\right) \frac{1}{n}\left\{
\sum_{k=1}^{n}\widetilde{P}_{k}\left( {\beta }\right) \left( \widehat{b}%
-\beta \right) ^{\prime }G_{k}\left( \widetilde{b}\right) \left( 
\widehat{b}-\beta \right) \mathcal{J}_{k}\left( \widetilde{x}%
_{i}\right) \right\} \mathcal{I}_{i}(x).
\end{eqnarray*}%
We shall first examine the behaviour of $\mathcal{A}_{q,n}\left( x;\beta
\right) $, $q=1,2,3$, and then we shall show that, uniformly in $x$, $%
\mathcal{A}_{q,n}\left( x;\widehat{b}\right) -\mathcal{A}_{q,n}\left(
x;\beta \right) =o_{p}\left( 1\right) $.

Proceeding as with the proof of Theorem \ref{M_n}, we can establish 
\begin{equation*}
\mathcal{A}_{1,n}\left( \beta \right) \overset{weakly}{\Rightarrow }\mathcal{%
U}\left( x\right) .
\end{equation*}%
It is also obvious that $\mathcal{A}_{2,n}\left( x;\beta \right) =0$. Now we
want to show that 
\begin{equation*}
\sup_{x}\left\vert \mathcal{A}_{3,n}\left( x;\beta \right) \right\vert
=o_{p}\left( 1\right) .
\end{equation*}%
To that end, we shall look at%
\begin{eqnarray*}
\mathcal{C}_{1,n}\left( x;\beta \right) &=&:\left( \widehat{b}-\beta
\right) ^{\prime }\left\{ \frac{1}{n^{1/2}}\sum_{i=1}^{n}\mathcal{I}%
_{i}(x)G_{i}\left( \widetilde{b}\right) \right\} \left( \widehat{b}%
-\beta \right) \\
\mathcal{C}_{2,n}\left( x;\beta \right) &=&:\frac{1}{n^{1/2}}\sum_{i=1}^{n}%
\widetilde{P}_{i}^{\prime }\left( \beta \right) \mathcal{D}_{n}^{+}\left(
\beta ;i\right) \left\{ \left( \sum_{k=1}^{n}\widetilde{P}_{k}\left( \beta
\right) \left( \widehat{b}-\beta \right) ^{\prime }G_{k}\left( 
\widetilde{b}\right) \left( \widehat{b}-\beta \right) \mathcal{J}%
_{k}\left( \widetilde{x}_{i}\right) \right) \right\} \mathcal{I}_{i}(x)\text{%
.}
\end{eqnarray*}%
We already know that $\widehat{b}-\beta =O_{p}\left( \left( L/n\right)
^{1/2}\right) $. So, by definition of $G\left( x;\beta \right) $, we obtain
that 
\begin{eqnarray*}
\left\Vert \mathcal{C}_{1,n}\left( x;\beta \right) \right\Vert &=&\left\Vert
H\left( \widetilde{b}\right) \right\Vert \left\vert \frac{1}{n}%
\sum_{i=1}^{n}Lp_{\ell _{0}}\left( x_{i}\right) \mathcal{I}%
_{i}(x)\right\vert O_{p}\left( \frac{1}{n^{1/2}}\right) \\
&=&\left\Vert H\left( \widetilde{b}\right) \right\Vert \frac{1}{n}%
\sum_{i=1}^{n}\left\vert Lp_{\ell _{0}}\left( x_{i}\right) \right\vert
O_{p}\left( \frac{1}{n^{1/2}}\right) \\
&=&:O_{p}\left( n^{-1/2}\right) \text{,}
\end{eqnarray*}%
because $E\left\vert Lp_{\ell _{0}}\left( x_{i}\right) \right\vert \leq K$.
So, because the last displayed expression yields that%
\begin{equation*}
\sup_{x}\left\vert \mathcal{C}_{1,n}\left( x;\beta \right) \right\vert
=O_{p}\left( n^{-1/2}\right) \text{.}
\end{equation*}

Next, we examine $\mathcal{C}_{2,n}\left( x;\beta \right) $. Again, from the
definition of $G\left( x;\beta \right) $ we have that 
\begin{equation*}
\left( \widehat{b}-\beta \right) ^{\prime }G_{k}\left( \widetilde{b}\right) \left( \widehat{b}-\beta \right) =\left( \widehat{b}%
-\beta \right) ^{\prime }H\left( \widetilde{b}\right) \left( \widehat{%
\beta }-\beta \right) p_{\ell _{0}}\left( x_{k}\right) \text{.}
\end{equation*}%
From the fact that $\widehat{b}-\beta =O_{p}\left( \left( L/n\right)
^{1/2}\right) $ we can conclude that 
\begin{equation*}
\mathcal{C}_{2,n}\left( \beta \right) =C\frac{L}{n^{3/2}}\sum_{i=1}^{n}%
\widetilde{P}_{i}^{\prime }\left( \beta \right) \mathcal{D}_{n}^{+}\left(
\beta ;i\right) \left\{ \sum_{k=1}^{n}\widetilde{P}_{k}\left( \beta \right)
p_{\ell _{0}}\left( x_{k}\right) \iota _{L}^{\prime }H\left( \widetilde{%
b }\right) \iota _{L}\mathcal{J}_{k}\left( \widetilde{x}_{i}\right)
\right\} \mathcal{I}_{i}(x)\text{,}
\end{equation*}%
where $\iota _{L}=(1,\ldots ,1)^{\prime }$ and $C$ is some positive constant.

We can establish that $\underline{\lambda }\left( \frac{L}{n}\mathcal{D}%
_{n}\left( \beta ;i\right) \right) >0$ as this is what we have essentially
in expression $\left( \ref{A_inv}\right) $ in the Appendix of the paper.
Also note that 
\begin{eqnarray*}
E\left( L^{1/2}\widetilde{P}_{i}^{\prime }\left( \beta \right) \right)
&\simeq &L^{-1/2}\left( 1,...,1\right) ; \\
E\left( L^{1/2}\widetilde{P}_{k}\left( \beta \right) p_{\ell _{0}}\left(
x_{k}\right) \right) &\simeq &L^{-1/2}\left( d_{1,h},...,d_{L,h}\right)
,^{\prime }
\end{eqnarray*}%
where $d_{\ell ,h}$, $\ell \neq \ell _{0}$, is a 0/1 variable which takes
the values of 1 if and only if the function $h$ does depend on $\beta _{\ell
}$. For instance, in the case of HA-convexity, which was given as an
example, $h$ would depend only on parameters $\beta _{\ell _{0}-1}$ and $%
\beta _{\ell _{0}-2}$ and, thus, we would only have $d_{\ell _{0}-1,h}=1$
and $d_{\ell _{0}-2,h}=1$ whereas the rest of such indicators would be zero.

In this case we would have 
\begin{equation*}
\mathcal{C}_{2,n}\left( x;\beta \right) =O_{p}\left( n^{-1/2}\right)
\end{equation*}%
once we observe that $L^{-1/2}\left( 1,...,1\right) L^{-1/2}\left(
d_{1,h},...,d_{L,h}\right) ^{\prime }\simeq L^{-1}$.

If $h$ depends on parameters whose number is growing at the rate $L^{1/2}$,
then we would have 
\begin{equation*}
\mathcal{C}_{2,n}\left( x;\beta \right) =O_{p}\left( \left( L/n\right)
^{1/2}\right)
\end{equation*}%
since now $L^{-1/2}\left( 1,...,1\right) L^{-1/2}\left(
d_{1,h},...,d_{L,h}\right) ^{\prime }\simeq L^{-1/2}$

Of course, in the scenario when $h$ depends on all the other parameters, we
have that $L^{-1/2}\left( 1,...,1\right) L^{-1/2}\left(
d_{1,h},...,d_{L,h}\right) ^{\prime }\simeq 1,$which yields that 
\begin{equation*}
\mathcal{C}_{2,n}\left( x;\beta \right) =O_{p}\left( \left( L^{2}/n\right)
^{1/2}\right) \text{,}
\end{equation*}%
which by Condition $C3$, it is $o_{p}\left( 1\right) $.

Now we show that 
\begin{equation*}
\mathcal{A}_{q,n}\left( x; \widehat{b}\right) -\mathcal{A}_{q,n}\left(
x; \beta \right) =o_{p}\left( 1\right) \text{, \ \ \ }q=1,2,3\text{.}
\end{equation*}%
For that purpose, observe that by mean value theorem, we have that 
\begin{equation}
\widetilde{P}_{i}\left( \widehat{b}\right) -\widetilde{P}_{i}\left(
\beta \right) = H\left( \widetilde{b} \right) \left( \widehat{b}%
-\beta \right) p_{\ell_0}\left( x_{i}\right) = O_{p}\left( \left( Ln\right)
^{-1//2}\right)  \label{3}
\end{equation}%
because $\sup_{\beta }\left\Vert H\left( \beta \right) \right\Vert \leq K$, $%
\widehat{b} -\beta =O_{p}\left( \left( L/n\right) ^{1/2}\right) $ and $%
E\left\Vert p_{\ell_0 }\left( x_{k}\right) \right\Vert =O\left(
L^{-1}\right) $.

Next, we examine the behaviour of $\mathcal{D}_{n}^{+}\left( \widehat{b}%
;i\right) -\mathcal{D}_{n}^{+}\left( \beta ;i\right) $. By the mean value
theorem and standard algebra, we have that 
\begin{eqnarray*}
\mathcal{D}_{n}\left( \widehat{b};i\right) &=&\mathcal{D}_{n}\left(
\beta ;i\right) +2H\left( \widetilde{b}\right) \left( \widehat{b}%
-\beta \right) \sum_{k=1}^{n}\widetilde{P}_{k}\left( \beta \right) p_{\ell
_{0}}\left( x_{k}\right) \mathcal{J}_{k}\left( \widetilde{x}_{i}\right) \\
&&+H\left( \widetilde{\beta }\right) \left( \widehat{b}-\beta \right)
\left( \sum_{k=1}^{n}p_{\ell _{0}}^{2}\left( x_{k}\right) \mathcal{J}%
_{k}\left( \widetilde{x}_{i}\right) \right) \left( \widehat{b}-\beta
\right) ^{\prime }H^{\prime }\left( \widetilde{b}\right) \text{,}
\end{eqnarray*}
where $\widetilde{b}$ is an intermediate point between $\beta$ and $\widehat{b}$.

The norm of the third term on the right is $O_{p}\left( 1\right) $ whereas
the second term on the right is $O_{p}\left( \left( n/L\right) ^{1/2}\right) 
$. So, using the pseudo inverse of sum of matrices formula, it is easily
seen that 
\begin{equation}
\mathcal{D}_{n}^{+}\left( \widehat{b};i\right) -\mathcal{D}%
_{n}^{+}\left( \beta ;i\right) =O_{p}\left( \left( n/L\right) ^{1/2}\right) 
\text{.}  \label{5}
\end{equation}%
Now, we examine $\mathcal{A}_{q,n}\left( x;\widehat{b}\right) -\mathcal{%
A}_{q,n}\left( x;\beta \right) $, $q=1,2,3$. Let's examine e.g. the
behaviour $\mathcal{A}_{1,n}\left( x;\widehat{b}\right) -\mathcal{A}%
_{1,n}\left( x;\beta \right) $. By definition, it is 
\begin{eqnarray*}
\mathcal{A}_{1,n}\left( x;\widehat{b}\right) -\mathcal{A}_{1,n}\left(
x;\beta \right) &=&\frac{1}{n^{1/2}}\sum_{i=1}^{n}\left\{ \widetilde{P}%
_{i}^{\prime }\left( \widehat{b}\right) \mathcal{D}_{n}^{+}\left( 
\widehat{b};i\right) \sum_{k=1}^{n}\widetilde{P}_{k}\left( \widehat{b }\right) u_{k}\mathcal{J}_{k}\left( \widetilde{x}_{i}\right) \right. \\
&&\text{ \ \ \ \ \ \ \ \ \ \ \ \ \ \ }\left. -\widetilde{P}_{i}^{\prime
}\left( \beta \right) \mathcal{D}_{n}^{+}\left( \beta ;i\right)
\sum_{k=1}^{n}\widetilde{P}_{k}\left( \beta \right) u_{k}\mathcal{J}%
_{k}\left( \widetilde{x}_{i}\right) \right\} \mathcal{I}_{i}(x)\text{.}
\end{eqnarray*}

It suffices to show that, uniformly in $x$, 
\begin{eqnarray*}
\left( \mathbf{i}\right) \text{ \ \ }\frac{1}{n^{1/2}}\sum_{i=1}^{n}\left\{
\left( \widetilde{P}_{i}\left( \widehat{b}\right) -\widetilde{P}%
_{i}\left( \beta \right) \right) ^{\prime }\mathcal{D}_{n}^{+}\left( \beta
;i\right) \sum_{k=1}^{n}\widetilde{P}_{k}\left( \beta \right) u_{k}\mathcal{J%
}_{k}\left( \widetilde{x}_{i}\right) \right\} \mathcal{I}_{i}(x)
&=&o_{p}\left( 1\right) \\
\left( \mathbf{ii}\right) \text{ \ \ }\frac{1}{n^{1/2}}\sum_{i=1}^{n}\left\{ 
\widetilde{P}_{i}^{\prime }\left( \beta \right) \mathcal{D}_{n}^{+}\left(
\beta ;i\right) \sum_{k=1}^{n}\left( \widetilde{P}_{k}\left( \widehat{b}%
\right) -\widetilde{P}_{k}\left( \beta \right) \right) u_{k}\mathcal{J}%
_{k}\left( \widetilde{x}_{i}\right) \right\} \mathcal{I}_{i}(x)
&=&o_{p}\left( 1\right) \\
\left( \mathbf{iii}\right) \text{ \ \ }\frac{1}{n^{1/2}}\sum_{i=1}^{n}\left%
\{ \widetilde{P}_{i}^{\prime }\left( \beta \right) \left[ \mathcal{D}%
_{n}^{+}\left( \widehat{b};i\right) -\mathcal{D}_{n}^{+}\left( \beta
;i\right) \right] \sum_{k=1}^{n}\widetilde{P}_{k}\left( \beta \right) u_{k}%
\mathcal{J}_{k}\left( \widetilde{x}_{i}\right) \right\} \mathcal{I}_{i}(x)
&=&o_{p}\left( 1\right) \text{.}
\end{eqnarray*}

The proof of (i)-(iii) follows quite easily based on $\left( \ref{3}\right) $
and $\left( \ref{5}\right) $, and that we have already shown that
\begin{eqnarray*}
\mathcal{D}_{n}^{+1/2}\left( \beta ;i\right) \sum_{k=i+1}^{n}\widetilde{P}
_{k}\left( \beta \right) u_{k} &=&O_{p}\left( L^{1/2}\right) \text{.} \\
\underline{\lambda }\left( \frac{L}{n}\mathcal{D}_{n}\left( \beta ;i\right)
\right) &>&0\text{ \ with probability }1\text{.}
\end{eqnarray*}

Analogously, we can easily establish that $\mathcal{A}_{q,n}\left( x; 
\widehat{b}\right) -\mathcal{A}_{q,n}\left( x; \beta\right) = o_p(1)$. Note that $\mathcal{A}_{4,n}\left( x;\widehat{b}\right) $ can be
written as $\frac{1}{n^{1/2}}\sum_{i=1}^{n}m^{bias}(x_{i})\mathcal{I}_{i}(x)$
and in the same way as in the proof of Theorem \ref{M_nest} can be shown to
be $o_{p}(1)$. 
 $\blacksquare $

\subsection{\textbf{PROOF OF THEOREM \protect\ref{Mb_n}}}

$\left. {}\right. $

Define%
\begin{equation*}
\mathcal{M}_{n}^{\ast }\left( x\right) =:\frac{1}{n^{1/2}}%
\sum_{i=1}^{n}v_{i}^{\ast }\mathcal{I}_{i}\left( x\right) \text{,}
\end{equation*}%
where $v_{i}^{\ast }=u_{i}^{\ast }-\boldsymbol{P}_{i}^{\prime
}A_{n,i}^{+}C_{n,i}^{\ast }$ and $C_{n,i}^{\ast }=:n^{-1}\sum_{k=1}^{n}%
\boldsymbol{P}_{k}u_{k}^{\ast }\mathcal{J}_{k}\left( \widetilde{x}%
_{i}\right) $. The proof is completed if we show that, (\emph{in probability}%
), 
\begin{equation*}
\left( \mathbf{a}\right) \text{ }\mathcal{M}_{n}^{\ast }\left( x\right) 
\overset{weakly}{\Rightarrow }\mathcal{B}\left( F_{X}\left( x\right) \right)
\end{equation*}%
\begin{equation*}
\left( \mathbf{b}\right) \text{ }\sup_{x\in \left( 0,1\right) }\left\vert 
\widetilde{\mathcal{M}}_{n}^{\ast }\left( x\right) -\mathcal{M}_{n}^{\ast
}\left( x\right) \right\vert =o_{p^{\ast }}\left( 1\right) \text{.}
\end{equation*}

Part $\left( \mathbf{b}\right) $ holds true trivially using Lemma \ref{bias}
because $\widehat{u}_{i}^{\ast }-\widehat{u}_{i}=\boldsymbol{P}_{i}^{\prime
}\left( \widehat{b}^{\ast }-\widehat{b}\right) $. To show part $\left( 
\mathbf{a}\right) $ it suffices to show that $\left( \mathbf{a1}\right) $
the finite dimensional distributions converge to a Gaussian random variable
with the appropriate covariance structure and $\left( \mathbf{a2}\right) $
tightness of the sequence. The proofs proceed similarly as those in Theorem %
\ref{M_n}. Notice that in the proof of Theorem \ref{M_n}, we first
conditioned on $X$ and then we examined its asymptotic unconditional limit.

We begin with $\left( \mathbf{a1}\right) $. To that end, we first examine
the structure of the second moments. Proceeding as in the proof of Theorem %
\ref{M_n}, because $E^{\ast }u_{i}^{\ast 2}=:\widehat{\sigma }_{u}^{2}$, we
have that%
\begin{eqnarray*}
E^{\ast }\left( \mathcal{M}_{n}^{\ast }\left( x^{1}\right) \mathcal{M}%
_{n}^{\ast }\left( x^{2}\right) \right) &=&\frac{\widehat{\sigma }_{u}^{2}}{n%
}\sum_{i=1}^{n}\left( 1-\frac{\boldsymbol{P}_{i}^{\prime }A_{n,i}^{+}%
\boldsymbol{P}_{i}}{n}\right) \mathcal{I}_{i}\left( x^{1}\right)
+o_{p}\left( 1\right) \\
&&\overset{P}{\rightarrow }\sigma _{u}^{2}F_{X}\left( x^{1}\right)
\end{eqnarray*}%
proceeding as in the proof of $\left( \ref{eg}\right) $ and by Proposition %
\ref{sigma_est}.

Next, we examine the weak convergence of $\mathcal{M}_{n}^{\ast }\left(
x\right) $, which due to Cram\'{e}r-Wold device, it suffices to do so for a
fixed $x$. First observe that we have shown that%
\begin{equation*}
\sum_{i=1}^{n}E^{\ast }\left( \upsilon _{in}^{\ast 2}\left( x\right) \mid
X\right) \overset{P}{\rightarrow }1\text{,}
\end{equation*}%
where $\upsilon _{in}^{\ast }\left( x\right) =:\widehat{\sigma }%
_{u}^{-1}F_{X}^{-1/2}\left( x\right) \mathcal{I}_{i}\left( x\right)
v_{i}^{\ast }/n^{1/2}$ is a martingale difference triangular array of r.v.'s.

So, to complete the proof, it suffices to show the Lindeberg's condition for
which a sufficient condition is $\sum_{i=1}^{n}E^{\ast }\left\vert
v_{i}^{\ast }\right\vert ^{4}=o_{p}\left( n^{2}\right) $, which follows
proceeding as with the proof of Theorem \ref{M_n} and Proposition \ref%
{sigma_est}.

$\left( \mathbf{a2}\right) $ We now examine the tightness of%
\begin{equation*}
\frac{1}{n^{1/2}}\sum_{i=1}^{n}v_{i}^{\ast }\mathcal{I}_{i}\left( x\right) 
\text{.}
\end{equation*}%
The proof follows as that of Lemma \ref{tight} because again the proof there
was done conditionally on $X$ and then we examined its asymptotic
unconditional limit. Indeed, as we argued in the proof of Theorem \ref{M_n}
part $\left( \mathbf{b}\right) $, it suffices to show that%
\begin{equation}
E^{\ast }\left( \frac{1}{n^{1/2}}\sum_{i=1}^{n}v_{i}^{\ast }\mathcal{I}%
_{i}\left( x^{1};x^{2}\right) \right) ^{4}\overset{P}{\rightarrow }K\left\{ 
\frac{1}{n}\left( F_{X}\left( x^{2}\right) -F_{X}\left( x^{1}\right) \right)
+\left( x^{2}-x^{1}\right) ^{2}\sup_{x\in \left( x^{1},x^{2}\right)
}f_{X}^{2}\left( x\right) \right\},  \label{theo5_1}
\end{equation}
where $\mathcal{I}_{i}\left( x^{1},x^{2}\right) =\mathcal{I}_{i}\left(
x^{2}\right) -\mathcal{I}_{i}\left( x^{1}\right) $. By Burkholder's
inequality implies that the left side of $\left( \ref{theo5_1}\right) $ is
bounded by%
\begin{equation}
\frac{K}{n^{2}}E\left( \sum_{i=1}^{n}\left( v_{i}^{\ast 2}-E\left(
v_{i}^{\ast 2}\mid \mathcal{G}_{i}^{\ast }\right) \right) \mathcal{I}%
_{i}\left( x^{1},x^{2}\right) \right) ^{2}+\frac{K}{n^{2}}\left(
\sum_{i=1}^{n}E\left( v_{i}^{\ast 2}\mid \mathcal{G}_{i}^{\ast }\right) 
\mathcal{I}_{i}\left( x^{1},x^{2}\right) \right) ^{2}\text{,}  \label{burk_2}
\end{equation}%
where $\mathcal{G}_{i}^{\ast }$ denotes the sigma algebra generated by $%
\left\{ u_{i+1}^{\ast },...,u_{n}^{\ast }\right\} $ and $E\left( v_{i}^{\ast
2}\mid \mathcal{G}_{i}^{\ast }\right) =\widehat{\sigma }_{u}^{2}+\boldsymbol{%
P}_{i}^{\prime }A_{n,i}^{+}C_{n,i}^{\ast }C_{n,i}^{\ast \prime }A_{n,i}^{+}%
\boldsymbol{P}_{i}$ as it is easily seen. From here the proof proceeds as
that of Lemma \ref{tight} after we observe that the only difference is that
we have, say $\widehat{\sigma }_{u}^{2}$ instead of $\sigma _{u}^{2}$, and
observing that $\widehat{\sigma }_{u}^{2}-\sigma _{u}^{2}=o_{p}\left(
1\right) $ by Proposition \ref{sigma_est} and that 
\begin{equation*}
\kappa _{4}\left( u^{\ast }\right) =\frac{1}{n}\sum_{i=1}^{n}\widehat{u}%
_{i}^{4}-3\left( \frac{1}{n}\sum_{i=1}^{n}\widehat{u}_{i}^{2}\right) ^{2}%
\overset{P}{\rightarrow }Eu_{i}^{4}-3\left( Eu_{i}^{2}\right) ^{2}\text{,}
\end{equation*}%
so that the left side of $\left( \ref{burk_2}\right) $ is bounded by%
\begin{equation*}
K\left\{ \frac{1}{n}\left( F_{X}\left( x^{2}\right) -F_{X}\left(
x^{1}\right) \right) +\left( x^{2}-x^{1}\right) ^{2}\sup_{z\in \left(
x^{1},x^{2}\right) }f_{X}^{2}\left( z\right) \right\} \left( 1+o_{p}\left(
1\right) \right)
\end{equation*}%
and the conclusion follows by using the same arguments as those for the
proof of Billingsley's \cite{Billingsley} Theorems 22.1, see also
arguments in Wu's \cite{Wu} Lemma 14.\hfill $\blacksquare $

\section{\textbf{APPENDIX B}}

In what follows we shall abbreviate $p_{\ell ,L}\left( x;q\right) $ by $%
p_{\ell ,L}\left( x\right) $ for all $\ell =1,...,L$.

\begin{lemma}
\label{bias}Any linear combination of the \emph{B-splines}, $\boldsymbol{p}%
_{L}\left( x\right) =:\sum_{\ell =1}^{L}a_{\ell }p_{\ell ,L}\left( x\right) $%
, satisfies that $\left( \mathcal{T}_{n}\mathcal{B}_{n,L}\right) \left(
x\right) =0$, where $\mathcal{B}_{n,L}\left( x\right) =n^{-1}\sum_{k=1}^{n}%
\boldsymbol{p}_{L}\left( x_{k}\right) \mathcal{I}_{k}\left( x\right) $.
\end{lemma}

\begin{proof}
The proof is immediate after we notice that $\int_{x_{i}}^{1}\boldsymbol{P}%
_{L}\left( w\right) \mathcal{W}_{n}\left( dw\right) =A_{n,i}$. Indeed, $%
\left( \ref{Khma_n}\right) $ implies that $\left( \mathcal{T}_{n}\mathcal{B}%
_{n,L}\right) \left( x\right) $ is%
\begin{eqnarray*}
&&\frac{1}{n}\sum_{k=1}^{n}\left\{ \boldsymbol{p}_{L}\left( x_{k}\right) -%
\boldsymbol{P}_{k}^{\prime }A_{n,k}^{+}\frac{1}{n}\sum_{j=1}^{n}\boldsymbol{P%
}_{j}\boldsymbol{p}_{L}\left( x_{j}\right) \mathcal{J}_{j}\left( \widetilde{x%
}_{k}\right) \right\} \mathcal{I}_{k}\left( x\right) \\
&=&\frac{1}{n}\sum_{k=1}^{n}\left( \boldsymbol{P}_{k}^{\prime }-\boldsymbol{P%
}_{k}^{\prime }A_{n,k}^{+}\frac{1}{n}\sum_{j=1}^{n}\boldsymbol{P}_{j}%
\boldsymbol{P}_{j}^{\prime }\mathcal{J}_{j}\left( \widetilde{x}_{k}\right)
\right) a\mathcal{I}_{k}\left( x\right) \\
&=&\frac{1}{n}\sum_{k=1}^{n}\left( \boldsymbol{P}_{k}^{\prime }-\boldsymbol{P%
}_{k}^{\prime }A_{n,k}^{+}A_{n,k}\right) a\mathcal{I}_{k}\left( x\right) \\
&=&0\text{,}
\end{eqnarray*}%
by Harville's \cite{Harville} Theorem 12.3.4, where $a=\left(
a_{1},...,a_{L}\right) $.
\end{proof}

We now introduce some notation useful for the next lemmas. We shall denote $%
\Lambda \left( r\right) =\left\{ x:z^{r-1}\leq x<z^{r}-n^{-\varsigma
}\right\} $ and $\overline{\Lambda }\left( r\right) =\left\{
x:z^{r}-n^{-\varsigma }\leq x<z^{r}\right\} $.

\begin{lemma}
\label{Nor_P_i}Under Condition $C1$, we have that $E\left\Vert \widetilde{%
\boldsymbol{P}}_{i}\right\Vert ^{s}=O\left( L^{s/2}+Ln^{\left( s/2-1\right)
\varsigma }\right) $ for any $s\geq 2$, with $\widetilde{\boldsymbol{P}}_{i}$
given in $\left( \ref{q_1}\right) $.
\end{lemma}

\begin{proof}
First, 
\begin{equation}
E\left\Vert \widetilde{\boldsymbol{P}}_{i}\right\Vert
^{s}=\sum_{k=1}^{L}E\left\{ \left\Vert \widetilde{\boldsymbol{P}}%
_{i}\right\Vert ^{s}\mathcal{I}_{i}\left( z^{k-1};z^{k}\right) \right\} 
\text{,}  \label{lem_2_p}
\end{equation}%
where, for each $k=1,...,L$, we have that 
\begin{eqnarray}
E\left\{ \left\Vert \widetilde{\boldsymbol{P}}_{i}\right\Vert ^{s}\mathcal{I}%
_{i}\left( z^{k-1};z^{k}\right) \right\} &=&E\left\{ \left\Vert \overline{A}%
_{L,i}^{+1/2}D_{L}\left( \widetilde{x}_{i}\right) \mathbf{P}_{i}\right\Vert
^{s}\mathcal{I}_{i}\left( z^{k-1};z^{k}\right) \right\}  \notag \\
&=&KE\left\{ \left\Vert D_{L}\left( x_{i}\right) \mathbf{P}_{i}\right\Vert
^{s}\mathcal{I}\left( x_{i}\in \Lambda \left( k\right) \right) \right\}
\label{Lemma2_3} \\
&&+KE\left\{ \left\Vert D_{L}\left( z^{k}\right) \mathbf{P}_{i}\right\Vert
^{s}\mathcal{I}\left( x_{i}\in \overline{\Lambda }\left( k\right) \right)
\right\}  \notag
\end{eqnarray}%
because $\overline{A}_{L,i}=:diag\left( 0,B_{L,i}\right) $ and \underline{$%
\lambda $}$\left( B_{L,i}\right) >0$. The proof is now standard after
observing that the first term on the right of $\left( \ref{Lemma2_3}\right) $
is bounded by 
\begin{eqnarray}
&&E\left\Vert \frac{p_{k,L}^{s}\left( x_{i}\right) \mathcal{I}\left(
x_{i}\in \Lambda \left( k\right) \right) }{L^{sq}\left( z^{k}-x_{i}\right)
^{s\left( q+1/2\right) }}\right\Vert +K\sum_{j=k+1}^{k+q}E\left\Vert
L^{s/2}p_{j,L}^{s}\left( x_{i}\right) \mathcal{I}\left( x_{i}\in \Lambda
\left( k\right) \right) \right\Vert  \notag \\
&=&K\left( E\left\Vert x_{i}^{-s/2}\mathcal{I}_{i}\left( n^{-\varsigma
};L^{-1}\right) \right\Vert +L^{s/2-1}\right)  \label{Lemma2_4} \\
&=&O\left( n^{\left( s/2-1\right) \varsigma }+L^{s/2-1}\right) \text{.} 
\notag
\end{eqnarray}%
since $E\left( p_{j,L}^{s}\left( x_{i}\right) \right) =O\left( L^{-1}\right) 
$ and $p_{k,L}\left( x_{i}\right) \leq KL^{q}\left( z^{k}-x_{i}\right) ^{q}$%
, whereas the second term on the right of $\left( \ref{Lemma2_3}\right) $ is
bounded by%
\begin{equation*}
K\sum_{j=k+1}^{k+q}E\left\Vert L^{s/2}p_{j,L}^{s}\left( x_{i}\right) 
\mathcal{I}\left( x_{i}\in \overline{\Lambda }\left( k\right) \right)
\right\Vert =O\left( L^{s/2-1}\right)
\end{equation*}%
proceeding as with the second term on the left of $\left( \ref{Lemma2_4}%
\right) $.
\end{proof}

\begin{lemma}
\label{Est_A_i}Under Conditions $C1$ and $C3$, we have that 
\begin{equation}
\left\Vert A_{L,i}^{+1/2}A_{n,i}A_{L,i}^{+1/2}-I\right\Vert _{E}\overset{P}{%
\rightarrow }0\text{.}  \label{lem_51}
\end{equation}
\end{lemma}

\begin{proof}
We shall consider the scenario where $x_{i}\in \Lambda =:\left\{ \Lambda
\left( r\right) ;\text{ }r=1,...,L\right\} $, the case when $x_{i}\in 
\overline{\Lambda }=:\left\{ \overline{\Lambda }\left( r\right) ;\text{ }%
r=1,...,L\right\} $ is handled similarly, if not easier. Because when $%
x_{i}\in \Lambda \left( r\right) $, $\widetilde{x}_{i}=x_{i}$, the matrix
inside the norm in $\left( \ref{lem_51}\right) $ is $\overline{A}%
_{L,i}^{+1/2}H_{n,i}\overline{A}_{L,i}^{+1/2}$, where 
\begin{equation}
H_{n,i}=:D_{L}\left( x_{i}\right) \left( \frac{1}{n}\sum_{k=1;\neq i}^{n}%
\boldsymbol{P}_{k}\boldsymbol{P}_{k}^{\prime }\mathcal{J}_{k}\left(
x_{i}\right) -A_{L,i}\right) D_{L}\left( x_{i}\right) +\frac{1}{n}%
D_{L}\left( x_{i}\right) \boldsymbol{P}_{i}\boldsymbol{P}_{i}^{\prime
}D_{L}\left( x_{i}\right) \text{.}  \label{h_n}
\end{equation}

The second term on the right of $\left( \ref{h_n}\right) $ is $o_{p}\left(
1\right) $ as we now show. Because 
\begin{equation}
p_{\ell ,L}\left( x_{k}\right) p_{m,L}\left( x_{k}\right) =0\text{ \ \ if }%
m\geq \ell +q\text{,}  \label{p_p}
\end{equation}%
Cauchy-Schwarz and then Markov's inequalities implies that it suffices to
show that 
\begin{equation}
E\left\vert \frac{1}{n}\sum_{\ell =1}^{L}d_{\ell }^{2}\left( x_{i}\right)
p_{\ell ,L}^{2}\left( x_{i}\right) \right\vert =o\left( 1\right) \text{.}
\label{8_3_1}
\end{equation}%
But this is the case because, recall that $x_{i}\in \Lambda $, the left side
is bounded by%
\begin{equation*}
\frac{K}{n}\sum_{r=1}^{L}E\left( \frac{1}{\left( z^{r}-x_{i}\right) }%
+\sum_{\ell =r+1}^{r+q}Lp_{\ell ,L}^{2}\left( x_{i}\right) \right) \mathcal{I%
}\left( x_{i}\in \Lambda \left( r\right) \right) =O\left( \frac{L\left(
1+\varsigma \log n\right) }{n}\right) =o\left( 1\right) \text{,}
\end{equation*}%
since $E\left( Lp_{\ell ,L}^{2}\left( x_{i}\right) \mathcal{I}\left(
x_{i}\in \Lambda \left( r\right) \right) \right) =O\left( 1\right) $. It is
worth observing that $\left( \ref{8_3_1}\right) $ also holds uniformly in $x$
since $p_{r,L}^{2}\left( x\right) <K$ and Condition $C3$.

Next, the first term on the right of $\left( \ref{h_n}\right) $ is also $%
o_{p}\left( 1\right) $. Indeed, because $\overline{A}_{L,i}=:diag\left(
0,B_{L,i}\right) $ and \underline{$\lambda $}$\left( B_{L,i}\right) >0$, it
suffices to show that $\left\Vert H_{n,i}\right\Vert _{E}\overset{P}{%
\rightarrow }0$ and more specifically, in view of $\left( \ref{p_p}\right) $%
, to show that 
\begin{equation}
\sum_{r=1}^{L}\left\{ \sum_{\ell =r}^{L}\frac{1}{n}\sum_{k=1;\neq
i}^{n}d_{\ell }^{2}\left( x_{i}\right) \left( p_{\ell ,L}^{2}\left(
x_{k}\right) \mathcal{J}_{k}\left( x_{i}\right) -a_{i,\ell \ell }\right)
\right\} \mathcal{I}\left( x_{i}\in \Lambda \left( r\right) \right) \overset{%
P}{\rightarrow }0\text{.}  \label{H_i}
\end{equation}%
But $\left( \ref{H_i}\right) $ holds true because, conditionally on $x_{i}$,
when $\ell =r$, we have that Condition $C1$ and $\mathcal{I}\left( x_{i}\in
\Lambda \left( r_{1}\right) \right) \mathcal{I}\left( x_{i}\in \Lambda
\left( r_{2}\right) \right) =0$ if $r_{1}\neq r_{2}$ implies that 
\begin{eqnarray*}
&&E\left( \sum_{r=1}^{L}\frac{1}{n}\sum_{k=1;\neq i}^{n}\frac{\left( 
\overline{p}_{r,L}^{2}\left( x_{k}\right) \mathcal{J}_{k}\left( x_{i}\right)
-E\left( \overline{p}_{r,L}^{2}\left( x_{k}\right) \mathcal{J}_{k}\left(
x_{i}\right) \right) \right) }{\left( z^{r}-x_{i}\right) ^{2q+1}}\mathcal{I}%
\left( x_{i}\in \Lambda \left( r\right) \right) \right) ^{2} \\
&=&\sum_{r=1}^{L}\frac{K}{n^{2}}\sum_{k=1;\neq i}^{n}\frac{E\overline{p}%
_{r,L}^{4}\left( x_{k}\right) \mathcal{J}_{k}\left( x_{i}\right)
-E^{2}\left( \overline{p}_{r,L}^{2}\left( x_{k}\right) \mathcal{J}_{k}\left(
x_{i}\right) \right) }{\left( z^{r}-x_{i}\right) ^{4q+2}}\mathcal{I}\left(
x_{i}\in \Lambda \left( r\right) \right) \\
&=&\frac{K}{n}\sum_{r=1}^{L}\left( z^{r}-x_{i}\right) ^{-1}\mathcal{I}\left(
x_{i}\in \Lambda \left( r\right) \right) =O_{p}\left( \frac{L\log n}{n}%
\right) \text{,}
\end{eqnarray*}%
because $\overline{p}_{r,L}\left( x_{k}\right) \mathcal{J}_{k}\left(
x_{i}\right) \mathcal{I}\left( x_{i}\in \Lambda \left( r\right) \right) \leq
K\left( z^{r}-x_{k}\right) ^{q}\mathcal{J}_{k}\left( x_{i}\right) \mathcal{I}%
\left( x_{i}\in \Lambda \left( r\right) \right) $ and Markov's inequality,
where 
\begin{equation}
\overline{p}_{r,L}\left( x_{k}\right) =L^{-q}p_{r,L}\left( x_{k}\right)
\label{p_bar}
\end{equation}%
and $E\left\{ \left( z^{r}-x_{i}\right) ^{-1}\mathcal{I}\left( x_{i}\in
\Lambda \left( r\right) \right) \right\} =O\left( \log n\right) $.

Next when $\ell >r$. Denoting $d_{\ell }^{2}\left( x_{i}\right) \left(
p_{\ell ,L}^{2}\left( x_{k}\right) \mathcal{J}_{k}\left( x_{i}\right)
-a_{i,\ell \ell }\right) =:\psi _{i,\ell }\left( x_{k}\right) $, the second
conditional moments of the left side of $\left( \ref{H_i}\right) $ is%
\begin{equation*}
\frac{1}{n^{2}}\sum_{r=1}^{L}\left\{ \sum_{\ell _{1},\ell
_{2}=r}^{L}\sum_{k=1;\neq i}^{n}E\left( \psi _{i,\ell _{1}}\left(
x_{k}\right) \psi _{i,\ell _{2}}\left( x_{k}\right) \mid x_{i}\right)
\right\} \mathcal{I}\left( x_{i}\in \Lambda \left( r\right) \right) =O\left( 
\frac{L^{2}}{n}\right)
\end{equation*}%
because $a_{i,\ell \ell }=O\left( L^{-1}\right) $, $d_{\ell }^{2}\left(
x_{i}\right) =L$, $Ep_{\ell ,L}^{s}\left( x_{k}\right) =O\left(
L^{-1}\right) $ for any $s\geq 1$ and $\left( \ref{p_p}\right) $ implies
that 
\begin{equation*}
\left\vert E\left( \psi _{i,\ell _{1}}\left( x_{k}\right) \psi _{i,\ell
_{2}}\left( x_{k}\right) \mid x_{i}\right) \right\vert \leq K\mathcal{I}%
\left( \left\vert \ell _{2}-\ell _{1}\right\vert \geq q\right) +KL\mathcal{I}%
\left( \left\vert \ell _{2}-\ell _{1}\right\vert <q\right) \text{.}
\end{equation*}%
So, $\left( \ref{H_i}\right) $ holds true for these terms by Markov's
inequality, Condition $C3$ and because $\sum_{r=1}^{L}\mathcal{I}\left(
x_{i}\in \Lambda \left( r\right) \right) =1$, which concludes the proof of
the lemma.
\end{proof}

\begin{lemma}
\label{Est_A_uni}Under Conditions $C1$ and $C3$, we have that 
\begin{equation*}
\sup_{x}\left\Vert A_{L}^{+1/2}\left( x\right) A_{n}\left( x\right)
A_{L}^{+1/2}\left( x\right) -I\right\Vert _{E}\overset{P}{\rightarrow }0%
\text{.}
\end{equation*}
\end{lemma}

\begin{proof}
Arguing as we did in the proof of Lemma \ref{Est_A_i}, it suffices to show
that 
\begin{equation}
\sup_{x}\left\Vert H_{n}\left( x\right) \right\Vert _{E}=:\sup_{1\leq r\leq
L}\left\{ \sup_{x\in \Lambda \left( r\right) }+\sup_{x\in \overline{\Lambda }%
\left( r\right) }\right\} \left\Vert H_{n}\left( x\right) \right\Vert _{E}%
\overset{P}{\rightarrow }0\text{.}  \label{Lemma_5_1}
\end{equation}%
We begin examining the contribution due to the sets $\Lambda \left( r\right) 
$. First, as we argued in Lemma \ref{Est_A_i}, we observe that $\left( \ref%
{Lemma_5_1}\right) $ holds true if the diagonal elements of $H_{n}\left(
x\right) $ converges uniformly to zero in probability. That is, it suffices
to show that 
\begin{eqnarray}
&&\sup_{1\leq r\leq L}\sup_{x\in \Lambda \left( r\right) }\sum_{\ell
=r+q}^{L}\left\vert \frac{1}{n}\sum_{k=1}^{n}d_{\ell }^{2}\left( x\right)
\left( p_{\ell ,L}^{2}\left( x_{k}\right) \mathcal{J}_{k}\left( x\right)
-Ep_{\ell ,L}^{2}\left( x_{k}\right) \mathcal{J}_{k}\left( x\right) \right)
\right\vert ^{2}  \label{L_4} \\
&&+\left( \sup_{1\leq r\leq L}\sup_{x\in \Lambda \left( r\right) }\sum_{\ell
=r}^{r+q-1}\left\vert \frac{1}{n}\sum_{k=1}^{n}d_{\ell }^{2}\left( x\right)
\left( p_{\ell ,L}^{2}\left( x_{k}\right) \mathcal{J}_{k}\left( x\right)
-Ep_{\ell r,L}^{2}\left( x_{k}\right) \mathcal{J}_{k}\left( x\right) \right)
\right\vert \right) ^{2}  \notag \\
&=&o_{p}\left( 1\right) \text{.}  \notag
\end{eqnarray}%
Recall that due to the definition of \emph{B-splines}, for any $r$, $p_{\ell
,L}\left( x_{k}\right) \mathcal{J}_{k}\left( x\right) \mathcal{I}\left( x\in
\Lambda \left( r\right) \right) =0$ if $\ell <r$.

Because $x\in \Lambda \left( r\right) $ implies that $p_{\ell ,L}\left(
x_{k}\right) \mathcal{J}_{k}\left( x\right) =p_{\ell ,L}\left( x_{k}\right) 
\mathcal{J}_{k}\left( z^{r}\right) $ and $d_{\ell }^{2}\left( x\right) =L$
for $\ell \geq r+q$, the first term on the left of $\left( \ref{L_4}\right) $
is 
\begin{eqnarray*}
&&\sup_{1\leq r\leq L}\sum_{\ell =r+q}^{L}\left\vert \frac{L}{n}%
\sum_{k=1}^{n}\left( p_{\ell ,L}^{2}\left( x_{k}\right) \mathcal{J}%
_{k}\left( z^{r}\right) -Ep_{\ell ,L}^{2}\left( x_{k}\right) \mathcal{J}%
_{k}\left( z^{r}\right) \right) \right\vert ^{2} \\
&\leq &\sum_{\ell =1}^{L}\left\vert \frac{L}{n}\sum_{k=1}^{n}\left( p_{\ell
,L}^{2}\left( x_{k}\right) \mathcal{J}_{k}\left( z^{\ell }\right) -Ep_{\ell
,L}^{2}\left( x_{k}\right) \mathcal{J}_{k}\left( z^{\ell }\right) \right)
\right\vert ^{2} \\
&=&O_{p}\left( L^{2}/n\right)
\end{eqnarray*}%
by Condition $C1$ and that $Ep_{\ell ,L}^{4}\left( x_{k}\right) =O\left(
L^{-1}\right) $, and then by Markov's inequality.

So, to complete the proof we need to show that the second term on the left
of $\left( \ref{L_4}\right) $ is also $o_{p}\left( 1\right) $. We shall look
at the case when $\ell =r$, being the cases when $\ell >r$ similarly, if not
easier, handled. To that end, we first notice that this term is bounded by 
\begin{equation}
\left( \sup_{1\leq r\leq L}\sup_{x\in \Lambda \left( r\right) }\frac{1}{%
\left( z^{r}-x\right) ^{1/2}}\left\vert \frac{1}{n}\sum_{k=1}^{n}\eta
_{r}\left( x_{k};x\right) -E\eta _{r}\left( x_{k};x\right) \right\vert
\right) ^{2}\text{,}  \label{L_1}
\end{equation}%
where $\eta _{r}\left( x_{k};x\right) =\overline{p}_{r,L}^{2}\left(
x_{k}\right) \mathcal{J}_{k}\left( x\right) /\left( z^{r}-x\right) ^{2q+1/2}$
with $\overline{p}_{r,L}\left( x_{k}\right) $ defined in $\left( \ref{p_bar}%
\right) $.

Consider points $z^{r}\left[ j\right] =z^{r}-\left( j-1\right) /n$, $%
j=n^{1-\varsigma }+1,...,n/L$, so that $z^{r}\left[ n^{1-\varsigma }+1\right]
=z^{r}-n^{-\varsigma }$ and denote $\eta _{r}\left( x_{k};x^{1},x^{2}\right)
=\eta _{r}\left( x_{k};x^{2}\right) -\eta _{r}\left( x_{k};x^{1}\right) $.
Now if $x\in \left( z^{r}\left[ j+1\right] ,z^{r}\left[ j\right] \right) $,
we have that 
\begin{equation*}
\left\vert \frac{1}{n}\sum_{i=1}^{n}\eta _{r}\left( x_{i};x,z^{r}\left[ j%
\right] \right) \right\vert \leq \frac{1}{n}\sum_{i=1}^{n}\left\vert \eta
_{r}\left( x_{i};z^{r}\left[ j+1\right] ,z^{r}\left[ j\right] \right)
\right\vert \text{,}
\end{equation*}%
so that the triangle inequality yields that $\left( \ref{L_1}\right) $ is,
except constants, bounded by%
\begin{eqnarray*}
&&\left( \sup_{1\leq r\leq L}\sup_{n^{1-\varsigma }+1\leq j\leq n/L}\frac{1}{%
\left( z^{r}-z^{r}\left[ j\right] \right) ^{1/2}}\left\vert \frac{1}{n}%
\sum_{i=1}^{n}\eta _{r}\left( x_{i};z^{r}\left[ j\right] ,z^{r}\left[
n^{1-\varsigma }+1\right] \right) -E\left( \cdot \right) \right\vert \right)
^{2} \\
&&+\left( \sup_{1\leq r\leq L}\sup_{n^{1-\varsigma }+1\leq j\leq n/L}\frac{1%
}{\left( z^{r}-z^{r}\left[ j\right] \right) ^{1/2}}\left\vert \frac{1}{n}%
\sum_{i=1}^{n}\left\vert \eta _{r}\left( x_{i};z^{r}\left[ j+1\right] ,z^{r}%
\left[ j\right] \right) \right\vert -E\left( \cdot \right) \right\vert
\right) ^{2}
\end{eqnarray*}%
\begin{equation}
+\left( \sup_{1\leq r\leq L}\sup_{n^{1-\varsigma }+1\leq j\leq n/L}\frac{1}{%
\left( z^{r}-z^{r}\left[ j\right] \right) ^{1/2}}E\left\vert \eta _{r}\left(
x_{i};z^{r}\left[ j+1\right] ,z^{r}\left[ j\right] \right) \right\vert
\right) ^{2}\text{.}  \label{L_2}
\end{equation}

Because for any $q>0$, 
\begin{equation}
E\left\vert \eta _{r}\left( x_{i};z^{r}\left[ j+1\right] ,z^{r}\left[ j%
\right] \right) \right\vert ^{q}\leq K\int_{z^{r}\left[ j+1\right] }^{z^{r}%
\left[ j\right] }\left( z^{r}-x\right) ^{-q/2}dx\leq K\left( z^{r}-z^{r}%
\left[ j\right] \right) ^{-q/2}n^{-1}\text{,}  \label{L_3}
\end{equation}%
the last term of $\left( \ref{L_2}\right) $ is bounded by $K\left(
\sup_{1\leq r\leq L}\sup_{n^{1-\varsigma }+1\leq j\leq n/L}n^{\varsigma
-1}\right) ^{2}=o\left( 1\right) $ because $\varsigma <1$.

Next using $\left( \ref{L_3}\right) $ and the inequalities $\left(
\sup_{x}\left\vert g\left( x\right) \right\vert \right)
^{q}=\sup_{x}\left\vert g\left( x\right) \right\vert ^{q}$ and $\sup_{\ell
}\left\vert c_{\ell }\right\vert \leq \sum_{\ell }\left\vert c_{\ell
}\right\vert $, the expectation of the second term of $\left( \ref{L_2}%
\right) $ to the power $p/2$ is bounded by 
\begin{eqnarray*}
&&\sum_{r=1}^{L}\sum_{j=n^{1-\varsigma }+1}^{n/L}\frac{1}{\left( z^{r}-z^{r}%
\left[ j\right] \right) ^{p/2}}E\left\vert \frac{1}{n}\sum_{i=1}^{n}\left%
\vert \eta _{r}\left( x_{i};z^{r}\left[ j+1\right] ,z^{r}\left[ j\right]
\right) \right\vert -E\left( \cdot \right) \right\vert ^{p} \\
&=&K\sum_{r=1}^{L}\frac{1}{n^{p}}\sum_{j=n^{1-\varsigma }+1}^{n/L}\frac{1}{%
\left( z^{r}-z^{r}\left[ j\right] \right) ^{p}} \\
&=&KL\sum_{j=n^{1-\varsigma }+1}^{n/L}j^{-p}=O\left( Ln^{\left( p-1\right)
\left( \varsigma -1\right) }\right) =o\left( 1\right)
\end{eqnarray*}%
for any $\varsigma <1$ choosing $p$ large enough. Thus, the second term of $%
\left( \ref{L_2}\right) $ is $o_{p}\left( 1\right) $.

Finally, the first term of $\left( \ref{L_2}\right) $. First, as we argued
with the second term, Condition $C1$ yields that its $p-th$ absolute moment
is bounded by 
\begin{eqnarray*}
&&\sum_{r=1}^{L}\sum_{j=n^{1-\varsigma }+1}^{n/L}\frac{1}{\left( z^{r}-z^{r}%
\left[ j\right] \right) ^{p/2}}E\left\vert \frac{1}{n}\sum_{i=1}^{n}\eta
_{r}\left( x_{i};z^{r}\left[ j\right] ,z^{r}\left[ n^{1-\varsigma }+1\right]
\right) -E\left( \cdot \right) \right\vert ^{p} \\
&=&\frac{K}{n^{p/2}}\sum_{r=1}^{L}\sum_{j=n^{1-\varsigma }+1}^{n/L}\frac{%
\log ^{p/2}n}{\left( z^{r}-z^{r}\left[ j\right] \right) ^{p/2}} \\
&=&O\left( \frac{L\log n}{n^{\left( p/2-1\right) \left( 1-\varsigma \right) }%
}\right)
\end{eqnarray*}%
which is $o\left( 1\right) $ since we can always choose $p$ large enough
such that $L=o\left( n^{p/2\left( 1-\varsigma \right) }\right) $ for any $%
\varsigma <1$ and Condition $C3$. Notice that we can also bound the left
side of $\left( \ref{L_3}\right) $ by $K\log n$ when $q=2$ there.

To complete the proof of the lemma, we need to examine the contribution due
to the sets $\overline{\Lambda }\left( r\right) $ into the left of $\left( %
\ref{Lemma_5_1}\right) $. However, observing that when $x_{i}\in \overline{%
\Lambda }\left( r\right) $, $\mathcal{J}_{k}\left( x_{i}\right) =\mathcal{J}%
_{k}\left( z^{r}\right) $, the proof proceeds as that of the first term on
the left of $\left( \ref{L_4}\right) $, and so it is omitted.
\end{proof}

\begin{lemma}
\label{tight}Under Conditions $C1$ and $C3$, we have that for any $0\leq
x^{1}<x^{2}\leq 1$, 
\begin{eqnarray}
&&E\left( \frac{1}{n^{1/2}}\sum_{i=1}^{n}v_{i}\mathcal{I}_{i}\left(
x^{1},x^{2}\right) \right) ^{4}  \label{lem_1} \\
&=&K\left\{ \frac{1}{n}\left( F_{X}\left( x^{2}\right) -F_{X}\left(
x^{1}\right) \right) +\left( x^{2}-x^{1}\right) ^{2}\sup_{x\in \left(
x^{1},x^{2}\right) }f_{X}^{2}\left( x\right) \right\} \left( 1+o\left(
1\right) \right) \text{.}  \notag
\end{eqnarray}
\end{lemma}

\begin{proof}
First, as we mentioned in Section 3, we notice that we can arrange the
observations according to $x$ without modifying the properties of 
\begin{equation*}
\frac{1}{n^{1/2}}\sum_{i=1}^{n}v_{i}\mathcal{I}_{i}\left( x^{1},x^{2}\right) 
\text{,}
\end{equation*}%
so that $v_{i}$ becomes a martingale difference sequence of r.v.'s. So,
Burkholder's inequality implies that the left side of $\left( \ref{lem_1}%
\right) $ is bounded by%
\begin{equation}
\frac{K}{n^{2}}E\left( \sum_{i=1}^{n}\left( v_{i}^{2}-E\left( v_{i}^{2}\mid 
\mathcal{G}_{i}\right) \right) \mathcal{I}_{i}\left( x^{1},x^{2}\right)
\right) ^{2}+\frac{K}{n^{2}}\left( \sum_{i=1}^{n}E\left( v_{i}^{2}\mid 
\mathcal{G}_{i}\right) \mathcal{I}_{i}\left( x^{1},x^{2}\right) \right) ^{2}%
\text{,}  \label{lem_2}
\end{equation}%
where $\mathcal{G}_{i}$ denotes the sigma algebra generated by $\left\{
u_{i+1},...,u_{n}\right\} $ and $E\left( v_{i}^{2}\mid \mathcal{G}%
_{i}\right) =\sigma _{u}^{2}+\boldsymbol{P}_{i}^{\prime
}A_{n,i}^{+}C_{n,i}C_{n,i}^{\prime }A_{n,i}^{+}\boldsymbol{P}_{i}$ as it is
easily seen. We shall first examine the second term of $\left( \ref{lem_2}%
\right) $. By standard inequalities, that term is bounded by%
\begin{equation}
K\sigma _{u}^{4}\left( \frac{1}{n}\sum_{i=1}^{n}\mathcal{I}_{i}\left(
x^{1},x^{2}\right) \right) ^{2}+\frac{K}{n^{2}}\left( \sum_{i=1}^{n}%
\boldsymbol{P}_{i}^{\prime }A_{n,i}^{+}C_{n,i}C_{n,i}^{\prime }A_{n,i}^{+}%
\boldsymbol{P}_{i}\mathcal{I}_{i}\left( x^{1},x^{2}\right) \right) ^{2}\text{%
.}  \label{lem_4}
\end{equation}%
Because Condition $C1$ yields 
\begin{equation*}
\frac{1}{n}\sum_{i=1}^{n}\mathcal{I}_{i}\left( x^{1},x^{2}\right) \overset{P}%
{\rightarrow }F_{X}\left( x^{2}\right) -F_{X}\left( x^{1}\right) \text{,}
\end{equation*}%
we have that the contribution of the first term of $\left( \ref{lem_4}%
\right) $ into the left of $\left( \ref{lem_1}\right) $ is%
\begin{equation*}
\left( F_{X}\left( x^{2}\right) -F_{X}\left( x^{1}\right) \right) ^{2}\leq
K\left( x^{2}-x^{1}\right) ^{2}\sup_{x\in \left( x^{1},x^{2}\right)
}f_{X}^{2}\left( x\right) \text{.}
\end{equation*}%
Next, the second term of $\left( \ref{lem_4}\right) $ is also $K\left(
x^{2}-x^{1}\right) ^{2}\sup_{x\in \left( x^{1},x^{2}\right) }f_{X}^{2}\left(
x\right) $, as we now show. Indeed, Cauchy-Schwarz inequality and Lemma \ref%
{Est_A_uni}, see also $\left( \ref{A_inv}\right) $, yield that this term is
bounded by%
\begin{equation}
\frac{1}{n}\sum_{i=1}^{n}\widetilde{\boldsymbol{P}}_{i}^{\prime }\left( 
\widetilde{C}_{n,i}\widetilde{C}_{n,i}^{\prime }\right) ^{2}\widetilde{%
\boldsymbol{P}}_{i}\mathcal{I}_{i}\left( x^{1},x^{2}\right) ~\frac{1}{n}%
\sum_{i=1}^{n}\widetilde{\boldsymbol{P}}_{i}^{\prime }\widetilde{\boldsymbol{%
P}}_{i}\mathcal{I}_{i}\left( x^{1},x^{2}\right) \text{,}  \label{lem_5_3}
\end{equation}%
where $\widetilde{C}_{n,i}=A_{n,i}^{+1/2}n^{-1}\sum_{j=i+1}^{n}\boldsymbol{P}%
_{j}u_{j}$.

Now the second factor of $\left( \ref{lem_5_3}\right) $ converges to 
\begin{equation}
E\left( \left\Vert \widetilde{\boldsymbol{P}}_{i}\right\Vert ^{2}\mathcal{I}%
_{i}\left( x^{1},x^{2}\right) \right) =O\left( n^{\varsigma }+L\right)
\left( F_{X}\left( x^{2}\right) -F_{X}\left( x^{1}\right) \right)
\label{Lema_5_1}
\end{equation}%
because 
\begin{equation}
\left\Vert \widetilde{\boldsymbol{P}}_{i}\right\Vert ^{2}=O\left(
n^{\varsigma }+L\right) \text{.}  \label{p_norm}
\end{equation}%
Next, because 
\begin{equation*}
E\left( \widetilde{C}_{n,i}\widetilde{C}_{n,i}^{\prime }\mid X\right) =\frac{%
\sigma _{u}^{2}}{n}diag\left( \text{\b{0}},1,...,1\right) \text{,}
\end{equation*}%
we have that the conditional expectation of the first factor of $\left( \ref%
{lem_5_3}\right) $ is%
\begin{equation}
\frac{3}{n}\sum_{i=1}^{n}\frac{\left\Vert \widetilde{\boldsymbol{P}}%
_{i}\right\Vert ^{2}}{n^{2}}\mathcal{I}_{i}\left( x^{1},x^{2}\right) +\frac{%
\kappa _{u}\left( u\right) }{n}\sum_{i=1}^{n}\frac{\left\Vert \widetilde{%
\boldsymbol{P}}_{i}\right\Vert ^{2}}{n^{4}}\sum_{j=i+1}^{n}\left\Vert
D_{L}\left( x_{i}\right) \boldsymbol{P}_{j}\mathcal{J}_{k}\left( \widetilde{x%
}_{i}\right) \right\Vert ^{4}\mathcal{I}_{i}\left( x^{1},x^{2}\right)
\label{Lema_5_2}
\end{equation}
\begin{equation*}
\leq O\left( n^{\varsigma }+L\right) \left( \frac{3}{n^{3}}\sum_{i=1}^{n}%
\mathcal{I}_{i}\left( x^{1},x^{2}\right) +\frac{1}{n}\sum_{i=1}^{n}\frac{1}{%
n^{4}}\sum_{j=i+1}^{n}\left\Vert D_{L}\left( x_{i}\right) \boldsymbol{P}_{j}%
\mathcal{J}_{k}\left( \widetilde{x}_{i}\right) \right\Vert ^{4}\mathcal{I}%
_{i}\left( x^{1},x^{2}\right) \right)
\end{equation*}%
because $\overline{\lambda }\left( D_{L}^{-1}\left( x_{i}\right)
A_{n,i}^{+1/2}\right) <K$ by Lemma \ref{Est_A_uni} and $\left( \ref{p_norm}%
\right) $. By standard arguments, the first term on the right of $\left( \ref%
{Lema_5_2}\right) $ is 
\begin{equation*}
3\frac{n^{\varsigma }+L}{n^{2}}\left( F_{X}\left( x^{2}\right) -F_{X}\left(
x^{1}\right) \right) \left( 1+o_{p}\left( 1\right) \right)
\end{equation*}%
Next we examine the second term of $\left( \ref{Lema_5_2}\right) $. First we
have that 
\begin{eqnarray*}
&&\left\Vert D_{L}\left( x_{i}\right) \boldsymbol{P}_{j}\mathcal{J}%
_{k}\left( \widetilde{x}_{i}\right) \right\Vert ^{4}\mathcal{I}_{i}\left(
x^{1},x^{2}\right) \\
&=&\sum_{k=1}^{L}\left\Vert D_{L}\left( x_{i}\right) \boldsymbol{P}_{j}%
\mathcal{J}_{k}\left( \widetilde{x}_{i}\right) \right\Vert ^{4}\left\{ 
\mathcal{I}\left( x_{i}\in \Lambda \left( k\right) \right) +\mathcal{I}%
\left( x_{i}\in \overline{\Lambda }\left( k\right) \right) \right\} \mathcal{%
I}_{i}\left( x^{1},x^{2}\right) \text{.}
\end{eqnarray*}%
We shall examine the contribution due to the first term on the right, the
second is similarly handle. Now, the expectation of a typical term of the
last displayed expression is, conditionally on $x_{i}$, 
\begin{eqnarray*}
&&\sum_{k=1}^{L}E\left\Vert D_{L}\left( x_{i}\right) \boldsymbol{P}_{j}%
\mathcal{J}_{k}\left( x_{i}\right) \right\Vert ^{4}\mathcal{I}\left(
x_{i}\in \Lambda \left( k\right) \right) \\
&\leq &K\sum_{k=1}^{L}\left( \frac{E\overline{p}_{k\left( x_{i}\right)
,L}^{4}\left( x_{j}\right) \mathcal{I}_{i}\left( z^{k\left( x_{i}\right)
-1},z^{k\left( x_{i}\right) }\right) }{\left( z^{k\left( x_{i}\right)
}-x_{i}\right) ^{4q+2}}\mathcal{J}_{k}\left( x_{i}\right) \right) \mathcal{I}%
\left( x_{i}\in \Lambda \left( k\right) \right) \\
&&+L^{2}\sum_{k=1}^{L}E\left( p_{\ell ,L}^{4}\left( x_{j}\right) \mathcal{J}%
_{k}\left( x_{i}\right) \right) \mathcal{I}\left( x_{i}\in \Lambda \left(
k\right) \right) \\
&=&\left( K\left( z^{k\left( x_{i}\right) }-x_{i}\right) ^{-1}+L\right)
\sum_{k=1}^{L}\mathcal{I}\left( x_{i}\in \Lambda \left( k\right) \right) \\
&=&K\left( n^{\varsigma }+L\right) \text{,}
\end{eqnarray*}%
because $Ep_{\ell ,L}^{4}\left( x_{j}\right) =O\left( L^{-1}\right) $. So
the second term of $\left( \ref{Lema_5_2}\right) $ is 
\begin{equation*}
O\left( \left( n^{\varsigma }+L\right) \frac{1}{n^{3}}\right) \frac{1}{n}%
\sum_{j=1}^{n}\mathcal{I}_{i}\left( x^{1},x^{2}\right)
\end{equation*}%
which implies that by standard arguments that it is%
\begin{equation*}
O\left( \left( n^{\varsigma }+L\right) \frac{1}{n^{3}}\right) \left(
F_{X}\left( x^{2}\right) -F_{X}\left( x^{1}\right) \right) \left(
1+o_{p}\left( 1\right) \right) \text{.}
\end{equation*}

Thus, gathering terms, we have that the first factor of $\left( \ref{lem_5_3}%
\right) $, i.e. $\left( \ref{Lema_5_2}\right) $, is 
\begin{equation*}
O\left( \left( \frac{n^{\varsigma }+L}{n}\right) ^{2}\frac{1}{n}\right)
\left( F_{X}\left( x^{2}\right) -F_{X}\left( x^{1}\right) \right) \left(
1+o_{p}\left( 1\right) \right) \text{.}
\end{equation*}%
So, the last displayed expression together with $\left( \ref{Lema_5_1}%
\right) $ implies that the second term of $\left( \ref{lem_4}\right) $ is 
\begin{equation*}
O\left( \left( \frac{n^{\varsigma }+L}{n}\right) ^{3}\right) \left(
F_{X}\left( x^{2}\right) -F_{X}\left( x^{1}\right) \right) ^{2}\leq K\left(
x^{2}-x^{1}\right) ^{2}\sup_{x\in \left( x^{1},x^{2}\right) }f_{X}^{2}\left(
x\right) \text{.}
\end{equation*}

To complete the proof we need to examine the first term of $\left( \ref%
{lem_2}\right) $, whose first moment is 
\begin{eqnarray*}
\frac{K}{n^{2}}\sum_{i=1}^{n}E\left( \left( v_{i}^{2}-E\left( v_{i}^{2}\mid 
\mathcal{G}_{i}\right) \right) \mathcal{I}_{i}\left( x^{1},x^{2}\right)
\right) ^{2} &\leq &\frac{K}{n^{2}}\sum_{i=1}^{n}E\left( v_{i}^{4}\mathcal{I}%
_{i}\left( x^{1},x^{2}\right) \right) \\
&\leq &\frac{K}{n}\left( F_{X}\left( x^{2}\right) -F_{X}\left( x^{1}\right)
\right)
\end{eqnarray*}%
by standard arguments as $Ev_{i}^{4}<K$ proceeding as we did with the second
term of $\left( \ref{lem_2}\right) $.\newpage
\end{proof}

\end{document}